\documentclass[12pt]{article}
\usepackage{amssymb,amsmath}
\usepackage{graphicx,epsfig}
\usepackage{caption}
\usepackage{subcaption}
\usepackage{enumerate}
\usepackage{epstopdf}
\usepackage{ mathrsfs }
\usepackage{xcolor}
\usepackage{float}
\usepackage{booktabs}
\usepackage{rotating}
\usepackage{tikz}
\usepackage[section]{placeins}
\usepackage{textcomp}
\usepackage[normalem]{ulem}
\usepackage{derivative}
\usepackage{amsmath}

\usepackage{dsfont}
\usepackage{amsthm}
\usepackage{comment}
\usepackage{tabularx}
\usepackage{hyperref}
\usepackage{upgreek}
\usepackage[margin=3cm]{geometry}
\usepackage[english]{babel}
\usepackage[square,numbers]{natbib}
\usepackage[linesnumbered,lined,boxed,commentsnumbered]{algorithm2e}
\RestyleAlgo{ruled}
\SetKwComment{Comment}{/* }{ */}
\hypersetup{
    colorlinks=true,
    linkcolor=cyan,
    filecolor=cyan,      
    urlcolor=cyan,
    citecolor = cyan
    }
\graphicspath{{Rips sequences/}}
 \numberwithin{equation}{section}

\newtheorem{corollary}{Corollary}

\newtheorem{t1}{Theorem}[section]
\newtheorem{d1}{Definition}[section]

\newtheorem{l1}{Lemma}[section]
\newtheorem{r1}{Remark}[section]
\newtheorem{e1}{Example}[section]

\newtheorem{a1}{Assumption}[section]

\begin{document}
\title{L-Estimation Approach to Tobit Models with Endogeneity and Weakly Dependent Errors}
\author{
	Swati Shukla$^1$, Subhra Sankar Dhar$^2$, Shalabh$^3$\\
	Department 
 of Mathematics and Statistics, Indian Institute of Technology Kanpur \\
	 Kanpur 208016, India\\
	 Emails: $^1$shukla@iitk.ac.in , $^2$subhra@iitk.ac.in,  $^3$ shalab@iitk.ac.in}

\maketitle
\begin{abstract}
        This article introduces an L-estimator for the semiparametric Tobit model with endogenous regressors. The estimation procedure follows a two-stage approach: the first stage employs least squares, while the second stage utilizes the L-estimation technique. We establish the large-sample properties of the proposed estimators under weakly dependent data. The utility of the proposed methodology is demonstrated for various simulated data and a benchmark real data set.\\
    
\end{abstract}

\section{Introduction and Literature Review }
We consider the problem of estimation of parameters associated with the instrumental variables regression (see, e.g., \cite{MR4725141}, \cite{MR4712657}, \cite{MR4393477}, \cite{MR4294546},  \cite{Smith}, \cite{Roger1984}, \cite{St2018}, \cite{Dharsh}, \cite{Dharsh2023}, \cite{Kim2011}, \cite{Simon2012}, and \cite{Camponovo2015}) for censored data under the violation of the regularity assumption that regressors are uncorrelated with unobserved random errors in the model. In the literature, this situation is described as the presence of endogeneity in the model, as the regressors correlated with errors are called endogenous. For the sake of simplicity in exposition, here we consider the estimation of the parameters associated with the aforesaid model under the presence of one endogenous regressor. Till now, various estimation procedures have been proposed to estimate the parameters associated with various regression models for censored data under the presence of endogeneity in the model (see, e.g., \cite{MR4066064}, \cite{MR3449070}, \cite{BLUNDELL200765}, \cite{MR3928467}, \cite{CHERNOZHUKOV2015201}, \cite{HongTamer}, \cite{KHAN2009104}, \cite{HONORE2004293} and \cite{CHEN201830}).
 The main limitation of the method as mentioned earlier is their reliance on moment conditions, which lead to some serious model identifiability issues. Additionally, these methods impose conditions on the error distribution, requiring it to have finite moments, thereby excluding distributions like the Cauchy distribution. Therefore, most of these methods may not perform well due to their stringent theoretical assumptions, which are challenging to meet in practice. In contrast, the theoretical requirements of our proposed estimator are less cumbersome than those of the existing methodologies. Moreover, our estimator accommodates error distributions with non-finite moments and performs well in these conditions. Furthermore, the proposed class of L-estimators consists of several estimators in the context of censored regression. This class of L-estimator is also robust to the large number of censored observations and offers an advantage in practice.
 
 The standard procedure to deal with the endogeneity is to adopt a two-stage estimation procedure in which one takes suitable measures to remove the endogeneity from the model in the first step and subsequently applies a suitable estimation procedure (see, e.g., \cite{POWELL1984}, \cite{Winsorized}, \cite{bradic2019} and \cite{shukla2023}). In this work, we also adopt a two-step estimation procedure by using the instrumental variables regression approach, which is often referred to as the control function approach (see, e.g., \cite{Smith}) to remove the endogeneity, and then propose a class of smooth L-estimators taking advantage of the continuity of the censored quantile regression coefficient process. The proposed class of estimators that are analogous to L-estimators for linear regression models, as L-estimators are well-known for their robustness property when the data are generated from a heavy-tailed distribution and/or the data consist of some outliers (see, e.g., \cite{MR0829458}).  Furthermore, we establish the consistency and asymptotic normality of the proposed L-estimators and study the performance of the proposed methodology on simulated and real data sets. With notations, the problem is briefly described in the following. 

Let $(y_{1}, \boldsymbol{\Tilde{x}_{1}}, w_{1})^\top,\ldots,(y_{n}, \boldsymbol{\Tilde{x}_{n}}, w_{n})^\top$ be $n$ identically distributed observations on the random vector $(Y, \boldsymbol{\Tilde{X}}, W)^\top$, where $Y$ denotes the response variable, $\boldsymbol{\Tilde{X}}$ denotes the vector of $p$  exogenous regressors with the first component being 1, and $W$ denotes an endogenous regressor. Note that for the sake of exposition, we have taken only one endogenous regressor; however, in principle, one may take more than one endogenous regressor as well. We now wish to fit the following linear regression model to the observed censored data:   
\begin{equation}\label{eq:1.1} 
y_{i}=\max~(c,y_{i}^{*})= \max~\{c,\boldsymbol{\Tilde{x}^\top_{i}} \boldsymbol{\alpha}_{0}+w_{i}\gamma_{0} +\epsilon_{i}\},\hspace{0.4cm} i = 1,\ldots,n,
\end{equation}
where $\boldsymbol{\alpha}_{0} \in \mathbb{R}^{p},\; p\geq 1$ is the vector of unknown parameters associated with the exogenous regressors, including the intercept term, $\gamma_{0} \in \mathbb{R}$ is an unknown parameter associated with the endogenous regressor, $\epsilon_{i}$ is an unobserved error term, and $c$ denotes the censoring threshold, which is assumed to be non-random and constant for all the observations. Here the response variable $y_{i}$ is fully observed, but the latent
variable $y_{i}^{*}$ is not fully observed.  A particular case of this model is the model with the censoring threshold $c = 0$, such a model is referred to as a Tobit model. In what follows, we restrict our study to Tobit models, since results can be relayed back to the model with censoring threshold $c$ just by transforming $y$ to $y - c$ for a known threshold $c$. Therefore, in this chapter, we are concerned with the following model:       
\begin{equation}\label{eq:1.2}
y_{i} =  \max  \{0, \boldsymbol{\Tilde{x}^\top_{i}} \boldsymbol{\alpha}_{0}+w_{i}\gamma_{0} +\epsilon_{i} \}, \hspace{0.1cm} i = 1,\ldots,n.
\end{equation}
The model above in \eqref{eq:1.2} is an instance of a linear regression model for censored data with left censoring at a non-random and constant censoring threshold $c = 0$.  We assume errors $\epsilon_{i}$'s to be identically distributed with some unknown distribution (denoted by $F$ later on), and possibly dependent random variables. 
This model is often called the semi-parametric censored regression model since we assume the linearity of the regression function but do not assume any parametric form of the error distribution $F$. 

 In the context of statistical inference, many estimation procedures are available for the estimation of the unknown parameters associated with the model \eqref{eq:1.2} (see, e.g., \cite{POWELL1984}, \cite{Winsorized}, \cite{Powell1986}, \cite{SCLS} and \cite{shukla2023}). However, in the presence of endogeneity, parameter estimates from these conventional procedures do not have optimal properties. In such a situation, various methods have been proposed (see, e.g., \cite{BLUNDELL200765}, \cite{CHERNOZHUKOV2015201}, \cite{HongTamer}, \cite{KHAN2009104}, \cite{HONORE2004293}, \cite{CHEN201830}) to deal with the endogeneity such that the parameter estimates have optimal properties. One such approach proposed by \cite{Smith} uses the 
instrumental variables regression approach or control function approach to remove the endogeneity from the model. Subsequently, any suitable estimation procedure can be applied to estimate the parameters as if the regularity condition of uncorrelated errors is not violated. 
\subsection{Background on L-Functionals}
In this chapter, we adopt a two-stage estimation procedure. In the first stage, we employ the instrumental variables regression to account for the endogeneity in the model, and in the second stage, we define a class of estimators analogous to L-estimators in linear regression models. Given a response variable $Y$ and a vector of regressors $\boldsymbol{X}$, the class of L-estimators for the linear regression model is defined as follows: 

Let $F_{Y|\boldsymbol{X}}$ denote the conditional distribution of $Y|\boldsymbol{X}$, and the L-functional for $F_{Y|\boldsymbol{X}}$ is denoted as $T(F_{Y|\boldsymbol{X}})$, which is defined as 
\begin{equation} \label{Lf}
     \mathscr{L}(\boldsymbol{x}) := T(F_{Y|\boldsymbol{X} = \boldsymbol{x}}) = \int^1_0 \mathcal{Q}_{Y|\boldsymbol{X}}(\uptau) d \mu(\uptau),
\end{equation}
where $\mathcal{Q}_{Y|\boldsymbol{X}}(\uptau) := \inf \{y \in \mathcal{S} : F_{Y|\boldsymbol{X}}(y) \geq \uptau \}$ is the conditional quantile of $F_{Y|\boldsymbol{X}}$, $\mathcal{S}$ denotes the support of $F_{Y|\boldsymbol{X}}$, and $\mu$ is a generic weight function that corresponds to a signed measure. The choice of $\mu$ can be any signed measure for functionals such as the measure of scale, skewness, and kurtosis (heavy-tailedness) of the response, conditioned on covariates (see,  \cite{L_functional}). However, when the functional of interest (or parameter) corresponds to a measure of location such as the population mean or quantile of the response, conditioned on covariates, the weight function is chosen to be a positive measure on $[0, 1]$ (see,\cite{L_functional}). 

Further, a measure of location is equivariant to the linear transformation of the underlying random variable, i.e., for a random variable $Z$ having a distribution $F_{Z}$, and a functional $T_{loc}$ that corresponds to a measure of location, we have 
$$T_{loc}(F_{aZ + b}) = a T_{loc}(F_{Z}) + b, \text{ for any } a, b \in \mathbb{R}.$$
Since the conditional quantiles $\mathcal{Q}_{Y|\boldsymbol{X}}$ are equivariant to linear transformation of $Y$ that is $\mathcal{Q}_{aY + b|\boldsymbol{X}} = a\mathcal{Q}_{Y|\boldsymbol{X}} + b$ for any $a, b \in \mathbb{R}$, this implies that the L-functionals $\mathscr{L} (.)$ (see \eqref{Lf}) are equivariant if 
$$ \int^1_0 d \mu(\uptau) = 1.$$
Therefore, we are mainly concerned with the weight functions $\mu$ that are probability measures on $[0, 1]$ that split as follows:
\begin{equation} \label{wf}
    \mu(B) = \alpha\lambda(B) + (1 - \alpha)\nu(B), \text{ for every Borel set } B \in \mathscr{B}, \alpha \in [0, 1], 
\end{equation} 
where $\mathscr{B}$ denotes Borel sigma algebra on $[0, 1]$, and $\lambda$ and $\nu$ denotes an absolutely continuous and discrete measure on $[0, 1]$, respectively. 

Thus, the functional $\mathscr{L}(\boldsymbol{x})$ encompasses various measures of locations for different choices of the weight function $\mu$, hence forming a class of measures of location. A few examples of measures of location for different choices of weight function are in order. For a given $\uptau \in [0, 1]$, the population conditional quantile $\mathcal{Q}_{Y|\boldsymbol{X}}(\uptau)$ can be derived from $\mathscr{L}(\boldsymbol{x})$ by choosing the weight function $\mu$ to be a single point mass distribution at $\uptau$. Take $\alpha = 0$, and $\nu$ be a single point mass distribution at $\uptau$ in \eqref{wf}. The midrange of $F_{Y|\boldsymbol{X}}$, assuming that $\mathcal{S}$ is bounded, can be obtained by taking $\alpha = 0$, and $\nu$ to be Bernoulli (1/2). This yields the midrange $[\mathcal{Q}_{Y|\boldsymbol{X}}(0) + \mathcal{Q}_{Y|\boldsymbol{X}}(1)]/2$. Some other examples of location functionals with absolutely continuous choice of the weight function $\mu$ include the so-called ``smoothed conditional quantiles''. The smoothed conditional quantiles of $F_{Y|\boldsymbol{X}}$ can be derived from \eqref{wf} by taking $\alpha = 1$, and choosing $\lambda$ to be an absolutely continuous distribution. Following is an example of a smooth quantile: 
\begin{equation} \label{sq}
    T_{SCQ} (F_{Y|\boldsymbol{X}}) = \int^1_0  \mathcal{Q}_{Y|\boldsymbol{X}}(\uptau) \frac{1}{h(\pi)} K\left( \frac{\uptau - \pi}{h(\pi)}\right) d\uptau,
\end{equation}  
where $K$ is a smoothing probability density with the bandwidth $h(\pi)$ depending on $\pi \in [0, 1]$ such that $ 0 < h(\pi) < \min (\pi, 1-\pi)$. The conditional trimmed mean of the response variable can be derived from \eqref{sq} by the following choice of smoothing density $K$, and the bandwidth $h(\pi)$:
\begin{equation} \label{tm}
    K(u) = \frac{1}{2} \mathds{1}(|u| \leq 1), u \in [0, 1], \pi = \frac{1}{2}(\pi_{1} + \pi_{2}), h(\pi) = \frac{1}{2} (\pi_{2} - \pi_{1}), \text{ such that } \pi_{1} + \pi_{2} = 1.
\end{equation}
Now, using \eqref{tm} along with \eqref{sq} yields the following conditional trimmed mean functional:
\begin{equation} \label{cm}
    T_{CTM}(F_{Y|\boldsymbol{X}}) = \frac{1}{\pi_{2} - \pi_{1}}\int^{\pi_{2}}_{\pi_{1}} \mathcal{Q}_{Y|\boldsymbol{X}}(\uptau) d\uptau.   
\end{equation}
Note that the conditional mean $\mathbb{E}(Y|\boldsymbol{X})$ is a particular case of $T_{CTM}(F_{Y|\boldsymbol{X}})$ for $\pi_{1} = 0, \pi_{2} = 1$. For other kinds of trimmed means in the general non-parametric regression model, the readers may look at Dhar et al., \cite{MR4511147} and a few references therein.

 We now want to discuss how to estimate the functional $\mathscr{L}(\boldsymbol{x})$ defined in \eqref{Lf} for some suitable choice of the weight function. One can estimate the functional $\mathscr{L}(\boldsymbol{x})$ by plugging in an estimator of the conditional quantile $\mathcal{Q}_{Y|\boldsymbol{X}}(\uptau)$ in \eqref{Lf}, and observe that the conditional quantile $\mathcal{Q}_{Y|\boldsymbol{X}}(\uptau)$ can in turn be estimated by plugging in an estimate of $F_{Y|\boldsymbol{X}}$. For a given $n$ observations on the response and regressors, usually one takes the empirical distribution function $F_{n}$ (see, e.g., \cite{Serfling}, \cite{Koenker1987}) as an estimate of $F_{Y|\boldsymbol{X}}$. In this case, an estimate of $\mathcal{Q}_{Y|\boldsymbol{X}}(\uptau)$ is given by:
$$\Hat{\mathcal{Q}}_{Y|\boldsymbol{X}}(\uptau) = \inf\{ y \in \mathcal{S} : F_{n}(y) \geq \uptau\}.$$
Note that for the regression model without covariates, $\Hat{\mathcal{Q}}_{Y|\boldsymbol{X}}(\uptau)$ simply reduces to an unconditional quantile $\Hat{\mathcal{Q}}_{Y}(\uptau)$ that can be written in terms of order statistics $Y_{(i)}, i = 1, \ldots, n$, as follows:
$$\Hat{\mathcal{Q}}_{Y}(\uptau) = Y_{(i)} \quad if \quad \frac{i-1}{n} < \uptau \leq \frac{i}{n}, \text{ for all } i = 1, \ldots, n.$$
Thus, for a regression model without covariates, the estimate of L-functional reduces to the following linear combinations of order statistics:
$$ \Hat{\mathscr{L}}_{n} = \sum _{i = 1}^n Y_{(i)} w_{i}, \quad w_{i} = \int_{(i-1 )/ n} ^{i/n} d\mu(\uptau).$$
However, for the linear regression models with covariates, we are concerned with the following L-functional and its estimation:
\begin{equation} \label{LE}
     \mathscr{L}_{0} = \int_0^1 \boldsymbol{\beta}_{0}(\uptau) d\mu(\uptau),
\end{equation}
where $ \boldsymbol{\beta}_{0}(\uptau)$ is the unknown true regression quantile for a fix, $\uptau \in [0, 1]$ which is defined as:
\begin{equation}
   \boldsymbol{\beta}_{0}(\uptau) = \underset{\beta \in \mathscr{B}}{\arg\min}\;\mathbb{E} \left(\rho_{\uptau}\left(Y-\boldsymbol{x}^\top \boldsymbol{\beta}\right)\right).
\end{equation}
Here $\mathscr{B}$ is the parameter space of $\boldsymbol{\beta},$ $\rho_{\uptau}(u)=(\uptau-\mathds{1}(u\leq 0))u$, $\mathbb{E}$ denotes the conditional expectation corresponding to $F_{Y|\boldsymbol{X}}$. For details on the L-estimation in the linear regression model (see, e.g., 
  \cite{Koenker1987}, \cite{Bickel} and \cite{Gutenbrunner}) and a few relevant references therein. Finally, to estimate $\mathscr{L}_{0}$, one may consider the following type of estimator of $\mathscr{L}_{0}$ defined in \eqref{LE}:
\begin{equation}
    \Hat{\mathscr{L}}_{n} = \int_0^1 \Hat{\boldsymbol{\beta}}_{n}(\uptau) d\mu(\uptau),
\end{equation}
where $\Hat{\boldsymbol{\beta}}_{n}(\uptau)$ is an estimator of $\boldsymbol{\beta}_{0}(\uptau)$ defined as:
\begin{equation}
     \Hat{\boldsymbol{\beta}}_{n}(\uptau) \in \underset{\beta \in \mathscr{B}}{\arg\min}\;\mathbb{E}_{n} \left(\rho_{\uptau}\left(Y-\boldsymbol{x}^\top \boldsymbol{\beta}\right)\right).
\end{equation}
Here, $\mathbb{E}_{n}$ denotes the expectation corresponding to the empirical distribution $F_{n}$ of $F_{Y|\boldsymbol{X}}$. In general, the readers may look at \cite{MR4550234} for details on regression quantiles associated with general regression models like time-varying regression models, and for classical regression quantiles, the readers may refer to \cite{MR0474644}, \cite{MR1105843}, \cite{MR2760140}, \cite{MR2420417}, \cite{MR0640165}, and the relevant references therein. 
\subsection{Our contribution}
 In this chapter, we propose an L-estimator for the Tobit model under endogeneity with $\alpha$-mixing (weakly dependent) data. To achieve this, we develop a two-stage estimation procedure for the Tobit regression model and investigate both its estimation and inference properties. The asymptotic behavior of the joint estimator is examined as a function of the quantile index. In the second stage, the regression is carried out using L-estimation, which is constructed from the quantile process index. Since the regression function involves an indicator function, the resulting objective function is non-differentiable and non-convex with respect to the unknown parameters. To address this challenge, we employ empirical process methods for dependent data and establish the consistency and asymptotic normality of the joint estimator. Specifically, the first stage is based on least squares estimation, while the second stage applies L-estimation through a Tobit quantile regression process. Furthermore, recognizing that L-estimators are a functional of the regression quantiles, we extend these results to derive the consistency and asymptotic normality of the proposed L-estimator under suitable regularity conditions. We also provide a consistent estimator of the covariance matrix. Finally, the finite-sample performance of the proposed estimators is assessed through comprehensive simulation studies and real data applications.
\subsection{ Challenges}
 Before closing this section, it may be an appropriate place to mention that the model described in \eqref{eq:1.2} is entirely different from the usual multiple linear regression model, and hence, the mathematical treatments of the L-estimator considered here (see Section \ref{L}) is entirely different too than the usual L-estimators of the unknown parameters involved in the usual regression model. An advanced technique of empirical process is used to deal with the complications involved in model \eqref{eq:1.2} and to investigate the asymptotic properties of the proposed estimator. For sake of better presentation, the mathematical challenges in implementing the concept of L-estimators of the unknown parameters associated with \eqref{eq:1.2} are thoroughly discussed in the first two paragraphs in Section \ref{MR}. 
\subsection{Outline of the Paper}
The rest of the chapter is organized as follows: The proposed two-step estimation procedure is articulated in Section \ref{Methodology}. Section \ref{Est} describes the control function approach used to address the endogeneity issue in the linear regression model for censored data, and Section \ref{L} extends the formulation of L-functional and its estimation procedure for the same model. The asymptotic results along with the regularity conditions are articulated in Section \ref{MR}. Simulation studies are presented in Section \ref{FSSS} to demonstrate the consistency of estimates of L-functionals, and benchmark real data is analyzed in Section \ref{RDA}. We conclude the chapter with concluding remarks in Section \ref{Conclusion}. Finally, proofs of the main theoretical results are given in the Appendix \ref{appendix}.   

\textit{ The R codes of simulated and real data studies are available at \url{https://github.com/swati-1602/CRML.git}.}

{\bf{Notations:}} 

We adhere to \cite{vanderVaart1996}'s notation for empirical processes.
\begin{itemize}
    \item For $\mathbb{W} =(y,x,w,e)$,  $\mathbb{E}_{n}[f(\mathbb{W})]= \frac{1}{n}\sum\limits_{i=1}^{n}f(\mathbb{W}_{i})$ is the empirical average at $f$, and $\mathbb{G}_{n}[f(\mathbb{W})]= n^{-1/2}\sum\limits_{i=1}^{n}(f(\mathbb{W}_{i})-\mathbb{E}[f(\mathbb{W}_{i})])$ is the empirical process evaluated at same function $f.$  Here $W_1, \ldots, W_n$ are i.i.d. copies of $W$.
    \item If $\Hat{f}$ is an estimated function, $\mathbb{G}_{n}[\hat{f}(\mathbb{W})]= n^{-1/2}\sum\limits_{i=1}^{n}(f(\mathbb{W}_{i})-\mathbb{E}[f(\mathbb{W}_{i})])_{f=\Hat{f}}.$
    \item ``$\Rightarrow$" represents convergence in distribution, and  $l^{\infty}(\mathcal{T})$ denotes set of all uniformly bounded real functions on $\mathcal{T}.$
    \item For any $\{X_{n},n\geq 1\}$ and $\{Y_{n},n\geq 1\},$
    we define the following:
    \begin{itemize}
        \item[$\bullet$]  $X_{n}=o_{p}(Y_{n})$ if $\frac{X_{n}}{Y_{n}}\overset{p}{\to}0.$
          \item[$\bullet$] $X_{n}=O_{p}(Y_{n})$ if $\forall \;\epsilon^{*}>0,$ $\exists\; m_{0}\;\&\; n_{0}\; \in N$ such that $p\left[|\frac{X_{n}}{Y_{n}}|\leq m_{0}\right]\leq 1-\epsilon^{*}\; \forall n\geq n_{0}.$
    \end{itemize}
   \item ``.'' represents the product.
\end{itemize}

\section{Methodology} \label{Methodology}
We adopt a two-stage estimation procedure to estimate the parameters of the model described in \eqref{eq:1.2}. In the first stage, we employ the instrumental variable regression approach to remove the effect of endogenous variables on the parameter estimates. Subsequently, in the second stage, we define a class of estimators analogous to L-estimators in linear regression models associated with censored data.

 \subsection{The IV Estimation Technique}\label{Est}
The instrumental variable (IV) estimation technique is the standard method for estimating a regression model's parameters when one or more endogenous regressors are present. This approach uses instruments, or instrumental variables, to remove endogeneity at the first stage. Subsequently, the parameters involved in the model are estimated in the next stage using an appropriate estimation technique. It is imperative to assume the validity of the instruments while using the instrumental variable approach. If an instrument is both exogenous and correlated with the endogenous variables, it is considered a valid instrument. In these cases, endogeneity is removed from the model using instrumental variables based on the instrumental variables technique (see, e.g., \cite{Smith}, \cite{NEWEY19}, and  \cite{RIVERS1988347}). 

Let $Z$ be an instrumental variable for the endogenous variable W (see model \eqref{eq:1.2}), and  $z_{1,i}$ denote the $i^{th}$ value of the instrumental variable corresponding to the endogenous variable $w_{i}$ in the model \eqref{eq:1.2}, for $i = 1,\ldots, n$. The instrumental variables technique proceeds along the following steps. In the first step, one 
regresses the endogenous variable $w_{i}$ on the instrumental variable $z_{1,i}$ and exogenous vectors $(\Tilde{\boldsymbol{x}}_{i}^\top)$ using the least squares methodology, i.e.,
\begin{equation} \label{eq:2.1}
   w_{i}= \boldsymbol{z}^\top_{i}\boldsymbol{\delta}_{0} + \vartheta_{i},\quad i= 1, \ldots, n,
\end{equation}
where $\boldsymbol{z}_{i}=(z_{1,i},\boldsymbol{\Tilde{x}}^\top_{i})^\top$ represents the $i$-th observation on a $p+1$-dimensional vector of instruments and $\boldsymbol{\delta}_{0}=(\delta_{1,0},\boldsymbol{\delta}_{2,0})^\top\in \mathbb{R}^{p+1}$ is the unknown parameter vector associated with the instrumental vector $\boldsymbol{z}$, and $\vartheta_{i}$ is the random error. For identification, it is assumed that there is at least one component of $\boldsymbol{z}$ that is not included in $\boldsymbol{\Tilde{x}}_{i}$, and that there is at least one non-zero coefficient for the excluded components of $\boldsymbol{z}$. The first-stage estimator $\hat{\boldsymbol{\delta}}_{n}$ is obtained by minimizing the following objective function:
\begin{equation}\label{eq:2.2}
    \hat{\boldsymbol{\delta}}_{n}=\underset{\boldsymbol{\delta} \in \Delta}{\arg\min} T_{n}(\boldsymbol{\delta}),
\end{equation}
where $T_{n}(\boldsymbol{\delta})=\frac{1}{n}\sum_{i=1}^{n}(w_{i}-\boldsymbol{z}^\top_{i}\boldsymbol{\delta}_{0})^2$ and $\Delta \subset \mathbb{R}^{p+1}$ is a parameter space.
Let $\mathcal{F}_{i}$ be the smallest sigma field generated by $\{(\boldsymbol{\Tilde{x}}^\top_{j+1},w_{j},\boldsymbol{z}^\top_{j+1},\epsilon_{j}):0\leq j\leq i\leq n\}=\sigma(\boldsymbol{\Tilde{x}}_{i},\boldsymbol{z}_{i},w_{i-1},\epsilon_{i-1},\ldots).$
Now, we assume the following conditions:
\begin{a1}\label{a:2.1} For all $i=1,\ldots, n,$ $\mathbb{E}[\vartheta_{i}|\mathcal{F}_{i-1}]=0$ holds almost surely.
\end{a1}
\begin{a1}\label{a:2.2} For all $i=1,\ldots, n,$ $$Q_{\epsilon_{i}|\mathcal{F}_{i-1},w_{i}}(\uptau)=Q_{\epsilon_{i}|\vartheta_{i}}(\uptau)=\rho_{10}\vartheta_{i}+Q_{\varepsilon_{i}}(\uptau)$$ holds almost surely. In this context,  $Q_{\epsilon_{i}|\mathcal{F}_{i-1},w_{i}}(\uptau)$ denotes the $\uptau$-th conditional quantile function of $\epsilon_{i},$ given, $\mathcal{F}_{i-1}$ and $w_{i},$ the other expressions are defined analogously. In addition, $\rho_{10}$ is an unknown parameter, and the error term $\varepsilon_{i}$ is the same as defined in \eqref{eq:2.2}.
\end{a1}
 Assumption~\ref{a:2.1} imposes mean independence, implying that $\vartheta_{i}$ is uncorrelated to $\boldsymbol{\Tilde{x}}_{i-1}, \boldsymbol{z}_{i-1}$, $w_{i-2},\epsilon_{i-2}$ and their past values. In addition, Assumption~\ref{a:2.2} states that the control function is linearly related to $\vartheta_{i}.$
The variable $w_{i}$ is endogenous due to its contemporaneous correlation with the error term 
$\epsilon_{i}$ in model \eqref{eq:1.2}, and $\vartheta_{i}$, it is represented by
\begin{equation}\label{eq:2.3}
    \epsilon_{i}=\rho_{10}\vartheta_{i}+\varepsilon_{i},\; i=1,\ldots, n,
\end{equation}
     where $\boldsymbol{\varepsilon}_{i}$ is the random error and $\rho_{10}$ is an unknown parameter in \eqref{eq:2.3}. 
The essence of the aforesaid approach is to incorporate the control function in the original model (see \eqref{eq:1.2}) to account for endogeneity and use $\boldsymbol{\varepsilon}_{i}$ as the error term instead of $\epsilon_{i}$. Consequently, the following transformed model will be free from endogeneity that can be estimated using any suitable technique.
\begin{equation}\label{eq:2.4}
y_{i} =  \max \{0, \boldsymbol{\Tilde{x}^\top_{i}} \boldsymbol{\alpha}_{0}+w_{i}\gamma_{0} + \rho_{10} \vartheta_{i} +\varepsilon_{i}\}= \max\{0, \boldsymbol{x}^\top_{i}\boldsymbol{\beta}_{0}+\varepsilon_{i}\},\quad i= 1, \ldots, n.
          \end{equation}
Here $\boldsymbol{x}_{i}=[\boldsymbol{\Tilde{x}^\top_{i}},w_{i},\vartheta_{i}]^\top\in \mathbb{R}^{p+2}$ denotes the vector of regressors, and $\boldsymbol{\beta}_{0} = [\boldsymbol{\alpha}_{0},\gamma_{0},\rho_{10}]^\top$ s the corresponding 
$(p+2)$-dimensional vector of unknown parameters. Note that the model described in  \eqref{eq:2.4} contains an unobserved random component $(\vartheta_{i})$ as one of the regressors. Therefore, it is necessary to replace $\vartheta_{i}$ by the residuals $e_{i}$ of the model in \eqref{eq:2.1} to estimate the parameters of the model \eqref{eq:2.4}. Now, let $e_{1}, \ldots, e_{n} $ denote the residuals of the model in \eqref{eq:2.1}, where $e_{i}=e_{i}(\hat{\boldsymbol{\delta}}_{n}) = w_{i}-\hat{w}_{i}=w_{i}-\boldsymbol{z}_{i}^\top\hat{\boldsymbol{\delta}}_{n}.$ Then the model in \eqref{eq:2.4} reduces to the following model:
\begin{equation}\label{eq:2.5}
              y_{i} =  \max \{0, \boldsymbol{\Tilde{x}^\top_{i}} \boldsymbol{\alpha}_{0}+w_{i}\gamma_{0}+e_{i}\rho_{10}+\eta_{i}\} =  \max\{0, \hat{\boldsymbol{x}}^\top_{i}\boldsymbol{\beta}_{0}+\eta_{i}\}.
          \end{equation}
In this case, $\hat{\boldsymbol{x}}_{i}=[\boldsymbol{\Tilde{x}^\top_{i}},w_{i},e_{i}]^\top$ and $\boldsymbol{\beta}_{0} = [\boldsymbol{\alpha}_{0},\gamma_{0},\rho_{10}]^\top$ is the $((p+2) \times 1)$-dimensional vector of unknown parameters and $\eta_{i}=\varepsilon_{i}-\rho_{10}(e_{i}-\vartheta_{i})$ is a new error term. Next, the estimation of the parameters in the model \eqref{eq:2.5} is described.

\subsection{The Class of L-estimators}\label{L} 
Given a vector of regressors $\boldsymbol{x}$ and the response variable $Y$, the conditional quantile function corresponding to the model \eqref{eq:1.2} is represented by $\mathcal{Q}_{Y|\boldsymbol{x}}$ and defined as follows: 
$$\mathcal{Q}_{Y|\boldsymbol{x}}(\uptau) = F^{-1}_{Y|\boldsymbol{x}}(\uptau) = \inf\{{t\in \mathbb{R}: F_{Y|\boldsymbol{x}}(t) \geq \uptau}\}, \; \uptau \in (0,1),$$ where $F_{Y|\boldsymbol{x}}$ denotes the cumulative conditional distribution function of $Y$ given $\boldsymbol{x}$. The class of L-estimators of the parameters involved in the model \eqref{eq:2.4} is based on the estimators studied by (Powell, \cite{Powell1986}). For that model, under the assumption of integrability, the $\uptau$-th quantile of true unknown parameter vector $\boldsymbol{\beta}_{0}(\uptau )$ solves
\begin{equation}\label{eq:2.6}
   \boldsymbol{\beta}_{0}(\uptau) = \underset{\beta \in \mathscr{B}}{\arg\min}\;\mathbb{E} \big[\rho_{\uptau}(Y-\max(0,\boldsymbol{x}^\top \boldsymbol{\beta})\big].
\end{equation}
Here 
$\boldsymbol{\beta} = (\boldsymbol{\alpha},\gamma,\rho_{1})^\top,$ $\mathscr{B}$ is the parameter space of $\boldsymbol{\beta},$ and $\rho_{\uptau}(u)=(\uptau-\mathds{1}(u\leq 0))u$. The expectation corresponding to the conditional distribution function of $Y$ given $\boldsymbol{x}$ is represented by $\mathbb{E}$. In view of the equivariance properties of the quantile function to monotone transformation, the quantile function under censoring is as follows: 
\begin{equation}\label{eq:2.7}
\begin{split}
\mathcal{Q}_{Y|\mathcal{F}_{i-1},w_{i}}(\uptau) = F^{-1}_{Y|\mathcal{F}_{i-1},w_{i}}(\uptau) &= \mathcal{Q}\big(\max\{0, \boldsymbol{\Tilde{x}}^\top\alpha_{0}+w_{i}\gamma_{0}+\epsilon_{i}\}|\mathcal{F}_{i-1}.w_{i}\big),\;\uptau \in \mathcal{T}\\
&=\max\{0,\boldsymbol{\Tilde{x}}^\top\alpha_{0}+w_{i}\gamma_{0}+\mathcal{Q}_{\epsilon_{i}|\mathcal{F}_{i-1},w_{i}}(\uptau) \}\\
&=\max\{0,\boldsymbol{\Tilde{x}}^\top\alpha_{0}+w_{i}\gamma_{0}+\mathcal{Q}_{\epsilon_{i}|\vartheta_{i}}(\uptau) \}\\
&=\max\{0,\boldsymbol{x}^\top_{i}\boldsymbol{\beta}_{0}(\uptau)\},\; \uptau \in \mathcal{T}.
\end{split}
\end{equation}
Let $\boldsymbol{\beta}_{0}(\uptau) \in \mathbb{R}^{p+2}$ denote the true parameter vector, defined as $$\boldsymbol{\beta}_{0}(\uptau)=(\boldsymbol{\alpha}_{0}(\uptau),\gamma_{0},\rho_{10})^\top,$$ where $\boldsymbol{\alpha}_{0}(\uptau)=(\alpha_{1,0}+\mathcal{Q}_{\varepsilon}(\uptau),\alpha_{2,0},\ldots,\alpha_{p,0})^\top,$ and $\mathcal{Q}_{\varepsilon}(\uptau)$ denotes the $\uptau$-th quantile function of error term $\varepsilon.$  The error term $\varepsilon_{i},$ defined in equation~\eqref{eq:2.3}, can then be rewritten as follows:
\begin{equation}\label{eq:2.8}
    \varepsilon_{\uptau,i} = y_{i}^{*}-\boldsymbol{x}^\top_{i}\boldsymbol{\beta}_{0}(\uptau) = \varepsilon_{i}-\mathcal{Q}_{\varepsilon_{i}}(\uptau).
\end{equation}
Let $\mathcal{Q}_{\varepsilon_{\uptau,i}}(\uptau)$ is the $\uptau$-th quantile of $\varepsilon_{\uptau,i}.$ Then, the censored quantile estimator $\Hat{\boldsymbol{\beta}}_{n}(\uptau)$ of  $\boldsymbol{\beta}_{0}(\uptau)$ is defined as follows:
\begin{equation}\label{eq:2.9}
     \Hat{\boldsymbol{\beta}}_{n}(\uptau)=\Hat{\boldsymbol{\beta}}_{n}(\uptau,\hat{\boldsymbol{\delta}}_{n}) \in \underset{\boldsymbol{\beta}\in \mathscr{B}}{\arg\min}~\mathscr{Q}_{n}(\boldsymbol{\beta},\hat{\boldsymbol{\delta}}_{n},\uptau),
\end{equation}
where the corresponding objective function is given by  \begin{equation}\label{eq:2.10}
\mathscr{Q}_{n}(\boldsymbol{\beta},\hat{\boldsymbol{\delta}}_{n},\uptau) =\frac{1}{n}\sum_{i=1}^{n}\rho_{\uptau}(y_{i},\hat{\boldsymbol{\delta}}_{n},\boldsymbol{\beta}) =\frac{1}{n}\sum_{i=1}^{n}\rho_{\uptau}(y_{i}-  \max \{0,\hat{\boldsymbol{x}}^\top_{i}\boldsymbol{\beta}\}). \end{equation}
and the quantile loss function is defined as $\rho_{\uptau}(u)=(\uptau -\mathds{1}\{u\leq 0\})u.$
Furthermore, in the context of the model described in \eqref{eq:2.4}, we define the censored quantile process as follows:
\begin{equation} \label{eq:2.11}
 \Hat{\boldsymbol{\beta}}_{n}(\uptau)=\{\Hat{\boldsymbol{\beta}}_{n}(\uptau),\uptau\in \mathcal{T}\subset(0,1)\}.
\end{equation}
Here, $\mathcal{T}$ is a closed subinterval of $(0,1)$, denoted as $[\uptau_{0},1-\uptau_{0}]$, for some $\uptau_{0}>0$. A wide class of estimators called as smooth L-estimators, denoted by  
\begin{equation}\label{eq:2.12}
\mathscr{L}_{n}=\int_{0}^{1}\hat{\boldsymbol{\beta}}_{n}(\uptau)J_{1}(\uptau)\;d\uptau, 
\end{equation} where $J_{1}(\uptau)$ is a bounded measurable on $[0, 1]$ and vanishes outside of the subinterval $\mathcal{T}.$  The corresponding population version, denoted by $\mathscr{L}_{0}$, is
\begin{equation}\label{eq:2.13}
\mathscr{L}_{0}=\int_{0}^{1}\boldsymbol{\beta}_{0}(\uptau)J_{1}(\uptau)\;d\uptau. 
\end{equation}  Observe that various choices of $J_{1}(.)$ in \eqref{eq:2.12} and \eqref{eq:2.13} will provide various sample and population versions of L-estimators, which are later on considered in simulation studies (see in Section \ref{FSSS}).  Section \ref{MR} investigates the asymptotic properties of $\mathscr{L}_{n}$ after appropriate normalization.

\section{Main Results}\label{MR} 
In this section, the asymptotic normality and consistency of $\mathscr{L}_{n}$, as defined in \eqref{eq:2.12}, are established. We additionally provide a consistent estimator of the asymptotic covariance matrix of estimated parameters. We present the main theoretical results, including asymptotic normality and consistency in the following subsections.

\subsection{Strong Consistency of \texorpdfstring{$\mathscr{L}_{n}$}{}}\label{sc}
We establish the consistency of $\mathscr{L}_{n}$ under the following assumptions on the parameter space, errors, regressors, and the true regression function. 
 As per the requirements in understanding the technical assumptions, we here briefly discuss the notion of mixing dependent random variables, and the notations used in this description are self-contained.

 Suppose that $\{X_{n}\}_{n\geq 1}$ is a sequence of identically distributed random variables, and for any $a\in\mathbb{N}$ and $b\in\mathbb{N}$ such that $1\leq a\leq b <\infty$, let us denote ${\cal{F}}_{a}^{b} = \sigma(X_{a}, X_{a + 1}, \ldots, X_{b})$ and ${\cal{F}}_{m}^{\infty} = \sigma(X_{m}, X_{m + 1}, \ldots)$, where $\sigma(.)$ denotes the smallest $\sigma$-field generated by the random variables mentioned inside $(.).$ Note that if ${\cal{F}}_{1}^{k}$ and ${\cal{F}}_{k + m}^{\infty}$ are independent, then $P(A\cap B) - P(A) P (B) = 0$ for any $A\in {\cal{F}}_{1}^{k}$ and $B\in {\cal{F}}_{k + m}^{\infty}$. Let $\{(\boldsymbol{x}_{i}^\top,w_{i},\boldsymbol{z}_{i},\epsilon_{i})\}_{i\geq 1}$ denote a strict stationary sequence of random variables defined on a probability space $(\Omega, \mathcal{F},P)$. Define the sigma fields,
$\mathcal{F}_{-\infty}^{0} = \sigma\{(\boldsymbol{x}^\top_{i},w_{i},\boldsymbol{z}^\top_{i},\epsilon_{i}):i\leq 0\}$ and $\mathcal{F}_{m}^{\infty} = \sigma\{(\boldsymbol{x}^\top_{i},w_{i},\boldsymbol{z}^\top_{i},\epsilon_{i}):i\geq m\}.$ Then this sequence of random variables is said to be strongly mixing or $\alpha$-mixing (see, \cite{M.Rosenblatt}) if mixing coefficients $\alpha(m)$, defined as
\begin{eqnarray}\label{eq:3.1}
  \alpha(m) = \sup_{\substack{A \in \mathcal{F}_{-\infty}^{0}, \, B \in \mathcal{F}_{m}^{\infty}}} 
  \left| P(A \cap B) - P(A)P(B) \right|,
\end{eqnarray}
satisfy $\alpha(m) \to 0$ as $m \to \infty$. Furthermore, the mixing coefficients $\alpha(m)$ are assumed to satisfy the condition
\begin{equation}\label{eq:3.2}
\sum_{m \geq l}^{\infty} m^{l} \left(\alpha(m)\right)^{(r-2)/r} < \infty,
\end{equation}
for some constants $r > 2$ and $l > \frac{r-2}{r}$.
Now, recall the model \eqref{eq:2.4}, and all of the assumptions discussed below are based on it.

\subsection*{Assumptions:}

    \begin{a1}\label{a:3.3.1}
       The parameter space $\mathscr{B}$  defined in \eqref{eq:2.5} is a compact space, and for each $\uptau \in \mathcal{T},$ $\boldsymbol{\beta}_{0}(\uptau) \in \mathscr{B}^{o}.$ Here $\mathcal{T}$ is a closed subinterval of $(0,1),$ $\boldsymbol{\beta}_{0}(\uptau)$ is the same as defined in \eqref{eq:2.6}, and $\mathscr{B}^{o}$ is the interior of $\mathscr{B}\subset \mathbb{R}^{m}.$ Moreover, $\Delta$  defined in \eqref{eq:2.2} is a compact space with $\boldsymbol{\delta}_{0}\in \Delta^{o}.$ Here, $\boldsymbol{\delta}_{0}$ is the same as defined in \eqref{eq:2.1}, and $\Delta^{o}$ is an interior of $\Delta.$
     \end{a1}
      \smallskip
       \begin{a1}\label{a:3.3.2}
       There exists a constant $\phi_{1}<\infty$ such that $\mathbb{E}\|\boldsymbol{z}_{i}\|^{r}<\phi_{1}$ and $\mathbb{E}\|\boldsymbol{z}_{i}\vartheta_{i}\|^{r}<\phi_{1},$ where $r>2$ and same defined in \eqref{eq:3.2}. Moreover, let  $\mathbb{J}_{1,\boldsymbol{\delta}}=\mathbb{E}[\boldsymbol{z}_{i}\boldsymbol{z}_{i}^\top]>0,$ which exists and is positive definite. Also, the matrix $\mathbb{V}_{1,\boldsymbol{\delta}}=\underset{n \to \infty}{\lim} \mathrm{Var}\left(\frac{1}{\sqrt{n}}\sum\limits_{i=1}^{n}\boldsymbol{z}_{i}\vartheta_{i}\right)$ exists and is uniformly positive definite.
       \end{a1}
        \smallskip
      \begin{a1}\label{a:3.3.3}
      Let $f_{\varepsilon_{\uptau}}(.|\boldsymbol{x})$ denote the conditional probability density function of $\varepsilon_{\uptau}$ given $\boldsymbol{x},$ where $f_{\varepsilon_{\uptau}}(.|\boldsymbol{x})$ is strictly positive and bounded above. Moreover, the density $f_{\varepsilon_{\uptau}}(.|\boldsymbol{x})$ is Lipschitz continuous in $\varepsilon_{\uptau},$
        and twice continuously differentiable with uniformly bounded first and second derivatives on $\varepsilon_{\uptau} \in [-\varepsilon_{0},\varepsilon_{0}]$ for a constant $\varepsilon_{0}>0.$
     
      \end{a1}
      \smallskip
     \begin{a1}\label{a:3.3.4}Assume that
      there exists a constant $\phi_{2}$ such that $\mathbb{E}\|\boldsymbol{x}_{i}\|^{r'}<\phi_{2}<\infty, $ where $r^{'} =\max\{r, 4\},$ and $r>2,$ is given in \eqref{eq:3.2}. In addition, suppose that there exist a constant $\overline{c}>0$ such that $\underset{1\leq i\leq n}{\max}\|\boldsymbol{x}_{i}\|^{2}\leq \overline{c}(n)^{1/5}.$ Moreover, for all $\uptau \in \mathcal{T},$ define
     \begin{equation}
      \mathbb{A}(\uptau) =
      \mathbb{E}\left[\mathds{1}(\boldsymbol{x}^\top\boldsymbol{\beta}_{0}(\uptau)>0)\boldsymbol{x}\boldsymbol{x}^\top\right], 
     \end{equation}
where $\mathbb{A}(\uptau)$ is a uniformly positive definite matrix in the sense that the minimum eigenvalue of $\mathbb{A}(\uptau)$ is bounded away from zero.
     \end{a1}
     \smallskip
  \begin{a1}\label{a:3.3.5}
  The support of $J_{1}(.)$ is a compact subinterval of $(0,1)$, and $J_{1}(.)$ is a continuous function on that compact subinterval. Here $J_{1}$ is the same as defined in \eqref{eq:2.13}.
 \end{a1}
  \smallskip
\begin{r1} Assumption \ref{a:3.3.1} guarantees the existence and measurability of $\hat{\boldsymbol{\beta}}_{n}$ and $\hat{\boldsymbol{\delta}}_{n}$ as well as the uniformity of the almost sure convergence of the minimand over $\mathscr{B}$ and $\Delta$ respectively. Assumption \ref{a:3.3.2} implies the consistency of the first-stage estimator. Further, Assumption \ref{a:3.3.3} requires that the error has a uniformly bounded and continuous density function, that required for establishing the asymptotic distribution of censored quantile regression. Assumption \ref{a:3.3.4} specifies the identification condition that is needed to establish the consistency second-stage estimator. Lastly, Assumption \ref{a:3.3.5}  is common in practice, and it is the same as the condition considered in \cite{Koenker1987}.
\end{r1}

Theorem \ref{th:3.3.1} describes the uniform convergence of $\Hat{\boldsymbol{\beta}}_{n}(\uptau)$ to $\boldsymbol{\beta}_{0}(\uptau)$, which is a valuable result by its own worth.

\begin{t1}\label{th:3.3.1} For the model described in \eqref{eq:2.5}, under the assumptions \ref{a:3.3.1}--\ref{a:3.3.4}, we have $$\underset{\uptau \in \mathcal{T}}{\sup}\|\Hat{\boldsymbol{\beta}}_{n}(\uptau) - \boldsymbol{\beta}_{0}(\uptau)\|_{\mathbb{R}^{p+2}}\rightarrow 0$$ in probability as $n\rightarrow\infty$. Here  $\Hat{\boldsymbol{\beta}}_{n}(\uptau)$ is the same as defined in \eqref{eq:2.9}, and $\boldsymbol{\beta}_{0}(\uptau)$ is the same as defined in \eqref{eq:2.6}.
\end{t1}
\begin{proof}[$\textbf{Proof}$]
See the proof in Appendix-\ref{B.2}.
\end{proof}

Theorem \ref{th:3.3.2} describes the strong consistency of  $\mathscr{L}_{n}$ to $\mathscr{L}_{0}$.

\begin{t1}\label{th:3.3.2} For the model described in \eqref{eq:2.5}, under the assumptions \ref{a:3.3.1}--\ref{a:3.3.5}, we have  $$\|\mathscr{L}_{n}-\mathscr{L}_{0}\|_{\mathbb{R}^{p+2}}\rightarrow 0$$ in probability as $n\rightarrow\infty$. Here $\mathscr{L}_{n}$ is the same as defined in \eqref{eq:2.12}, and $\mathscr{L}_{0}$ is the same as defined in \eqref{eq:2.13}.
\end{t1}

\begin{proof}[$\textbf{Proof}$]
See the proof in Appendix-\ref{B.3}.
\end{proof}
\subsection{Asymptotic Normality of \texorpdfstring{$\mathscr{L}_{n}$}{}}
We use the censored quantile process $\hat{\boldsymbol{\beta}}_{n}(\tau)$ to establish the asymptotic normality of $\mathscr{L}{n}$, as defined in \eqref{eq:2.12}. Since the L-estimator is specified as a linear functional of the quantile process, the asymptotic normality of $\mathscr{L}_{n}$ is simply inherited from the properties of the censored quantile estimator.
Furthermore, the minimand $\mathscr{Q}_{n}(\boldsymbol{\beta},\hat{\boldsymbol{\delta}}_{n},\tau)$ in \eqref{eq:2.10} is non-differentiable, thus the standard method of demonstrating asymptotic normality via a Taylor expansion of the objective function is not applicable here. To address this, we derive the joint asymptotic distribution of the first-stage least squares estimator and the censored quantile estimator, $(\hat{\boldsymbol{\beta}}_{n}(\tau), \hat{\boldsymbol{\delta}}_{n})^\top$, by treating them as the joint solution to the system of first-order conditions defined in \eqref{eq:7.17}. We use these results to prove the asymptotic normality of $\mathscr{L}_{n}$.

We use an empirical process-based technique for weakly dependent data, which is especially useful for determining the asymptotic normality of estimators obtained from non-smooth loss functions. For example, \cite{Zhang21102025} used this strategy in threshold regression models, where estimators are determined by moment conditions that are not differentiable with respect to the parameters. However, in our scenario, the methods proposed by \cite{Zhang21102025} require modification to accommodate the two-stage estimation approach inherent in the Tobit model, which was not present in their framework.

As the loss function involved in \eqref{eq:2.10} is not differentiable at $\boldsymbol{\beta},$ where $\boldsymbol{x}^\top \boldsymbol{\beta}=0,$ we must ruled out a sequence of $\boldsymbol{x}_{i}$ values that are  orthogonal to $\boldsymbol{\beta}$ and have a positive frequency, establishing the asymptotic normality of $\hat{\boldsymbol{\beta}}_{n}(\uptau).$ The Assumption~\ref{a:3.3.7} is sufficient for this purpose.

This section investigates the asymptotic normality of $\mathscr{L}_{n}$ after appropriate normalization, and a few more assumptions, along with the earlier assumptions \ref{a:3.3.1} to \ref{a:3.3.5} are required.
 
  \begin{a1}\label{a:3.3.6}
       $\boldsymbol{\beta}_{0}(\uptau)$ is Lipschitz in $\uptau$ such that $\|\boldsymbol{\beta_{0}}(\mathcal{\tau^{'}})-\boldsymbol{\beta}_{0}(\uptau^{''})\|\leq c_{0}\;\|\uptau^{'}-\uptau^{''}\|,$ where $c_{0}$ is some constant.
     \end{a1}
     \smallskip
      \begin{a1}\label{a:3.3.7}
     For some constants $K_{1}>0$ and $\zeta_{0}>0$, and for all $1\leq i\leq n,$  the random vector $\boldsymbol{x}_{i}$ defined in \eqref{eq:2.4}, under the condition $\|\boldsymbol{\beta}-\boldsymbol{\beta}_{0}(\uptau)\|<\zeta_{0}$, satisfy the following inequality for all $0\leq z^*< \zeta_{0},$ and $ \tilde{r} = 0,1,2:$ 
$$\mathbb{E}\left[\mathds{1}\left(|\boldsymbol{x}^\top_{i}\boldsymbol{\beta}|\leq\|\boldsymbol{x}_{i}\|.z^*\right)\|\boldsymbol{x}_{i}\|^{\tilde{r}}\right]
  \leq K_{1}.z^*
  .$$
\end{a1}
\smallskip
 \begin{a1}\label{a:3.3.8}
 The Jacobian matrix $\mathbb{J}(\uptau)$ defined in \eqref{eq:3.3.5} is continuous with full rank, and $\mathbb{J}_{2,\boldsymbol{\delta}}(\uptau)$ has full column rank uniformly over $\mathcal{T}$. Moreover, the matrices $\mathbb{V}(\uptau,\uptau')$ and $\mathbb{V}_{\boldsymbol{\beta},\boldsymbol{\beta}}(\uptau,\uptau')$ defined in \eqref{eq:3.3.6} are uniformly positive definite over $\uptau\in \mathcal{T}.$ Particularly, there exist a constants $\underline{a}$ and $\overline{a}$ such that
\begin{equation}
\begin{split}
0<\underline{a}\leq \lambda_{\min}\big(\mathbb{V}_{\boldsymbol{\beta},\boldsymbol{\beta}}(\uptau,\uptau')\big)\leq\lambda_{\max}\big(\mathbb{V}_{\boldsymbol{\beta},\boldsymbol{\beta}}(\uptau,\uptau')\big)\leq \overline{a}< \infty, \\
0<\underline{a}\leq \lambda_{\min}\big(\mathbb{V}(\uptau,\uptau')\big)\leq\lambda_{\max}\big(\mathbb{V}(\uptau,\uptau')\big)\leq\overline{a}< \infty.
\end{split}
\end{equation}
Here $\lambda_{\min}(.)$ and $\lambda_{\max}(.)$ represent the minimum and maximum eigenvalues of the corresponding matrix.
 \end{a1}
 \smallskip
 \begin{a1}\label{a:3.3.9}Suppose that the kernel function $\mathbb{K}(.)>0$ defined in \eqref{eq:7.79} is symmetric with a compact support satisfying $\int \mathbb{K}(u)\;du = 1,$ and has a bounded first derivative. Moreover, the lag truncated parameter $b_{n}$ satisfies that $b_{n}\to \infty$ and $\frac{b_{n}^2}{n}\to 0,$ and the bandwidth $h_{n}$ satisfies the $h_{n}\to 0 $ and $nh_{n}^{2}\to \infty$ as $n \to \infty.$
 \end{a1}
  \smallskip
\begin{r1} Assumption \ref{a:3.3.6} implies that censored quantile regression coefficients are sufficiently smooth with respect to the quantile index. This is a reasonable assumption when the conditional distribution of $y^{*}$ given $\boldsymbol{x}$ is sufficiently smooth. 
Further, when for all $1\leq i\leq n $, $\boldsymbol{x}^\top_{i}\boldsymbol{\beta}>0,$ with probability 1, then  assumption \ref{a:3.3.7} holds uniformly in $\uptau \in \mathcal{T},$ under some smoothness condition on the distribution of $\boldsymbol{x}_{i}$. In addition, Assumptions  \ref{a:3.3.8} is needed to study the asymptotic distribution of joint estimators, it ensures that they do not degenerate for each $\uptau \in \mathcal{T},$ which is standard in the literature for censored quantile regression. Assumption \ref{a:3.3.9} is common across the literature of density estimation, and this assumption is useful in Theorems \ref{th:3.3.5} and \ref{th:3.3.6}.   
\end{r1}
\smallskip
Theorem \ref{th:3.3.3} describes the weak convergence of the joint stochastic process defined in \eqref{eq:7.16}, which is worthy of investigating by its own importance, and moreover, the assertion of this theorem is the key step to have the asymptotic distribution of $\mathscr{L}_{n}$ after appropriate normalization.
\begin{t1}\label{th:3.3.3} For the model
described in \eqref{eq:2.5}, under the assumptions \ref{a:3.3.1}- \ref{a:3.3.8}, we have 
$$\sqrt{n}\begin{pmatrix}
          \begin{bmatrix}
          \hat{\boldsymbol{\beta}}_{n}(\uptau)\\
          \hat{\boldsymbol{\delta}}_{n}
          \end{bmatrix} -
          \begin{bmatrix}
           \boldsymbol{\beta}_{0}(\uptau)\\
          \boldsymbol{\delta}_{0}
         \end{bmatrix}
         \end{pmatrix}\Rightarrow (\mathbb{G}_{2}(.),\mathbb{G}_{1})=\mathbb{G}(.)\quad in\quad \Delta \times\ell^{\infty}( \mathscr{B}\times\mathcal{T}),$$
where $\mathbb{G}(.)$ is centered Gaussian vector process with covariance function 
\begin{equation}\label{eq:3.3.5}
    \mathbb{E}\!\left[\mathbb{G}(\uptau)\mathbb{G}(\uptau')^\top\right]= \mathbb{J}^{-1}(\uptau)\;\mathbb{V}(\uptau,\uptau')\;\mathbb{J}^{-1}(\uptau).
\end{equation}
The Jacobian matrix $\mathbb{J}(\uptau)$ and covariance function matrix 
$\mathbb{V}(\uptau,\uptau')$ are given by 
\begin{equation}\label{eq:3.3.6}
\mathbb{J}(\uptau)=
\begin{bmatrix}
     \mathbb{J}_{2,\boldsymbol{\beta}}(\uptau) & \mathbb{J}_{2,\boldsymbol{\delta}}(\uptau)\\[6pt]
     \boldsymbol{0} & \mathbb{J}_{1,\boldsymbol{\delta}}
\end{bmatrix},
\qquad
\mathbb{V}(\uptau,\uptau')=
\begin{bmatrix}
        \mathbb{V}_{\boldsymbol{\beta},\boldsymbol{\beta}}(\uptau,\uptau') & \mathbb{V}_{\boldsymbol{\beta},\boldsymbol{\delta}} (\uptau)\\[6pt]
       \mathbb{V}_{\boldsymbol{\delta},\boldsymbol{\beta}}(\uptau)& \mathbb{V}_{\boldsymbol{\delta},\boldsymbol{\delta}}
\end{bmatrix},
\end{equation}
where the components are defined as
\begin{equation}\label{eq:3.3.7}
     \begin{split}
     \mathbb{J}_{1,\boldsymbol{\delta}}&=\mathbb{E}\big[\boldsymbol{z}\boldsymbol{z}^\top\big]\\
         \mathbb{J}_{2,\boldsymbol{\beta}}(\uptau)&=\mathbb{E}\big[\mathds{1}\big(\boldsymbol{x}^\top\boldsymbol{\beta}_{0}(\uptau)>0\big)f_{\varepsilon_{\uptau}}\big(0|\boldsymbol{x}\big)\boldsymbol{x}\boldsymbol{x}^\top\big]\\
         \mathbb{J}_{2,\boldsymbol{\delta}}(\uptau)&=\mathbb{E}\big[\rho_{10}\mathds{1}\big(\boldsymbol{x}^\top\boldsymbol{\beta}_{0}(\uptau)>0\big)f_{\varepsilon_{\uptau}}\big(0|\boldsymbol{x}\big)\boldsymbol{x}\boldsymbol{z}^\top\big]\\
          \mathbb{V}_{\boldsymbol{\delta},\boldsymbol{\delta}}&=\underset{n \to \infty}{\lim}\mathrm{Var}\left(\frac{1}{\sqrt{n}}\sum_{i=1}^{n}\Psi_{1}(\vartheta_{i},\boldsymbol{z}_{i},\boldsymbol{\delta}_{0})\right)
     \end{split}
 \end{equation}
 \begin{equation*}
     \begin{split}
          \mathbb{V}_{\boldsymbol{\beta},\boldsymbol{\beta}}(\uptau,\uptau')&=\underset{n \to \infty}{\lim}\frac{1}{n}\sum_{i=1}^{n}\sum_{j=1}^{n}\mathbb{E}\big[\Psi_{2}(\boldsymbol{s}_{j},\boldsymbol{\beta}_{0},\boldsymbol{\delta}_{0},\uptau)\Psi_{2}(\boldsymbol{s}_{i},\boldsymbol{\beta}_{0},\boldsymbol{\delta}_{0},\uptau')^\top\big]\\
          \mathbb{V}_{\boldsymbol{\beta},\boldsymbol{\delta}} (\uptau)&=\underset{n \to \infty}{\lim}\frac{1}{n}\sum_{i=1}^{n}\sum_{j=1}^{n}\mathbb{E}\big[\Psi_{2}(\boldsymbol{s}_{j},\boldsymbol{\beta}_{0},\boldsymbol{\delta}_{0},\uptau)\Psi_{1}(\vartheta_{i},\boldsymbol{z}_{i},\boldsymbol{\delta}_{0})^\top\big]\\
          \mathbb{V}_{\boldsymbol{\delta},\boldsymbol{\beta}}(\uptau)&=\mathbb{V}_{\boldsymbol{\beta},\boldsymbol{\delta}}(\uptau)^\top.
     \end{split}
 \end{equation*}
An application of partitioned matrix multiplication to this result will yield the conclusions:
$$\sqrt{n}(\hat{\boldsymbol{\delta}}_{n}-\boldsymbol{\delta}_{0})\xrightarrow{d}N_{p+1}\big(0,\; \mathbb{J}^{-1}_{1,\boldsymbol{\delta}}\mathbb{V}_{\boldsymbol{\delta},\boldsymbol{\delta}}(\mathbb{J}^{-1}_{1,\boldsymbol{\delta}})^\top\big),$$
and $$\sqrt{n}\big(\hat{\boldsymbol{\beta}}_{n}(\uptau)-\boldsymbol{\beta}_{0}(\uptau)\big)\Rightarrow \mathcal{G}_{2}(\uptau),\quad \text{in} \; \ell^{\infty}(\mathscr{B}\times\mathcal{T})$$
where $\mathcal{G}_{2}(.)$ is the Gaussian process with covariance kernel $\mathbb{E}[\mathcal{G}_{2}(\uptau)\mathcal{G}_{2}(\uptau')^\top]$ defined as
\begin{equation}\label{eq:3.3.8}
\begin{split}
  \mathbb{E}[\mathcal{G}_{2}(\uptau)\mathcal{G}_{2}(\uptau')^\top]&=\mathbb{J}^{-1}_{2,\boldsymbol{\beta}}(\uptau)\big\{ \mathbb{V}_{\boldsymbol{\beta},\boldsymbol{\beta}}(\uptau,\uptau')-\mathbb{V}_{\boldsymbol{\delta},\boldsymbol{\beta}}(\uptau)(\mathbb{J}^{-1}_{1,\boldsymbol{\delta}})^\top\mathbb{J}^\top_{2,\boldsymbol{\delta}}(\uptau)-\mathbb{J}_{2,\boldsymbol{\delta}}(\uptau)\mathbb{J}^{-1}_{1,\boldsymbol{\delta}}\mathbb{V}_{\boldsymbol{\beta},\boldsymbol{\delta}}(\uptau)\\
&+\mathbb{J}_{2,\boldsymbol{\beta}}(\uptau)(\mathbb{J}^{-1}_{1,\boldsymbol{\delta}}\mathbb{V}_{\boldsymbol{\delta},\boldsymbol{\delta}}(\mathbb{J}^{-1}_{1,\boldsymbol{\delta}})^\top)\mathbb{J}^\top_{2,\boldsymbol{\beta}}(\uptau)\big\}(\mathbb{J}^{-1}_{2,\boldsymbol{\beta}}(\uptau))^\top. 
\end{split}
\end{equation}

\end{t1}
\begin{proof}[$\textbf{Proof}$]
See the proof in Appendix-\ref{B.4}.
\end{proof}
\begin{r1}\label{r4}Note that, if all regressors are exogenous, the first component of the covariance matrix in \eqref{eq:3.3.8} coincides with the asymptotic variance of the censored quantile estimator $\hat{\boldsymbol{\beta}}_{n}(\uptau)$ in model \eqref{eq:2.5}. By Theorem~\ref{th:3.3.3}, the presence of endogenous regressors increases the variability of $\hat{\boldsymbol{\beta}}_{n}(\uptau),$ in a positive-semidefinite sense through the last term
$\mathbb{J}_{2,\boldsymbol{\beta}}(\uptau)(\mathbb{J}^{-1}_{1,\boldsymbol{\delta}}\mathbb{V}_{\boldsymbol{\delta},\boldsymbol{\delta}}(\mathbb{J}^{-1}_{1,\boldsymbol{\delta}})^\top)\mathbb{J}^\top_{2,\boldsymbol{\beta}}(\uptau),$ which captures the variability from the first-stage estimation. In addition, the cross-terms represent the covariance between first-stage and second-stage estimators. Thus, the full adjustment comprising the positive semidefinite sense together with the cross-variance corrections arises from the the first-stage residuals$e_{i}.$ Whenever $\mathbb{J}_{2,\boldsymbol{\delta}}(\uptau)\neq 0,$
ignoring this adjustment leads to asymptotic standard errors that are strictly smaller in the positive semidefinite sense than the correct standard errors of $\hat{\boldsymbol{\beta}}_{n}(\uptau),$ and therefore to confidence intervals that are invalid under endogeneity.
\end{r1}
\begin{r1}\label{r5}
An implication of Theorem \ref{th:3.3.3} is that for any given $\uptau,$ we have
$$\sqrt{n}\begin{pmatrix}
          \begin{bmatrix}
          \hat{\boldsymbol{\beta}}_{n}(\uptau)\\
          \hat{\boldsymbol{\delta}}_{n}
          \end{bmatrix} -
          \begin{bmatrix}
           \boldsymbol{\beta}_{0}(\uptau)\\
          \boldsymbol{\delta}_{0}
         \end{bmatrix}
         \end{pmatrix}\xrightarrow{d} Z,$$
where $Z$ is a $\mathbb{R}^{2p+3}$ valued random vector associated with a $2p+3$-dimensional multivariate normal
distribution, with mean vector $\boldsymbol{0}$ and covariance matrix $\mathbb{J}(\uptau)^{-1}\mathbb{V}(\uptau,\uptau)$
$\mathbb{J}(\uptau)^{-1}.$
Next, suppose that for any finite collection of quantile indices, $\{\uptau_{i}, i=\{1,2,\ldots,k\}\},$ the vector of k regression quantile estimator is represented by $\Hat{\zeta_{n}}=(\hat{\boldsymbol{\beta}}_{n}(\uptau_{1}),\ldots,\Hat{\boldsymbol{\beta}}_{n}(\uptau_{k}),\hat{\boldsymbol{\delta}}_{n})^\top,$
and the corresponding vector of true unknown regression quantiles is represented by
$\zeta_{0}=(\boldsymbol{\beta}_{0}(\uptau_{1}),\ldots,\boldsymbol{\beta}_{0}(\uptau_{k}),\boldsymbol{\delta}_{0})^\top$. Then 
$$\sqrt{n}(\Hat{\zeta}_{n}-\zeta_{0})\xrightarrow{d} \Tilde{Z},$$
where $\Tilde{Z}$ is a $\mathbb{R}^{k(p+2)+(p+1)}$ valued random vector associated with a standard m-dimensional multivariate
normal distribution with mean zero vector and covariance matrix $\Omega_{k}=\mathbb{J}_{k}^{-1}\mathbb{V}_{k}\mathbb{J}_{k}^{-1}.$ Here Jacobian matrix $\mathbb{J}_{k}$ is defined by 
$$
\mathbb{J}_k =
\begin{bmatrix}
\mathbb{J}_{2,\boldsymbol{\beta}}(\uptau_1) & & & \mathbb{J}_{2,\boldsymbol{\delta}}(\uptau_1)\\
& \ddots & & \vdots\\
& & \mathbb{J}_{2,\boldsymbol{\beta}}(\uptau_k) & \mathbb{J}_{2,\boldsymbol{\delta}}(\uptau_k)\\
\boldsymbol{0} & \cdots & \boldsymbol{0} & \mathbb{J}_{1,\boldsymbol{\delta}}
\end{bmatrix}.$$
The covariance matrix $\mathbb{V}_k$ is defined by 
$$
\mathbb{V}_k =
\begin{bmatrix}
\mathbb{V}_{\boldsymbol{\beta},\boldsymbol{\beta}}(\uptau_1,\uptau_{1}) & \cdots & \mathbb{V}_{\boldsymbol{\beta},\boldsymbol{\beta}}(\uptau_1,\uptau_k) & \mathbb{V}_{\boldsymbol{\beta},\boldsymbol{\delta}}(\uptau_1)\\
\vdots & \ddots & \vdots & \vdots \\
\mathbb{V}_{\boldsymbol{\beta},\boldsymbol{\beta}}(\uptau_k,\uptau_1) & \cdots & \mathbb{V}_{\boldsymbol{\beta},\boldsymbol{\beta}}(\uptau_k,\uptau_k) & \mathbb{V}_{\boldsymbol{\beta},\boldsymbol{\delta}}(\uptau_k,\uptau_k)\\
\mathbb{V}_{\boldsymbol{\delta},\boldsymbol{\beta}}(\uptau_1) & \cdots & \mathbb{V}_{\boldsymbol{\delta},\boldsymbol{\beta}}(\uptau_k) & \mathbb{V}_{\boldsymbol{\delta},\boldsymbol{\delta}}
\end{bmatrix}.
$$
\end{r1}

Finally, Theorem \ref{th:3.3.4} describes the asymptotic distribution of $\mathscr{L}_{n}$ after appropriate normalization, which is the key result of this work.

\begin{t1}\label{th:3.3.4} For the model
described in \eqref{eq:2.5}, under the assumptions \ref{a:3.3.1}--\ref{a:3.3.8}, we have 
$$\sqrt{n}(\mathscr{L}_{n}-\mathscr{L}_{0})\xrightarrow{d}Z_{1},$$
where $Z_1$ is a $\mathbb{R}^{p+2}$ valued random vector associated with $(p+2)$-dimensional multivariate normal distribution, with mean vector $\boldsymbol{0}$ and covariance matrix,
\begin{equation}\label{eq:3.3.9}
\boldsymbol{\Omega}= \underset{n \to \infty}{\lim}\mathrm{Var}\Big(\frac{1}{\sqrt{n}}\sum_{i=1}^{n}h_{i}\Big).
\end{equation}
Here, $h_{i}=\int_{0}^{1}v_{i}(\uptau)J_{1}(\uptau)\;d(\uptau),$ where  $v_{i}(\uptau)$ is defined in \eqref{eq:7.50} and $J_{1}(\uptau)$ is given in \ref{a:3.3.5}.
\end{t1}
\begin{proof}[$\textbf{Proof}$]
See the proof in Appendix-\ref{B.5}.
\end{proof}
\begin{r1}\label{r6} In order to implement the assertion in Theorem \ref{th:3.3.4} in practice, one needs to consistently estimate $\boldsymbol{\Omega}$ described in the statement of this theorem. 
\end{r1}
\subsection{Consistent Estimator of \texorpdfstring{$\boldsymbol{\Omega}$}{}}
 As we mentioned in Remark \ref{r6}, for a given data, one needs to estimate $\boldsymbol{\Omega}$ consistently to implement the asymptotic distribution stated in Theorem \ref{th:3.3.4}. Let $\hat{\boldsymbol{\Omega}}$ be an HAC estimator of $\boldsymbol{\Omega}$, which is defined as follows:  
\begin{eqnarray}\label{eq:3.3.10}
\Hat{\boldsymbol{\Omega}}=\frac{1}{n}\sum_{j=-n+1}^{n}\;\sum_{i=1}^{n-j}\hat{h}_{i}\hat{h}_{i+j}^\top \;\mathbb{K}\Big(\frac{j}{b_{n}}\Big), \end{eqnarray}
where $\mathbb{K}(.)$ and $b_{n}$ are given in Assumption~\ref{a:3.3.9}. Here $\hat{h}_{i} = \int\limits_{0}^{1}\hat{\Tilde{v}}_{i}(\tau)\hat{J}_{1}(\tau) d(\tau)$, where $\hat{\Tilde{v}}_{i}(\uptau)=\hat{\mathbb{J}}^{-1}_{2,\boldsymbol{\beta}}(\uptau)\big\{\Psi_{2}\big(\boldsymbol{s}_{i},\hat{\boldsymbol{\beta}}_{n}(\uptau),\hat{\boldsymbol{\delta}}_{n},\uptau\big)-\hat{\mathbb{J}}_{2,\boldsymbol{\delta}}(\uptau)\hat{\mathbb{J}}^{-1}_{1,\boldsymbol{\delta}}
 \Psi_{1}\big(\vartheta_{i},\boldsymbol{z}_{i},\hat{\boldsymbol{\delta}}_{n}\big)\big\},$
$i = 1, \ldots, n - j$ .

Theorem \ref{th:3.3.5} establishes the uniform consistency of $\hat{\mathbb{J}}(\uptau)$ and $\hat{\mathbb{V}}(\uptau,\uptau')$ (see \eqref{eq:7.73} and \eqref{eq:7.78}, respectively for their expressions). It is an intermediate result, and it is required to establish the consistency of $\hat{\boldsymbol{\Omega}}.$
\begin{t1}\label{th:3.3.5} Under the Assumptions \ref{a:3.3.1}-\ref{a:3.3.9}, we have
$$\hat{\mathbb{J}}(\uptau)-\mathbb{J}(\uptau)=o_{p}(1),\quad \hat{\mathbb{V}}(\uptau,\uptau')-\mathbb{V}(\uptau, \uptau')=o_{p}(1),$$
as $n \to \infty$ uniformly over $\uptau \in \mathcal{T}.$ Here, $\hat{\mathbb{J}}(\uptau)$ and $\hat{\mathbb{V}}(\uptau,\uptau')$ are the same as defined in \eqref{eq:7.73} and \eqref{eq:7.78}, respectively, and $\mathbb{J}(\uptau)$ and $\mathbb{V}(\uptau, \uptau')$ are the same as defined in \eqref{eq:3.3.6}.
\end{t1}
\begin{proof}[$\textbf{Proof}$]
See the proof in Appendix-\ref{B.6}.
\end{proof}
Now, Theorem \ref{th:3.3.6} asserts the weak consistency of $\hat{\boldsymbol{\Omega}}$ to $\boldsymbol{\Omega}$.
\begin{t1}\label{th:3.3.6}
Under the conditions \ref{a:3.3.1}-- \ref{a:3.3.9}, we have $$||\hat{\boldsymbol{\Omega}}-\boldsymbol{\Omega}||_{\mathbb{R}^{p+2}}=o_{p}(1),$$ as $n \to \infty,$  where $\hat{\boldsymbol{\Omega}}$ and $\boldsymbol{\Omega}$ are the same as defined in \eqref{eq:3.3.10} and \eqref{eq:3.3.9}, respectively. 
\end{t1}
\begin{proof}[$\textbf{Proof}$]
See the proof in Appendix-\ref{B.6}.
\end{proof}
Next, Corollary \ref{cor1} follows from the assertions in Theorem \ref{th:3.3.6} and Theorem \ref{th:3.3.4}. 

\begin{corollary}\label{cor1}
Under the conditions \ref{a:3.3.1}--\ref{a:3.3.6}, we have $$\sqrt{n}\hat{\boldsymbol{\Omega}}^{-\frac{1}{2}}(\mathscr{L}_{n}-\mathscr{L}_{0})\Rightarrow Z_{2},$$ where $Z_{2}$ is a $\mathbb{R}^{p+2}$ valued random vector associated with standard $p+2$-dimensional multivariate normal distribution. 
\end{corollary}

\begin{proof}[$\textbf{Proof}$]
 See the proof in Appendix-\ref{B.6}.
\end{proof}


\begin{r1}\label{r7}
In order to implement any statistical methodology based on the asymptotic properties of $\mathscr{L}_{n}$, one may employ the result stated in Corollary \ref{cor1} as $\hat{\boldsymbol{\Omega}}^{-\frac{1}{2}}$ is computable for a given data.
\end{r1}

\section{Finite Sample Simulation Studies}\label{FSSS}
The asymptotic properties of the estimator give an idea about the behaviour of the estimator when the sample size is large. To study and understand the behaviour of the proposed estimators in finite samples, we conduct a Monte Carlo simulation study. 

 The data on $ y^{*}_{i} (i = 1, \ldots, n)$ is generated from the  model:
\begin{equation}\label{eq:4.1}
          y^{*}_{i} = \beta_{0}+\beta_{1}x_{i}+\beta_{2}w_{i}+\epsilon_{i},
\end{equation}
and the data on the endogenous variable $w_{i}$ is generated using 
\begin{equation}\label{eq:4.2}
    w_{i} = \Tilde{\delta} z_{i}+\nu_{i},
\end{equation} where the random errors $\nu_{i}$ are generated from   $N(0,1)$, while $z_{i}$ is considered as the fixed observations from the $U(0,1)$. Afterwards, the observations on $w_{i}$ are generated for given values of $\Tilde{\delta}=1$. To address endogeneity within the model, we employ the instrumental variable estimation techniques (see Section \ref{Est}) and  decompose the error term  $\epsilon_{i}$ in \eqref{eq:4.1} into two distinct components  as:
\begin{equation}\label{eq:4.3}
    \epsilon_{i}=\rho_{1}\nu_{i}+\varepsilon_{i}.
\end{equation}
The random errors $\epsilon_{i}$ are generated using the relationship provided in \eqref{eq:4.3} for the given $\rho_1$ and the observations on  $\varepsilon_{i}$. Then $y_{i}$ is generated using  
\begin{equation}\label{eq:4.4}
          y_{i} =\max(0, y^{*}_{i})= \max(0,\beta_{0}+\beta_{1}x_{i}+\beta_{2}w_{i}+\rho_{1}\nu_{i}+\varepsilon_{i}).
\end{equation}
In this study, we choose $\boldsymbol{\beta} = (\beta_{0}, \beta_{1}, \beta_{2}, \rho_{1}) = (1, 2, 3, 0.5)$ and generate exogenous regressors $x_{i}$ from $N(0, 1)$. 

\noindent{\bf Dependent Errors:}
\noindent To introduce the dependent error structure, consider an auto-regressive process of first order as
\begin{equation}\label{eq:4.5}
    \varepsilon_{i}=\rho^{*} \varepsilon_{i-1}+\eta_{1,i}, i = 1, \ldots, n,
    \end{equation}
where $|\rho^{*}|\leq 1$ is the autocorrelation coefficient, and $\eta_{1,i}$ ($i = 1, \ldots, n$) are independent and identically distributed random variables that follow $N(0, \sigma^{2})$, see \cite{Zhou} for discussion on such models. In this numerical study, $\varepsilon_{i},$s are generated using the model \eqref{eq:4.5} with $\rho^{*} = 0.5$ and the observations $\{\eta_{1,1}, \ldots, \eta_{1,n}\}$.

The following examples are considered for the dependent error structure: 
\begin{e1}{(\textbf{$\Tilde{\alpha}$-trimmed mean}):}\label{e1} Let $J_{1}(\uptau)=\frac{1}{1-2\Tilde{\alpha}}$, where $\Tilde{\alpha} <\uptau <1-\Tilde{\alpha}$ and $\Tilde{\alpha} \in (0, 1/2)$. Then, the corresponding $\mathscr{L}_{n}$ (see \eqref{eq:2.12}) estimator takes the following form:
\begin{equation}
\mathscr{L}_{n}=\frac{1}{1-2\Tilde{\alpha}}\int_{\Tilde{\alpha}}^{1-\Tilde{\alpha}}\Hat{\boldsymbol{\beta}}_{n}(\uptau)\; d\uptau.
\end{equation}
\end{e1}

\begin{e1}{(\textbf{$\Tilde{\alpha}$-Winsorized mean}):}\label{e2} Let $J_{1}(\uptau)=1$, where $\Tilde{\alpha} <\uptau <1-\Tilde{\alpha}$ and $\Tilde{\alpha} \in (0,1/2)$. Then, the corresponding $\mathscr{L}_{n}$ estimator takes the following form:
\begin{equation}
\mathscr{L}_{n}=\int_{\Tilde{\alpha}}^{1-\Tilde{\alpha}}\Hat{\boldsymbol{\beta}}_{n}(\uptau)\; d\uptau +\Tilde{\alpha} (\Hat{\boldsymbol{\beta}}_{n}(\Tilde{\alpha}) + \Hat{\boldsymbol{\beta}}_{n}(1-\Tilde{\alpha})).
\end{equation}
\end{e1}

\begin{e1}\label{e3} Let $J_{1}(\uptau)= 6\uptau (1-\uptau)$, where $0 \leq \uptau \leq 1$. Then, the corresponding $\mathscr{L}_{n}$ estimator takes the following form:
\begin{equation}
\mathscr{L}_{n}=\int_{0}^{1} 6\uptau(1-\uptau)\;\Hat{\boldsymbol{\beta}}_{n}(\uptau)\; d\uptau. 
\end{equation}
\end{e1}

When data contains outliers or is heavily skewed, typical mean-based L-estimators might be greatly influenced. Trimmed or Winsorized means producing statistical summaries that are less impacted by extreme values, making them more informative.
Here, we use $\Tilde{\alpha} = \{0.01, 0.02, 0.20\}$ in Examples \ref{e1} and \ref{e2}, and $\beta_{0}(\uptau),$ $\beta_{1}(\uptau),$ $\beta_{2}(\uptau),$ and $\rho_{1}(\uptau)$ in \eqref{eq:4.4} are computed using 
 the ``\texttt{quantreg}'' function in R software. All the R codes for the simulated data study were executed on the server in the Linux operating system. The following steps are summarized for evaluating the integral involved in $\mathscr{L}_{n}.$

 \begin{itemize}
    \item[$\bullet$] Generate $M$ many random observations $\{\uptau_{1}, \uptau_{2},\ldots, \uptau_{m}\}$ from $\mathbb{U}(0,1).$
    \item[$\bullet$] For each sample $\uptau_{j}$, estimate $\boldsymbol{\beta}=(\beta_{0},\beta_{1},\beta_{2},\rho_{1})^\top$ using quantile regression.
    \item[$\bullet$] The integral $\mathscr{L}_{n} = (\mathcal{L}_{n0}, \mathcal{L}_{n1}, \mathcal{L}_{n2}, \mathcal{L}_{n3})^\top\in \mathbb{R}^{4}$ defined in \eqref{eq:2.11} is approximated  by
    $$\mathscr{L}_{n} \approx \frac{1}{M}\sum_{j=1}^{M}\Hat{\boldsymbol{\beta}}_{n}(\uptau_{j})J_{1}(\uptau_{j}),$$
    and $\mathscr{L}_{0}$ is approximated  by 
    $$\mathscr{L}_{0} \approx \frac{1}{M}\sum_{j=1}^{M}\boldsymbol{\beta}(\uptau_{j})J_{1}(\uptau_{j}).$$
\end{itemize}
As the estimates were getting converged with 2000 replications, we repeated this experiment $r=2000$ times for the sample size $n \in \{50,75,100,200,500,700,1000\}.$  Observe that $\mathscr{L}_{0}\in\mathbb{R}^{4}$, and for that reason, the  component-wise empirical bias Ebias and component-wise empirical mean squared error  of the parameter estimates of $$\mathscr{L}_{0}=(\mathcal{L}_{0},\mathcal{L}_{1},\mathcal{L}_{2},\mathcal{L}_{3})^\top$$ are computed as follows:
$$~\mbox{Ebias}~(\mathcal{L}_{j} )\approx \frac{1}{r} \sum_{k=1}^{r} \left(\mathcal{L}_{njk}-\mathcal{L}_{j}\right),j = 0, 1, 2, 3$$
and  
$$~\mbox{EMSE}~(\mathcal{L}_{j})\approx \frac{1}{r} \sum_{k=1}^{r} \left(\mathcal{L}_{njk}-\mathcal{L}_{j}\right)^2 , j = 0, 1, 2, 3$$
where $\mathcal{L}_{njk}$ is the estimate of $\mathcal{L}_{nj}$ ($j = 0, 1, 2, 3$) for the $k$-th replicate, where $k = 1, \ldots, r$.
For these designs, the overall censoring proportions vary between $28\%$ to $46\%$, and it is observed that the different censoring percentages do not change the qualitative conclusions.
In Section \ref{SST1}, tables \ref{table:6c}-\ref{table:8c} present the Ebias and  of the proposed estimators with censored observations for different sample sizes across the three examples. To maintain brevity, the Ebias and  are only reported for sample sizes $n\in \{50, 100, 500, 1000\}$.

The results demonstrate that all estimators improve as the sample size increases. Both Ebias and  approach zero, confirming their consistency under the dependent error. In examples \ref{e1} and \ref{e2}, the estimates $\mathscr{L}_{n2}$ and $\mathscr{L}_{n3}$ perform better with smaller Ebias and  across all sample sizes. In contrast, $\mathscr{L}_{n1}$ has a significant negative Ebias, making it the least efficient estimator. In Example~\ref{e3}, however, $\mathscr{L}_{n1}$ performs best overall, since it has the lowest Ebias and  throughout all $n$, whereas $\mathscr{L}_{n0}$ and $\mathscr{L}_{n3}$ improve with $n$ but remains less efficient.

Trimming is employed to enhance the robustness of the estimators by excluding extreme observations. The effect of the trimming proportion $\Tilde{\alpha}$ is clearly reflected in the results. When the trimming level is small ($\Tilde{\alpha} = 0.01, 0.02$), the estimators retain low bias, thereby preserving efficiency. However, when the trimming proportion is increased to $\Tilde{\alpha} = 0.20$, the Ebias becomes noticeably larger across all estimators. This inflation occurs because excessive trimming discards too much relevant information, leading to a loss of efficiency along with higher bias. Therefore, a moderate trimming level offers the best balance, providing robustness against outliers without substantially sacrificing efficiency.
\begin{figure}[H]
    \centering
    \includegraphics[width=\linewidth]{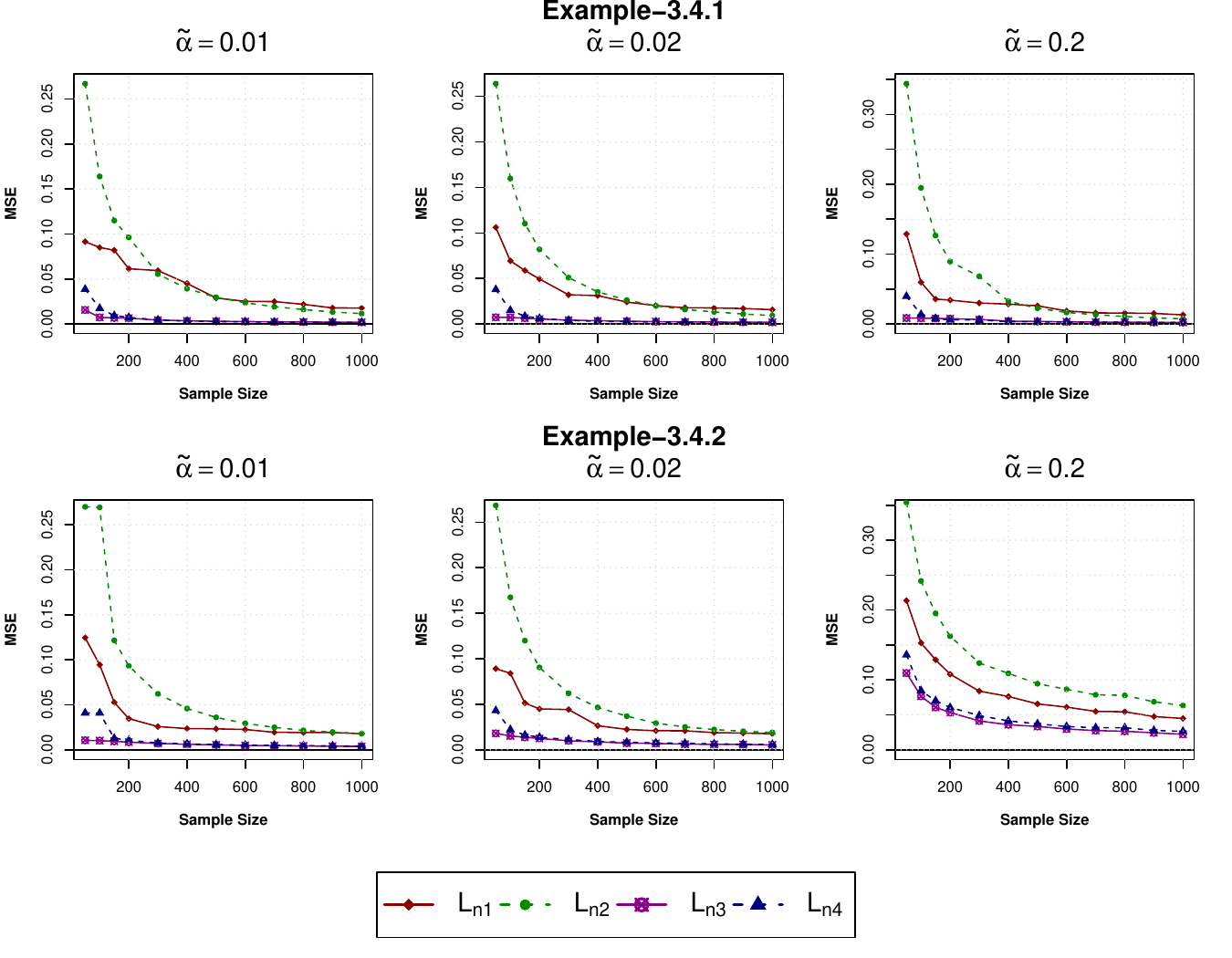}
    \caption{Empirical MSE of parameter estimates using $\Tilde{\alpha}$-trimmed mean and $\Tilde{\alpha}$-Winsorized mean (see Example \ref{e1}-\ref{e2}) for $\Tilde{\alpha}$ = 0.01, 0.02, 0.20 when errors are dependent and $n = 50, 100, 200, . . . , 1000$}
    \label{fig:3.1}
\end{figure}
\begin{figure}[H]
    \centering
    \includegraphics[width=0.55\linewidth, height=0.55\textheight]{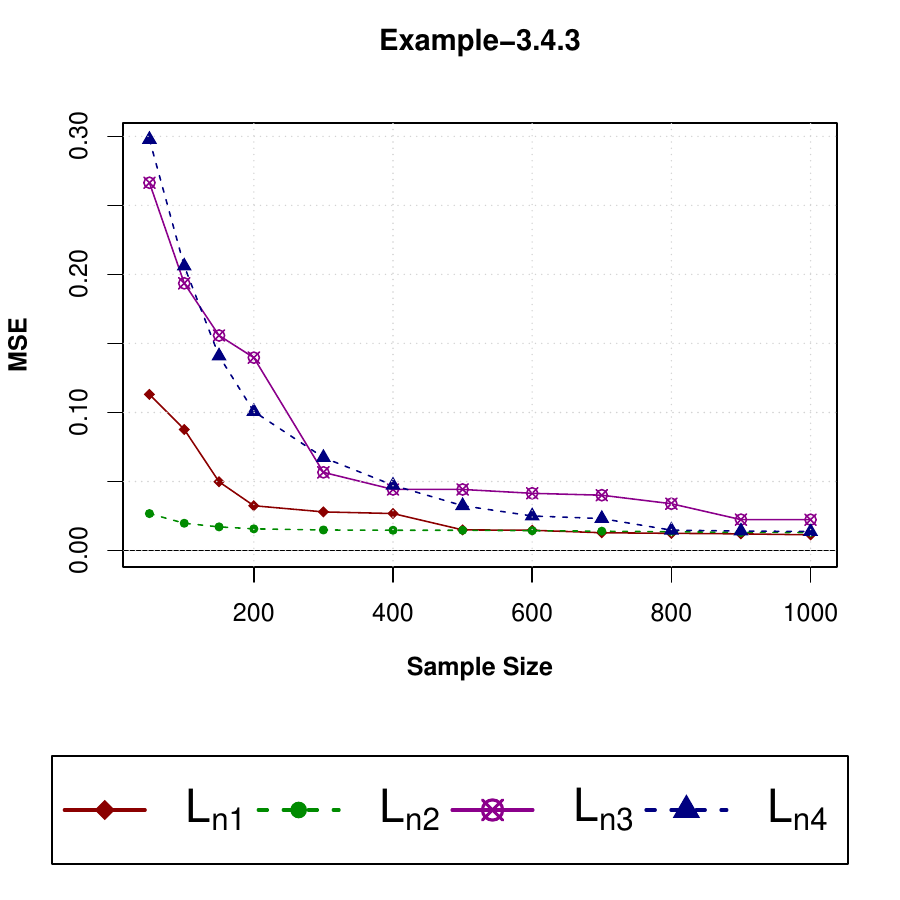}
    \caption{Empirical MSE of parameter estimates of Example \ref{e3} when errors are dependent and $n = 50, 100, 200, . . . , 1000.$}
    \label{fig:3.2}
\end{figure}
The graphical results for , presented in Figures \ref{fig:3.1}–\ref{fig:3.2}, demonstrate that the proposed estimator $\mathscr{L}_{n}$ performs well across the considered examples. The plots reveal that  tends to increase with higher trimming proportions, particularly in smaller sample sizes, while it decreases substantially as the sample size grows. This pattern highlights the usefulness of $\mathscr{L}_{n}$ in handling datasets contaminated with outliers. A similar trend is observed for the magnitude of Ebias, further confirming the robustness of the estimator.
\section{Real Data Analysis}\label{RDA}
We analyze well-known Mroz dataset, which is included in the ``Wooldridge" R package (available at \url{https://github.com/swati-1602/CRLM-Real-data.git}) to demonstrate the application of our proposed estimators. This dataset focuses on the labor force participation of U.S. women and includes data for 753 individuals. The response variable of interest is the total number of hours worked per week. Notably, 325 out of the 753 women reported not working any hours, rendering the dependent variable left-censored at zero. This implies that the censoring proportion C.P. is of approximately $0.43$.

In this dataset, the explanatory variables include the woman's years of education (educ), years of experience (exper) along with its squared value (expersq), age (age), the number of young children under six years (kidslt6), the number of older children six years or older (kidsge6), and non-wife household income (nwifeinc). We treat non-wife household income as an endogenous variable due to its potential correlation with unobserved household preferences that could affect the wife's labour force participation. In this study,  the variable, namely, the husband's years of education (huseduc) is considered as an instrumental variable in addressing this endogeneity. It is assumed that this variable does not directly impact the wife's decision to enter the labour force, but it presumably affects both his income and the income of the now-wife's household. The graphical representation in Figure \ref{Data} illustrates censored observations for the working hours of married women. The hours variable represents the total number of hours a woman worked in a year. A married woman's maximum recorded work hours are 4950. For instance, a typical full-time job consists of 40 hours per week for 52 weeks, totaling 2080 hours each year. Thus, 4950 hours imply that the individual worked around 91 hours per week in a regular 52-week year, indicating higher labour force participation.
\begin{figure}[hbt!]
    \centering
    \includegraphics[width=0.7\textwidth, height=0.5\textheight]{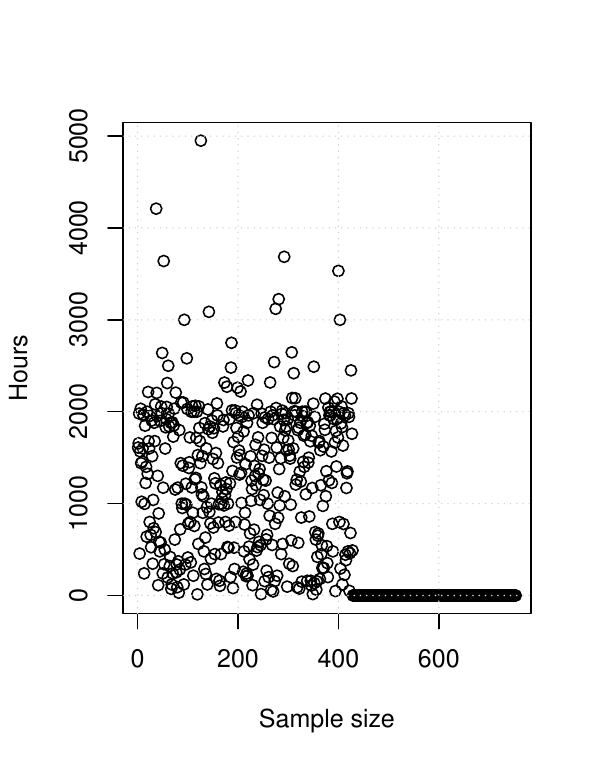}
    \caption{Working hours (yearly) for married women}
    \label{Data}
\end{figure}
 
 We obtain the estimates of the parameters and their bootstraps root mean squared error 
 to verify the theoretical result on the real data set. The steps to compute the bootstrap root mean squared error BRMSE of the parameters involved in the model are described in \eqref{eq:2.4}:
\begin{itemize}
    \item[$\blacksquare$]  Following the same algorithm provided in Section \ref{FSSS}, for the given data with size $n$, we find the approximated values of the $q$  components of $\mathscr{L}_{n}$ and  denote those components as $\mathcal{L}_{n, 1}, \ldots, \mathcal{L}_{n, q},$ where $q=8$ and $n=753.$ 
 \item[$\blacksquare$]Generate $b$ bootstrap samples with the same sample size as the original data set.
  \item[$\blacksquare$] Compute $\mathscr{L}_{n}$ for $b$ bootstrap resamples, which are denoted as $\{\mathscr{L}_{n, 1}, \ldots, \mathscr{L}_{n, b}\}$.
  \item[$\blacksquare$] Finally, 
   the bootstrap root mean squared error BRMSE is computed as 
     $$ BRMSE (\mathcal{L}_{n, i}) =\sqrt{\frac{1}{b}\sum_{k=1}^{b}\left(\mathcal{L}_{n, i, k}-\mathcal{L}_{n, i}\right)^2},$$
where for $i = 1, \ldots, q$, $\mathcal{L}_{n, i, k}$ is the value of $\mathcal{L}_{n, i}$ for the $k$-th bootstrap resample ($k = 1, \ldots, b$).  
\end{itemize}
Table \ref{table:1r} represents the parameter estimates for the model defined in \eqref{eq:2.4} for the $\Tilde{\alpha}$-trimmed mean and the $\Tilde{\alpha}$-Winsorized mean defined in Section \ref{FSSS}. These estimates are provided for $\Tilde{\alpha} = \{0.01, 0.02, 0.20\}$. Table \ref{table:2r} reports the BRMSE  for these parameter estimates across different bootstrap sample sizes
$b =\{25,50,75,100,200\}$. Additionally, Table \ref{table:3r} contains the parameter estimates and BRMSE of the estimator defined in Example \ref{e3} in Section \ref{FSSS}.

\begin{table}[H]
    \small
    \centering
    \caption{Parameter estimates for $\Tilde{\alpha}$-trimmed mean and $\Tilde{\alpha}$-Winsorized mean, for $\Tilde{\alpha}=0.01, 0.02, 0.20$}
    \begin{tabular}{|c|c|c|c|c|c|c|}
       \toprule
        \specialrule{\heavyrulewidth}{\abovetopsep}{\belowbottomsep}
          & \multicolumn{3}{c}{ $\boldsymbol{\Tilde{\alpha}}$\bf-trimmed mean}&\multicolumn{2}{c}{$\boldsymbol{\alpha}$\bf-Winsorized mean}&\\
          \specialrule{\heavyrulewidth}{\abovetopsep}{\belowbottomsep}
           \bf Variables&$\boldsymbol{\Tilde{\alpha} = 0.01}$&$\boldsymbol{\Tilde{\alpha} = 0.02}$&$\boldsymbol{\Tilde{\alpha} = 0.20}$&$\boldsymbol{\Tilde{\alpha} = 0.01}$&$\boldsymbol{\Tilde{\alpha} = 0.02}$&$\boldsymbol{\Tilde{\alpha} = 0.20}$\\
        \midrule
       \specialrule{\heavyrulewidth}{\abovetopsep}{\belowbottomsep}
        age &2.3004 &1.9851&-4.4195&2.6723&3.2969&-0.0265 \\
        educ &15.1809 &15.3170&12.3471&14.6552&14.3941&11.9117\\
        exper &23.7517&23.5119 &21.8545&24.7803&23.4696& 21.4691 \\
        expersq & -4.2199&-4.2695&-4.7762&-4.3246&-4.5611&-4.2299\\
        kidslt6 & -7.7039&-7.0448&-8.3693&-5.7924&-7.7605&-5.6729\\
        kidsge6 & -4.3523& -4.0552&-2.5310&-5.4866&-5.6679& 0.4380\\
        nwifeinc & -1.3488 & -1.9850&-5.3529&-1.2859&-1.4251&-2.2519\\
        residual & -4.4149 & -4.1638&-4.3432&-4.5116&-3.8497& -5.2681\\
        \bottomrule  
        \specialrule{\heavyrulewidth}{\abovetopsep}{\belowbottomsep}     
    \end{tabular}
    \label{table:1r}
\end{table}

\begin{table}[H]
    \small
    \centering
    \caption{Parameter estimates and BRMSE for the estimator defined in Example \ref{e3}}
    \begin{tabular}{|c|c|c|c|c|c|c|}
        \toprule
         \specialrule{\heavyrulewidth}{\abovetopsep}{\belowbottomsep}
         & \multicolumn{5}{c}{\bf Estimator and BRMSE for Example \ref{e3}}& \\
          \specialrule{\heavyrulewidth}{\abovetopsep}{\belowbottomsep}
           \bf Variables& \bf Estimates&$\boldsymbol{b = 25}$&$\boldsymbol{ b = 50}$ & $\boldsymbol{b = 75}$&$\boldsymbol{b =100}$& $\boldsymbol{b = 200}$ \\
        \midrule
    \specialrule{\heavyrulewidth}{\abovetopsep}{\belowbottomsep}
        age & -1.9084& 4.1714&5.6643&5.6334&5.4617&5.3665
        \\
        educ & 19.8326&17.4322&19.4899&18.4279&18.4255&17.8923 \\
        exper & 29.4500 &27.9768&29.6823&29.6398&29.1016&28.7205 \\
        expersq & -3.6951 & 9.3467&7.7024&7.4326&6.0869&6.0422\\
        kidslt6 & -9.6599 &10.8408&10.6208&11.3201&11.2522&10.6196 \\
        kidsge6 & -3.3560 &8.9355&7.5635&7.0924& 6.6037&5.4690\\
        nwifeinc & -4.1847&8.1198& 7.6849&7.5052& 7.4043& 6.6104\\
        residual & -4.3298 &8.1653&8.1534&8.5166&8.0292&7.9207\\
        \bottomrule  
        \specialrule{\heavyrulewidth}{\abovetopsep}{\belowbottomsep}     
    \end{tabular}
    \label{table:3r}
\end{table}
In Table \ref{table:1r}, $\Tilde{\alpha}$-trimmed mean and $\Tilde{\alpha}$-Winsorized mean demonstrate a consistent trend: the positive age coefficient when $\Tilde{\alpha} = {0.01, 0.02}$, and it suggests that as married women's age increases, they tend to increase their hours of work. However, for $\Tilde{\alpha} = 0.20$, the conclusion is the opposite, which makes sense as hours of work are supposed to be decreased as age increases for married women. It further indicates that the data set is likely to have some outliers (see Figure \ref{Data}). Also, the estimator defined in Example \ref{e3} in Section \ref{FSSS} presents a negative age coefficient, which implies that as married women grow older, they tend to reduce their working hours. Despite these differences, all three estimators demonstrate a positive coefficient for education (education) and experience (experience), implying that married women with higher levels of education and experience tend to work longer hours. Also, the coefficient for the squared term of experience (expersq) is negative for all three estimators. This suggests that while greater experience initially leads to more working hours, this effect diminishes as experience increases. Furthermore, the presence of small children under 6 or children over 6 is associated with fewer working hours, as indicated by the negative coefficients for the variables (kidslt6) and (kidsge6). In the overall context, working hours are positively influenced by age, education, and experience in the case of $\Tilde{\alpha}$-trimmed mean and $\Tilde{\alpha}$-Winsorized mean. However, they are negatively impacted by the expersq, kidslt6 or kidsge6, nwifeinc, and whereas in the case of the estimator defined in Example \ref{e3}, working hours are positively influenced by education and experience but negatively affected by the presence of ``kidsge6'', ``nwifeinc'', and other factors. The residual term also has a negative impact in this case.
\begin{table}[H]
    \small
    \centering
    \caption{BRMSE  of parameter estimates for $\Tilde{\alpha}$-trimmed mean and $\Tilde{\alpha}$-Winsorized mean when $\Tilde{\alpha} = 0.01,0.02, 0.20$}
      \resizebox{\textwidth}{!}{
 \begin{tabular}{|c|ccccc|ccccc|}
        \toprule
        \specialrule{\heavyrulewidth}{\abovetopsep}{\belowbottomsep}
              \cline{1-11}
           & \multicolumn{4}{c}{\bf BRMSE, $\boldsymbol{\Tilde{\alpha}}$\bf -trimmed mean }&&\multicolumn{4}{c}{\bf BRMSE, $\boldsymbol{\Tilde{\alpha}}$\bf-Winsorized mean }&\\ 
              \specialrule{\heavyrulewidth}{\abovetopsep}{\belowbottomsep}
         & \multicolumn{9}{c}{$\boldsymbol{\Tilde{\alpha} = 0.01}$}& \\
         \specialrule{\heavyrulewidth}{\abovetopsep}{\belowbottomsep}
           \bf Variables&$\boldsymbol{b = 25}$& $\boldsymbol{b = 50}$ & $\boldsymbol{b =75}$&$\boldsymbol{b =100}$& $\boldsymbol{b = 200}$ &$\boldsymbol{b = 25}$& $\boldsymbol{b = 50}$ & $\boldsymbol{b =75}$&$\boldsymbol{b =100}$& $\boldsymbol{b = 200}$\\
        \midrule
        \specialrule{\heavyrulewidth}{\abovetopsep}{\belowbottomsep}
        age & 6.9025&5.9225&4.5379&4.4733&4.4385& 3.4793&3.6082 & 5.6003& 5.0576 & 5.0017 \\
        educ & 14.9525& 14.5061&14.6298 &14.1398&14.5784  & 14.7756 & 14.2998 & 14.0766 & 14.0379 & 14.0019 \\
        exper & 23.1121 & 23.9217 &23.0052& 7.6685 &7.1951 &  23.1546 & 23.0288 & 23.0853 & 23.0555 & 23.0422\\
        expersq & 7.0653 & 7.2846 & 7.1502 & 7.0685&7.0251 &8.3199 & 8.0245 & 7.2879 & 7.2718 & 7.1265\\
        kidslt6 & 8.8129 & 8.9512 &8.1709&8.1342&8.0552  &7.9862 & 8.6181 & 9.2364 & 9.1216 & 8.7599 \\
        kidsge6 & 7.5438 & 8.2794&7.6612&7.0774&7.0020 &  9.5989 & 8.5860 & 8.4013 & 8.3491 & 8.1709\\
        nwifeinc & 3.8540 & 8.0073&6.7732&6.9257&6.4667 &  4.8807 & 5.0512 & 5.4469 & 4.5340 & 4.4217\\
        residual & 7.4851&7.9285&6.6482&6.2391&6.1497  &  6.6852 & 8.5323 & 8.3165 & 6.9992 & 6.4343\\
        \specialrule{\heavyrulewidth}{\abovetopsep}{\belowbottomsep}
         & \multicolumn{9}{c}{$\boldsymbol{\Tilde{\alpha} = 0.02}$}& \\
        \midrule
        age & 4.9219&4.1721&2.9882&4.1762&4.0887& 4.5430 &4.2311&7.2968 &7.0489 &5.9689 \\
        educ & 15.6677&15.3316&15.1799&14.9809&14.8205 & 9.6751 & 6.7921 & 6.4150 &5.4424 &4.9195 \\
        exper & 23.5777&24.3622&23.8121&23.3626&23.0762 &7.7911 &5.5320 & 6.0543 &6.3188 &4.4880\\
        expersq & 6.8877&6.4023&6.4418&6.3116&6.3047& 10.0748 &3.7419 &6.0543 &6.0088 &4.4880 \\
        kidslt6 &  8.9114&9.0661&8.7789&8.5657&8.0430&3.8884 & 5.9683 &5.3010 &5.2230 & 5.0822\\
        kidsge6 &  8.2963&6.0569&6.3264&6.0951&6.0325 & 12.4818&5.3319 &5.1951  &4.6089 &4.5068\\
        nwifeinc & 4.5687&7.1909&5.1988&5.0937&5.0268 &13.4388 &6.0679& 5.7135 & 5.4986 & 4.6289\\
        residual &  7.8353&8.4133&7.4161&7.1263&7.0406&10.0395 &10.7615 &10.1960 &5.1033 &4.2763\\
         \specialrule{\heavyrulewidth}{\abovetopsep}{\belowbottomsep}
         & \multicolumn{9}{c}{$\boldsymbol{\Tilde{\alpha}= 0.20}$}& \\
        \midrule
 age&6.7119&8.8537&7.8333&4.0662&4.0032&4.9207&6.4777&4.7017&4.5382&3.9375\\
  educ &12.4243&12.9229&12.7528&12.0157&11.7030&11.0408&12.2314&11.4508&11.3584&11.2935\\ 
   exper &21.5681&21.7707&21.2526&21.1134&19.9628&19.8669&20.6090&20.4094&20.3946&20.1240\\
    expersq &8.4187&8.2567&7.3494&5.3215&5.9297&8.1333&6.3889&7.5108&7.4375&7.3153 \\
     kidslt6 &11.7281
     &11.7853&11.5799&8.0941&7.6841&9.2211&9.1069&8.2365&7.8665&6.8852\\
      kidsge6 &8.3459&7.3279&6.1896&5.7441&4.7452&4.4740&3.8739&7.6806&5.3986&4.8871 \\
        nwifeinc &8.4402&8.5644&8.4793&8.4723&4.9742&4.7178&4.8311&6.3046&5.8445&5.4379\\
         residual &10.0431&8.1517&6.5809&6.1770&5.4622&6.6833&8.6092&8.3502&6.8890&6.3961 \\
        \bottomrule  
        \specialrule{\heavyrulewidth}{\abovetopsep}{\belowbottomsep}  
    \end{tabular}}
    \label{table:2r}
\end{table}

The results from Tables \ref{table:3r} and \ref{table:2r} reveal a mixed scenario regarding the performance of parameter estimates, as indicated by BRMSE. When considering a small bootstrap sample size, BRMSE  shows improvement for certain parameters like  ``educ,'' but it increases for other variables such as ``kidslt6'' and ``nwifeinc'' across all $\Tilde{\alpha}$ values. Overall, as expected, larger bootstrap sample sizes tend to yield more reliable estimates, with BRMSE  decreasing as sample size increases, indicating better precision. The influence of $\Tilde{\alpha}$ values on BRMSE  is evident, with $\Tilde{\alpha} = 0.20$ and $\Tilde{\alpha} = 0.02$ generally resulting in lower BRMSE  for larger bootstrap samples compared to $\Tilde{\alpha} = 0.01$. The data suggests that increasing bootstrap replications generally enhances the performance of parameter estimates in terms of BRMSE. Moreover, the performance of estimators improves with higher trimming proportions, implying the presence of outliers or influential observations in the dataset.

\section{Conclusions} \label{Conclusion}
In this chapter, we proposed the semi-parametric Tobit model with endogenous regressors based on the parametric control function approach. We develop inference for the weakly dependent data for L-estimators. The estimation procedure is based on two-stage procedures, and the asymptotic properties of the L-estimators of the unknown parameters have been thoroughly studied. The proposed estimator performs well for various simulated data and real data. 

Recently, Dhar and Shalabh (\cite{Dharsh}) proposed a goodness of fit statistic named the GIVE statistic for the instrumental variables regression. However, they neither considered dependent errors nor they related instrumental variable regression with the concept of endogeneity. In the same spirit of their approach, one may propose a goodness-of-fit statistic for the model considered here. Besides, one may also carry out test like $H_{0}: {\boldsymbol{\beta}} = {\boldsymbol{\beta}}_{0}$ against $H_{1} : {\boldsymbol{\beta}} \neq  {\boldsymbol{\beta}}_{0}$ using the L-estimator proposed here. Generally speaking, the applications of the proposed L-estimator may be of interest to future research.

As we mentioned earlier, we have done thorough simulation studies for finite samples of different sizes. However, it does not theoretically justify the performance of the proposed estimator when the limiting sample size is infinite. This work can be carried out using the assertion in Corollary \ref{cor1}, which is left as future work. 

\noindent {\bf Data Availability Statement:}
The data that support the findings of this study are openly available in ``CRLM-Real-data" at \url{https://github.com/swati-1602/CRLM-Real-data.git}, reference:  ``Data.txt".
\section{Appendix}\label{appendix}
We begin by introducing definitions that have been used throughout the article.

\subsection{Definition} 
\begin{d1}[\bf{Stochastically equicontinuous:}] (see, \cite{vanderVaart1996}) Let~$\{\mathscr{F}_{n}(\boldsymbol{\beta}):n\geq 1\}$ be a family of random functions defined from $\mathscr{B} \rightarrow \mathbb {R}$, where 
$\mathscr{B}$ is any normed metric space. Then 
$\displaystyle \{\mathscr{F}_{n}\}$ is stochastically equicontinuous on $\mathscr{B} $ if $\forall \;\epsilon^{*} >0$,
$ \exists\;\delta>0$ such that 
\begin{equation}
    \displaystyle \limsup _{n \to \infty } \mathbb{P} \left(\sup _{\boldsymbol{\beta} \in \mathscr{B} }\sup _{\boldsymbol{\beta} '\in \mathcal{B}(\boldsymbol{\beta} ,\delta )}|\mathscr{F}_{n}(\boldsymbol{\beta} ')-\mathscr{F}_{n}(\boldsymbol{\beta} )|>\epsilon^{*} \right)< \epsilon^{*}.
\end{equation}
Here $\mathcal{B}(\boldsymbol{\beta} ,\delta )=\{\boldsymbol{\beta} \in \mathcal{B}: \|\boldsymbol{\beta}'-\boldsymbol{\beta}\|\leq \delta\}.$
\end{d1}

\subsection{B.1  Weak Consistency of first-stage estimator \texorpdfstring{$\Hat{\boldsymbol{\delta}}_{n}$}{}}\label{B.1}
We use the following theorem to establish the weak consistency of the first-stage estimator $\hat{\boldsymbol{\delta}}_{n}.$

\begin{t1}[\bf{Theorem 14.1 of \cite{kosorok2008}}] \label{th:7.2.1} Let $\mathcal{Q}_{n}$ and $\mathcal{Q}$ be a stochastic process by a metric space $H$ such that $\mathcal{Q}_{n}\xrightarrow{p}\mathcal{Q}$ in $\ell^{\infty}(H)$ for a every compact $H \subset \Theta.$ Suppose also that almost all sample paths $\theta \mapsto \mathcal{Q}(\theta)$ are upper semi-continuous and possess a unique maximum at a (random) point $\theta_{0},$ which as a random map in $\Theta$ is tight. If the sequence $\hat{\theta}_{n}$ is uniformly tight
and satisfies $\mathcal{Q}_{n}(\hat{\theta}_{n})\geq \sup_{\theta\in \Theta}\mathcal{Q}_{n}(\theta)-o_{p}(1),$ then $\hat{\theta}_{n}\xrightarrow{p}\theta_{0}$ in $\Theta.$
\end{t1}
\vspace{0.2cm}
\begin{l1}\label{l:7.2.1} Under the assumptions \ref{a:2.1}-\ref{a:2.2} and  \ref{a:3.3.1}, and \ref{a:3.3.2}, we have $\|\hat{\boldsymbol{\delta}}_{n}-\boldsymbol{\delta}_{0}\|_{\mathbb{R}^{p+1}} \rightarrow 0$ in probability as $n\rightarrow \infty.$
\end{l1}

\begin{proof}[\bf{Proof of the Lemma~\ref{l:7.2.1}}]
The minimization of $\mathscr{T}_{n}(\boldsymbol{\delta})$ with respect to $\boldsymbol{\delta}$ is equivalent to that of $\mathscr{\overline{T}}_{n}(\boldsymbol{\delta})$
 defined as follows:
 \begin{equation}
 \begin{split}
     \mathscr{\overline{T}}_{n}(\boldsymbol{\delta})=\mathscr{T}_{n}(\boldsymbol{\delta})-\mathscr{T}_{n}(\boldsymbol{\delta}_{0})
     &=\mathbb{E}_{n}\big[\big(w-\boldsymbol{z}^\top\boldsymbol{\delta}\big)^2-\big(w-\boldsymbol{z}^\top\boldsymbol{\delta}_{0}\big)^2\big]\\
     &= \mathbb{E}_{n}[\ell(\boldsymbol{\delta})].
     \end{split}
 \end{equation}
 Then the population counterpart is defined as follows:
 \begin{equation}
     \mathscr{\overline{T}}(\boldsymbol{\delta})=\mathbb{E}\big[\big(w-\boldsymbol{z}^\top\boldsymbol{\delta}\big)^2-\big(w-\boldsymbol{z}^\top\boldsymbol{\delta}_{0}\big)^2\big]=\mathbb{E}[\ell(\boldsymbol{\delta})]
 \end{equation}
 We need to verify that $\{\ell(\boldsymbol{\delta}):\boldsymbol{\delta}\in \Delta\}=\{\big(w-\boldsymbol{z}^\top\boldsymbol{\delta}\big)^2-\big(w-\boldsymbol{z}^\top\boldsymbol{\delta}_{0}\big)^2:\boldsymbol{\delta}\in \Delta\}$ is Glivenko-Cantelli, that is $$\underset{\boldsymbol{\delta}\in \Delta}{\sup}\big|\mathscr{\overline{T}}_{n}(\boldsymbol{\delta})-\mathscr{\overline{T}}(\boldsymbol{\delta})\big|\xrightarrow{p} 0$$ as $n$ goes to $\infty.$ It is easy to verify that  both $\mathscr{\overline{T}}(\boldsymbol{\delta})$ and $\mathscr{\overline{T}}(\boldsymbol{\delta})$ are continuous in $\boldsymbol{\delta}.$
Next, we show that the pointwise convergence, we have 
$$|\mathscr{\overline{T}}(\boldsymbol{\delta})|\leq \mathbb{E}\big[\|\boldsymbol{z}_{i}\|^2\|\boldsymbol{\delta}-\boldsymbol{\delta}_{0}\|^2+2.|\vartheta_{i}|\|\boldsymbol{z}_{i}\|\|\boldsymbol{\delta}-\boldsymbol{\delta}_{0}\|\big] < \infty$$
By using the assumptions \ref{a:3.3.1}, \ref{a:3.3.2} and the ergodic theorem of $\alpha$-mixing series (see, Theorem~3.34 in \cite{white1984}), we have $\mathscr{\overline{T}}_{n}(\boldsymbol{\delta})\xrightarrow{p} \mathscr{\overline{T}}(\boldsymbol{\delta})$ pointwise in $\boldsymbol{\delta}.$  To show the stochastic equicontinuity, consider the two elements, $\boldsymbol{\Tilde{\delta}}$ and $\boldsymbol{\delta}$ in $\Delta.$ Then, we have
\begin{equation}
\big|\mathscr{\overline{T}}_{n}(\boldsymbol{\Tilde{\delta}})-\mathscr{\overline{T}}_{n}(\boldsymbol{\delta})\big|\leq \|\boldsymbol{\Tilde{\delta}}-\boldsymbol{\delta}\|\{\|\boldsymbol{z}_{i}\|^2\|\boldsymbol{\Tilde{\delta}}+\boldsymbol{\delta}-2.\boldsymbol{\delta}_{0}\|+2|\vartheta_{i}|\|\boldsymbol{z}_{i}\|\}
\end{equation}
Thus, by the assumptions \ref{a:2.1}, \ref{a:3.3.1}, and \ref{a:3.3.2} $\mathbb{E}\|\boldsymbol{z}_{i}\|,$ $\mathbb{E}|\vartheta_{i}|$, and $\|\boldsymbol{\delta}\|$ are bounded. Then, there exists a positive constant $\mathbb{C}_{n}=O_{p}(1),$ such that 
\begin{equation}
    \big|\mathscr{\overline{T}}_{n}(\boldsymbol{\Tilde{\delta}})-\mathscr{\overline{T}}_{n}(\boldsymbol{\delta})\big|\leq \mathbb{C}_{n}\|\boldsymbol{\Tilde{\delta}}-\boldsymbol{\delta}\|
\end{equation}
for all $\boldsymbol{\Tilde{\delta}}$ and $\boldsymbol{\delta}$ in $\Delta.$ Hence, the empirical process $\boldsymbol{\delta}\mapsto \mathscr{\overline{T}}_{n}$ is stochastically equicontinuous, which implies that uniform convergence of $\mathscr{\overline{T}}_{n}(\boldsymbol{\delta})$ to $\mathscr{\overline{T}}(\boldsymbol{\delta})$ for all $\boldsymbol{\delta} \in \Delta$ by Theorem~1 in \cite{Andrews92}.
Next, we show that identification of the parameter $\boldsymbol{\delta}_{0}.$
\begin{equation}
\begin{split}
    \mathbb{E}[(w_{i}-\boldsymbol{z}^\top_{i}\boldsymbol{\delta})^2]&= \mathbb{E}[(w_{i}-\boldsymbol{z}^\top_{i}\boldsymbol{\delta}_{0}+\boldsymbol{z}^\top_{i}(\boldsymbol{\delta}_{0}-\boldsymbol{\delta}))^2]\\
    &= \mathbb{E}[(w_{i}-\boldsymbol{z}^\top_{i}\boldsymbol{\delta}_{0})^2]+2\mathbb{E}[(w_{i}-\boldsymbol{z}^\top_{i}\boldsymbol{\delta}_{0})\boldsymbol{z}^\top_{i}(\boldsymbol{\delta}_{0}-\boldsymbol{\delta})]+\mathbb{E}[(\boldsymbol{z}^\top_{i}(\boldsymbol{\delta}_{0}-\boldsymbol{\delta}))^2]\\
    &\geq\mathbb{E}[(w_{i}-\boldsymbol{z}^\top_{i}\boldsymbol{\delta}_{0})^2]
    \end{split}
\end{equation}
This inequality will be strictly positive provided that $\boldsymbol{z}^\top_{i}\boldsymbol{\beta}\neq \boldsymbol{z}^\top_{i}\boldsymbol{\beta}_{0}.$ This condition is satisfied under the non-singularity of $\mathbb{E}[\boldsymbol{z}\boldsymbol{z}^\top],$ as stated in Assumption~\ref{a:3.3.2}, which ensures the uniqueness of the minimum of  $\overline{\mathscr{T}}_{n}(\boldsymbol{\delta})$ at $\boldsymbol{\delta}_{0}.$
Similar arguments in \cite{kosorok2008}, we can prove the existence of $\hat{\boldsymbol{\delta}}_{n}$ i.e. $\|\hat{\boldsymbol{\delta}}_{n}\| = O_{p}(1),$ Then, by the $\arg\max$ Theorem~\ref{th:7.2.1} it's complete the proof.
 \end{proof}
 \subsection{B.2 Uniform Consistency of \texorpdfstring{$\Hat{\boldsymbol{\beta}}_{n}(\uptau)$}{}}\label{B.2}
 \begin{l1}\label{l:7.2.2} Under the assumptions \ref{a:2.1}, \ref{a:3.3.1}-\ref{a:3.3.4}, $\mathscr{S}_{n}\big(\boldsymbol{\beta},\boldsymbol{\delta},\uptau\big)$ has a unique minimizer at $\boldsymbol{\beta}_{0}(\uptau)$ for all $\uptau \in \mathcal{T}.$
 \end{l1}
 \begin{proof}[\bf{Proof of the Lemma~\ref{l:7.2.2}}]
 To show identification of
  $\boldsymbol{\beta}_{0}$ note that $\mathscr{S}(\boldsymbol{\beta},\boldsymbol{\delta},\uptau)$ is uniformly dominated by an integrable function, so that by the Lebesgue dominated convergence theorem
  \begin{equation*}
  \begin{split}
\frac{\partial \mathbb{E}\big[\mathscr{S}(\boldsymbol{\beta},\boldsymbol{\delta},\uptau)]}{\partial \boldsymbol{\beta}}&=-2\mathbb{E}\big[\mathds{1}(\boldsymbol{x}^\top\boldsymbol{\beta}>0)\big\{\uptau-\mathds{1}\big(\varepsilon+\boldsymbol{x}^\top\boldsymbol{\beta}_{0}<\boldsymbol{x}^\top\boldsymbol{\beta}\big)\big\}\boldsymbol{x}\big]\\
&=-2\mathbb{E}\big[\mathds{1}(\boldsymbol{x}^\top\boldsymbol{\beta}>0)\big\{\uptau-\mathds{1}\big(\varepsilon-\boldsymbol{x}^\top\boldsymbol{\delta}^{*}<\mathcal{Q}_{\varepsilon}(\uptau)\big)\big\}\boldsymbol{x}\big],\\
\end{split}
  \end{equation*}
  where $\boldsymbol{\delta}^{*}=\boldsymbol{\beta}-\boldsymbol{\beta}_{0}(\uptau)$ and $\boldsymbol{\beta}_{0} (\uptau)$ is the same as defined in \eqref{eq:2.6}.
Now, using the iterative formula for expectation, we obtain
\begin{equation*}
  \begin{split}
\frac{\partial \mathbb{E}\big[\mathscr{S}(\boldsymbol{\beta},\boldsymbol{\delta},\uptau)]}{\partial \boldsymbol{\beta}}&=-2\mathbb{E}_{\mathcal{F}_{i-1}}\big[\mathds{1}(\boldsymbol{x}^\top\boldsymbol{\beta}>0)\big\{\uptau-\mathbb{P}\big(\varepsilon<\mathcal{Q}_{\varepsilon_{\uptau}}(\uptau)+\boldsymbol{x}^\top\boldsymbol{\delta}^{*}|\mathcal{F}_{i-1}\big)\big\}\boldsymbol{x}\big]\\
\end{split}
  \end{equation*}
By Assumption~\ref{a:3.3.3}, which guarantees the existence of a unique quantile, it follows that
\begin{align*}
\boldsymbol{x}^\top\boldsymbol{\beta}_{0}(\uptau)>\boldsymbol{x}^\top\boldsymbol{\beta} & : \quad \frac{\partial \mathbb{E}\left[ \mathscr{S}(\boldsymbol{\beta},\boldsymbol{\delta}, \uptau) \right]}{\partial \boldsymbol{\beta}} < 0 \\
\boldsymbol{x}^\top\boldsymbol{\beta}_{0}(\uptau)<\boldsymbol{x}^\top\boldsymbol{\beta} & : \quad \frac{\partial \mathbb{E}\left[ \mathscr{S}(\boldsymbol{\beta},\boldsymbol{\delta}, \uptau) \right]}{\partial \boldsymbol{\beta}} > 0 \\
\boldsymbol{x}^\top\boldsymbol{\beta}_{0}(\uptau)=\boldsymbol{x}^\top\boldsymbol{\beta} & : \quad \frac{\partial \mathbb{E}\left[ \mathscr{S}(\boldsymbol{\beta},\boldsymbol{\delta}, \uptau) \right]}{\partial \boldsymbol{\beta}} = 0.
\end{align*}
Therefore, it follows that the function $\mathscr{S}(\boldsymbol{\beta},\boldsymbol{\delta},\uptau)$ attains a minimum at $\boldsymbol{\beta}=\boldsymbol{\beta}_{0}(\uptau)$ or, equivalently, when $\boldsymbol{x}^\top\boldsymbol{\beta}_{0}(\uptau)=\boldsymbol{x}^\top\boldsymbol{\beta}$ with probability one. Next, to establish the uniqueness of the parameter, $\boldsymbol{\beta}_{0}(\uptau),$ it is sufficient to demonstrate that $$\max\{0, \boldsymbol{x}^\top\boldsymbol{\beta}\} \neq \max\{0, \boldsymbol{x}^\top\boldsymbol{\beta}_{0}(\uptau)\}
\quad \text{whenever} \quad \boldsymbol{\beta} \neq \boldsymbol{\beta}_{0}(\uptau).$$ 
This can be shown by considering two distinct cases:

\noindent\textbf{Case 1:}
If $\mathds{1}\big(\boldsymbol{x}^\top\boldsymbol{\beta}>0\big)\neq \mathds{1}\big(\boldsymbol{x}^\top\boldsymbol{\beta}_{0}(\uptau)>0\big)$ then it implies that
$$\max\{0, \boldsymbol{x}^\top\boldsymbol{\beta}\}\neq \max\{0, \boldsymbol{x}^\top\boldsymbol{\beta}_{0}(\uptau)\}.$$ 

\noindent\textbf{Case 2:} If
$\mathds{1}\big(\boldsymbol{x}^\top\boldsymbol{\beta}>0\big)= \mathds{1}\big(\boldsymbol{x}^\top\boldsymbol{\beta}_{0}(\uptau)>0\big)$ so that 
\begin{equation*}
 \max\{0, \boldsymbol{x}^\top\boldsymbol{\beta}\}- \max\{0, \boldsymbol{x}^\top\boldsymbol{\beta}_{0}(\uptau)\}=\mathds{1}\big(\boldsymbol{x}^\top\boldsymbol{\beta}_{0}(\uptau)>0\big)
\boldsymbol{x}^\top\boldsymbol{\delta}^{*}\neq 0.
\end{equation*}
In both cases, we find that the two expressions differ whenever $\boldsymbol{\beta}\neq\boldsymbol{\beta}_{0}(\uptau).$ Therefore, it is sufficient to assume that Assumption \ref{a:3.3.4} ensures that
$\mathscr{S}(\boldsymbol{\beta}, \boldsymbol{\delta}, \uptau)$ attains a unique minimum at $\boldsymbol{\beta}_{0}(\uptau)$ for every $\uptau \in \mathcal{T}$.
\end{proof}
 
 \begin{l1}\label{l:7.2.3} Under the assumptions \ref{a:2.1}-\ref{a:2.2} and \ref{a:3.3.1}-\ref{a:3.3.4}, $\mathscr{S}_{n}\big(\boldsymbol{\beta},\hat{\boldsymbol{\delta}}_{n},\uptau\big)$ converges uniformly to $\mathscr{S}\big(\boldsymbol{\beta},\boldsymbol{\delta},\uptau\big)$ for all $\uptau \in \mathcal{T}.$
 \end{l1}
 \begin{proof}[\bf{Proof of the Lemma~\ref{l:7.2.3}}]
 The normalized objective function is defined as
 \begin{equation}
 \begin{split}
\mathscr{S}_{n}\big(\boldsymbol{\beta},\hat{\boldsymbol{\delta}}_{n},\uptau\big)&=\mathscr{Q}_{n}\big(\boldsymbol{\beta},\hat{\boldsymbol{\delta}}_{n},\uptau\big)-\mathscr{Q}_{n}\big(\boldsymbol{\beta}_{0}(\uptau),\boldsymbol{\delta}_{0},\uptau\big),\\
     &=\mathbb{E}_{n}\big[\rho_{\uptau}\big(y-\max\{0, \hat{\boldsymbol{x}}^\top\boldsymbol{\beta}\}\big)-\rho_{\uptau}\big(y-\max\{0, \boldsymbol{x}^\top\boldsymbol{\beta}_{0}(\uptau)\}\big)\big]
     \end{split}
 \end{equation}
 where $\hat{\boldsymbol{x}}$ is the same as defined in \eqref{eq:2.5}. First, we aim to show that $\mathscr{S}_{n}\big(\boldsymbol{\beta},\hat{\boldsymbol{\delta}}_{n},\uptau\big)$ and $\mathscr{S}_{n}\big(\boldsymbol{\beta},\boldsymbol{\delta},\uptau\big)$ have the same probability limit.  The function  $\mathscr{S}_{n}\big(\boldsymbol{\beta},\boldsymbol{\delta}, \uptau\big)$ is  given by
 \begin{equation}
 \begin{split}
 \mathscr{S}_{n}\big(\boldsymbol{\beta},\boldsymbol{\delta}, \uptau\big)&=\mathscr{Q}_{n}\big(\boldsymbol{\beta},\boldsymbol{\delta},\uptau\big)-\mathscr{Q}_{n}\big(\boldsymbol{\beta}_{0}(\uptau),\boldsymbol{\delta}_{0},\uptau\big)\\
 &=\mathbb{E}_{n}\big[\rho_{\uptau}\big(y-\max\{0, \boldsymbol{x}^\top\boldsymbol{\beta}\}\big)-\rho_{\uptau}\big(y-\max\{0, \boldsymbol{x}^\top\boldsymbol{\beta}_{0}(\uptau)\}\big)\big]
 \end{split}
  \end{equation}
 Applying the Knight's identity 
  $\rho_{\uptau}\big(a-b\big)-\rho_{\uptau}\big(a\big)=-b\big(\uptau -\mathds{1}\{a<0\}\big)+\int_{0}^{b}\{\mathds{1}\big(a\leq s\big)-\mathds{1}\big(a\leq 0\big)\} ds$, then we obtain 
  \begin{equation}
       \big|\mathscr{S}_{n}\big(\boldsymbol{\beta},\hat{\boldsymbol{\delta}}_{n},\uptau\big)-\mathscr{S}_{n}\big(\boldsymbol{\beta},\boldsymbol{\delta}, \uptau\big)\big|\leq \mathbb{E}_{n}\big|\big(\hat{\boldsymbol{x}}-\boldsymbol{x}\big)\boldsymbol{\beta}\big|\leq \mathbb{E}_{n}\|\hat{\boldsymbol{x}}-\boldsymbol{x}\|\|\boldsymbol{\beta}\|.
  \end{equation}
  Since $\hat{\boldsymbol{\delta}}_{n}$ is consistent estimator of $\boldsymbol{\delta}_{0},$ it follows that $(\hat{\boldsymbol{x}}_{i}-\boldsymbol{x}_{i})=(e_{i}-\vartheta_{i})=\boldsymbol{z}_{i}^\top(\hat{\boldsymbol{\delta}}_{n}-\boldsymbol{\delta}_{0})\xrightarrow{p} 0.$ Hence, we conclude that 
  \begin{equation}
       \big|\mathscr{S}_{n}\big(\boldsymbol{\beta},\hat{\boldsymbol{\delta}}_{n},\uptau\big)-\mathscr{S}_{n}\big(\boldsymbol{\beta},\boldsymbol{\delta}, \uptau\big)\big|\xrightarrow{p} 0.
  \end{equation}
  Next, we establish the pointwise convergence by observing that
  \begin{equation}
\mathscr{S}\big(\boldsymbol{\beta},\boldsymbol{\delta},\uptau\big)=\mathbb{E}\big[\rho_{\uptau}\big(y-\max\big\{0,\boldsymbol{x}^\top\boldsymbol{\beta}\big\}\big)-\rho_{\uptau}\big(y-\max\big\{0,\boldsymbol{x}^\top\boldsymbol{\beta}_{0}(\uptau)\big\}\big)\big].
  \end{equation}
  Clearly, the function $\mathscr{S}(\boldsymbol{\beta}, \boldsymbol{\delta}, \uptau)$ is continuous in $\boldsymbol{\beta}$ and also continuous in $\uptau$. Thus, $\mathscr{S}(\boldsymbol{\beta}, \boldsymbol{\delta}, \uptau)$ is continuous over all $\mathscr{B}\times \mathcal{T}.$ Then using the Holder inequality and the assumptions \ref{a:3.3.1} and \ref{a:3.3.4}, we have 
  \begin{equation}
      |\mathscr{S}\big(\boldsymbol{\beta},\boldsymbol{\delta},\uptau\big)|\leq \mathbb{E}\big[\big|\boldsymbol{x}^\top\big(\boldsymbol{\beta}-\boldsymbol{\beta}_{0}(\uptau)\big)\big|\big]\leq \big(\mathbb{E}\big[\big\|\boldsymbol{x}\big\|\big]^{r'}\big)^{\frac{1}{r^{'}}}\big(\big\|\big(\boldsymbol{\beta}-\boldsymbol{\beta}_{0}(\uptau)\big)\big\|^{\tilde{s}}\big)^{\frac{1}{\tilde{s}}}<\infty,
  \end{equation}
   where $r^{'}$ is the same defined in \ref{eq:3.2} and $\tilde{s}=\frac{r^{'}}{r^{'}-1}.$ Then by the ergodic theorem for $\alpha$-mixing sequences (See,Theorem-3.34 in \cite{white1984}), we obtain $\mathscr{S}_{n}\big(\boldsymbol{\beta},\boldsymbol{\delta},\uptau\big)= \mathscr{S}(\boldsymbol{\beta},\boldsymbol{\delta},\uptau)+o_{p}(1)$ in $\big(\boldsymbol{\beta},\uptau\big).$ Moreover, for all $\big(\boldsymbol{\beta}^{'},\uptau^{'}\big)$ and $\big(\boldsymbol{\beta}^{"},\uptau^{"}\big),$ we have
   \begin{equation}
   \begin{split}
       \big|\mathscr{S}_{n}\big(\boldsymbol{\beta}^{'},\boldsymbol{\delta},\uptau^{'}\big)-\mathscr{S}_{n}\big(\boldsymbol{\beta}^{"},\boldsymbol{\delta},\uptau^{"}\big)\big|&\leq \big|\mathscr{S}_{n}\big(\boldsymbol{\beta}^{'},\boldsymbol{\delta},\uptau^{'}\big)-\mathscr{S}_{n}\big(\boldsymbol{\beta}^{'},\boldsymbol{\delta},\uptau^{"}\big)\big|
       +\big|\mathscr{S}_{n}\big(\boldsymbol{\beta}^{'},\boldsymbol{\delta},\uptau^{"}\big)\\
       &-\mathscr{S}_{n}\big(\boldsymbol{\beta}^{"},\boldsymbol{\delta},\uptau^{"}\big)\big|\\
       &\leq 2.\mathbb{E}_{n}[\|\boldsymbol{x}\|]\|\boldsymbol{\beta}\||\uptau^{'}-\uptau^{"}|+2.\mathbb{E}_{n}[\|\boldsymbol{x}\|]\|\boldsymbol{\beta}^{'}-\boldsymbol{\beta}^{"}\|\\
       &\leq C_{1,n}|\uptau^{'}-\uptau^{"}|+C_{2,n}\|\boldsymbol{\beta}^{'}-\boldsymbol{\beta}^{"}\|,
        \end{split}
   \end{equation}
   where $C_{1,n}=O_{p}(1)$ and $C_{2,n}=O_{p}(1)$ from the assumptions \ref{a:3.3.1} and \ref{a:3.3.4}. Hence the process $(\boldsymbol{\beta}, \uptau) \mapsto \mathscr{S}_{n}(\boldsymbol{\beta},\boldsymbol{\delta}, \uptau)$ is stochastically equicontinuous. Now using Theorem 1 by \cite{Andrews92}, it follows the uniform convergence.
   $$\underset{(\boldsymbol{\beta},\uptau)\in \mathscr{B}\times\mathcal{T}}{\sup}\big|\mathscr{S}_{n}(\boldsymbol{\beta},\boldsymbol{\delta},\uptau)-\mathscr{S}(\boldsymbol{\beta},\boldsymbol{\delta},\uptau)\big|\xrightarrow{p} 0.$$
 \end{proof}
 \begin{l1}[\bf{Lemma B.1 of \cite{Hansen}}]\label{l:7.2.4}
Suppose that uniformly in $\uptau$ in a compact set $\mathcal{T}$ and for a compact set $\Theta$
\begin{enumerate}
    \item[(i)] $\hat{\theta}_{n}(\uptau) $ is such that 
    $Q_{n}(\hat{\theta}_{n}(\uptau) \mid \uptau) \leq \inf_{\theta\in \Theta} Q_{n}(\theta \mid \uptau) + \alpha^{*}_n, \; \alpha^{*}_n \to 0,\;  \hat{\theta}_{n}(\uptau) \in \Theta \text{ w.p.} \to 1,$
   \item[(ii)] $\theta_{0}(\uptau) := \arg\inf_{\theta \in \Theta} Q_{\infty}(\theta \mid \uptau)$ is a uniquely defined continuous process in $\ell^\infty(\mathcal{T})$,
    \item[(iii)] $Q_{n}(\cdot \mid \cdot) \xrightarrow{p} Q_{\infty}(\cdot \mid \cdot)$ in $l^\infty(\Theta \times \mathcal{T})$, where $Q_{\infty}(\cdot \mid \cdot)$ is continuous.
\end{enumerate}
\noindent
Then, 
$$\hat{\theta}_n(\cdot) \xrightarrow{p} \theta_{0}(\cdot)  \quad \text{in } l^\infty(\mathcal{T}).$$
\end{l1}
\begin{proof}[\bf{Proof of the Theorem~\ref{th:3.3.1}}]
To establish the theorem, we verify that all the conditions of Lemma~\ref{l:7.2.4} are satisfied. The first condition is straightforward, and the existence of $\hat{\boldsymbol{\beta}}_{n}(\uptau)$
  can be established using an argument similar to that in \cite{kosorok2008}, under which $\hat{\boldsymbol{\beta}}_{n}(\uptau)=O_{p}(1)$. The second condition is fulfilled as demonstrated in Lemma~\ref{l:7.2.2}, while the third condition is verified through Lemma~\ref{l:7.2.3}. Since all the required conditions of Lemma~\ref{l:7.2.4} hold, we conclude that
$$\underset{\uptau \in \mathcal{T}}{\sup}\|\hat{\boldsymbol{\beta}}_{n}(\uptau)-\boldsymbol{\beta}_{0}(\uptau)\|\xrightarrow{p} 0.$$   
    \end{proof}

 \subsection{B.3 Weak Consistency of \texorpdfstring{$\mathcal{L}_{n}(.)$}{}}\label{B.3}
 \begin{proof}[\bf{Proof of the Theorem~\ref{th:3.3.2}}]
 Consider the smooth L-estimator defined as 
  \begin{equation}
      \mathscr{L}_{n}= \int_{0}^{1}\Hat{\boldsymbol{\beta}}_{n}(\uptau)J_{1}(\uptau)\;d\uptau,
      \end{equation}
      with the corresponding population counter part is given by  $$\mathscr{L}_{0}= \int_{0}^{1}\boldsymbol{\beta}_{0}(\uptau)J_{1}(\uptau)\;d\uptau.$$ 
     We begin by examining the difference between the sample and population L-estimators
   \begin{equation}
      \|L_{n}-L_{0}\|=\Big\|\int_{0}^{1}(\hat{\boldsymbol{\beta}}_{n}(\uptau)-\boldsymbol{\beta}_{0}(\uptau))J_{1}(\uptau)\;d\uptau\Big\|.\notag
       \end{equation}
By the triangle inequality and properties of the norm, we obtain
    \begin{equation}
\|L_{n}-L_{0}\| \leq \int_{0}^{1} \big\|(\hat{\boldsymbol{\beta}}_{n}(\uptau)-\boldsymbol{\beta}_{0}(\uptau))\big\| |J_{1}(\uptau)|\;d\uptau 
 \end{equation}
 Since $J_{1}(\uptau)$ is a smooth function and hence bounded on any compact subinterval of $(0,1),$ we assume there exists a constant $m_{0}>0$ such that
 $$|J_{1}(\uptau)|\leq m_{0},\quad \text{for all} \quad \uptau \in \mathcal{T}\subset (0,1).$$
 To handle possible irregularities near the boundaries, we restrict the domain of the integration to the interval $\mathcal{T}=[\uptau_{0},1-\uptau_{0}],$ where $\uptau_{0}>0,$ and consider the limit $\uptau_{0}\rightarrow 0,$ Thus, we have
 \begin{align*}
&=\lim_{\uptau_{0}\rightarrow 0} \int_{\uptau_{0}}^{1-\uptau_{0}}\|(\hat{\boldsymbol{\beta}}_{n}(\uptau)-\boldsymbol{\beta}_{0}(\uptau))\| |J_{1}(\uptau)|\;d\uptau\\
&\leq \lim_{\uptau_{0}\rightarrow 0}\: \underset{\uptau \in \mathcal{T}}\sup \|\hat{\boldsymbol{\beta}}_{n}(\uptau)-\boldsymbol{\beta}_{0}(\uptau)\|\; m_{0}\tag{Using the Assumption \ref{a:3.3.5}}\\
& \xrightarrow{p}0,\tag{{B.1}}\label{2}
  \end{align*}
where the convergence in probability follows from the uniform consistency of $\Hat{\boldsymbol{\beta}}_{n}(\uptau)$ for all $\uptau \in \mathcal{T}\subset (0,1).$
Therefore, we conclude that $\|\mathscr{L}_{n}-\mathscr{L}_{0}\|\xrightarrow{p}0.$ 
  \end{proof}
\subsection{B.4 Asymptotic Normality of \texorpdfstring{$\sqrt{n}(\Hat{\boldsymbol{\beta}}_{n}(\uptau)-\boldsymbol{\beta}_{0}(\uptau))$}{}}\label{B.4}
To prove the asymptotic Normality of the two-stage estimator process 
\begin{equation}\label{eq:7.16}
\{(\Hat{\boldsymbol{\beta}}_{n}(\uptau),\hat{\boldsymbol{\delta}}_{n})^\top: \boldsymbol{\beta}\in \mathscr{B},\boldsymbol{\delta} \in \Delta, \uptau \in \mathcal{T}\subset(0,1)\},
\end{equation} the adopted methods are similar to the methods considered in  \cite{Angrist}, \cite{CHEN201830}, and \cite{Hansen}).
 
\begin{proof}[{\bf{Proof of Theorem} \ref{th:3.3.3}:}] We derive the asymptotic normality of a two-stage estimator, which is obtained by jointly solving the first-order conditions corresponding to both stages of the estimation procedure. In particular, the true parameter values, denoted by $\boldsymbol{\beta}_{0}(\uptau)$ and $\boldsymbol{\delta}_{0}$, are characterized as the unique solution to the following system of estimating equations that combine both the first- and second-stage conditions.
\begin{align}\label{eq:7.17}
\mathbb{E}[\Psi(\boldsymbol{u},\boldsymbol{\beta}_{0}(\uptau),\boldsymbol{\delta}_{0},\uptau)]=
          \begin{bmatrix}
        \mathbb{E}[\Psi_{2}(\varepsilon_{\uptau},\boldsymbol{x},\boldsymbol{\beta}_{0}(\uptau),\boldsymbol{\delta}_{0},\uptau)] \\   
        \mathbb{E}[\Psi_{1}(\vartheta,\boldsymbol{z},\boldsymbol{\delta}_{0})]
          \end{bmatrix} =\boldsymbol{0}, 
  \end{align}
where $\boldsymbol{u}=(y,\boldsymbol{x}^\top,\boldsymbol{z}^\top),$ $\boldsymbol{s}=(y,\boldsymbol{x}^\top,\varepsilon_{\uptau})$ and  $\Psi(\boldsymbol{u},\boldsymbol{\beta}(\uptau),\boldsymbol{\delta},\uptau)= \begin{pmatrix}
     \Psi_{2}(\boldsymbol{s},\boldsymbol{\beta},\boldsymbol{\delta},\uptau),\\
      \Psi_{1}(\vartheta,\boldsymbol{z},\boldsymbol{\delta})
\end{pmatrix}$ and the individual components are given by
\begin{equation}\label{eq:7.18}
\begin{split}
  \Psi_{2}(\boldsymbol{s},\boldsymbol{\beta},\boldsymbol{\delta},\uptau) &= \mathds{1}(\boldsymbol{x}^\top\boldsymbol{\beta}>0)\psi_{\uptau}(y-\boldsymbol{x}^\top\boldsymbol{\beta})\boldsymbol{x},\\
   \Psi_{1}(\vartheta,\boldsymbol{z},\boldsymbol{\delta})&=(w-\boldsymbol{z}^\top\boldsymbol{\delta})\boldsymbol{z}.
   \end{split}
\end{equation}
In addition, $\Psi_{2}(\boldsymbol{s},\boldsymbol{\beta},\boldsymbol{\delta},\uptau)$ is vector of function with dimension equal to $\boldsymbol{x}\in \mathbb{R}^{p+2},$ and $\Psi_{1}(\vartheta,\boldsymbol{z},\boldsymbol{\delta})$ has dimension of $\boldsymbol{z}\in \mathbb{R}^{p+1}.$ The function $\psi_{\uptau}(u)=\uptau-\mathds{1}(u<0).$
 The joint two-stage estimators $(\hat{\boldsymbol{\beta}}_{n}(\uptau),\hat{\boldsymbol{\delta}}_{n})^\top$ can be characterized as the joint solution of the following equation:
 \begin{equation}\label{eq:7.19} \frac{1}{n}\sum_{i=1}^{n}\Psi(\boldsymbol{u}_{i},\boldsymbol{\beta} ,\boldsymbol{\delta},\uptau)= \frac{1}{n}\sum\limits_{i=1}^{n}\begin{bmatrix}
 \Psi_{2}(\boldsymbol{s}_{i},\boldsymbol{\beta} ,\boldsymbol{\delta},\uptau)\\
 \Psi_{1}(\vartheta_{i},\boldsymbol{z}_{i},\boldsymbol{\delta})
 \end{bmatrix}=
     \begin{bmatrix}
       \frac{1}{n}\sum\limits_{i=1}^{n}\Psi_{2}(\boldsymbol{s}_{i},\boldsymbol{\beta} ,\boldsymbol{\delta},\uptau)\\
       \frac{1}{n}\sum\limits_{i=1}^{n}\Psi_{1}(\vartheta_{i},\boldsymbol{z}_{i},\boldsymbol{\delta})
     \end{bmatrix}=\boldsymbol{0}.
 \end{equation}
Let us define the joint two-stage estimator process as $\{(\hat{\boldsymbol{\beta}}_{n}(\uptau),\hat{\boldsymbol{\delta}}_{n})^\top:\boldsymbol{\beta}\in \mathscr{B}, \boldsymbol{\delta}\in \Delta, \uptau \in \mathcal{T}\subset (0,1)\}.$ 

Now, the proof of this theorem follows from the following five steps.

 \noindent{\bf{First step:}} We first establish that 
 \begin{equation}\label{eq:7.20}
    \underset{\uptau \in \mathcal{T}}{\sup} \|\sqrt{n}\mathbb{E}_{n}[\Psi(\boldsymbol{u},\hat{\boldsymbol{\beta}}_{n}(\uptau) ,\hat{\boldsymbol{\delta}}_{n},\uptau)]\|=o_{p}(1).
 \end{equation}
 To prove this, it suffices to show that
 \begin{equation}\label{eq:7.21}
     \begin{split}
         \|\sqrt{n}\mathbb{E}_{n}[\Psi_{1}(\vartheta,\boldsymbol{z},\hat{\boldsymbol{\delta}}_{n})]&=o_{p}(1)\\
         \underset{\uptau \in \mathcal{T}}{\sup}\|\sqrt{n}\mathbb{E}_{n}[\Psi_{2}(\boldsymbol{s},\hat{\boldsymbol{\beta}}_{n}(\uptau),\hat{\boldsymbol{\delta}}_{n},\uptau)]&=o_{p}(1).\\
     \end{split}
 \end{equation}
To analyze the second term, define $\Psi_{2,j}(.)$ as the $j^{th}$ component of vector $\Psi_{2}(.),$ and let $e_{j}\in \mathbb{R}^{p+2}$ be a unit vector with 1 in the $j^{th}$ position. Consider the directional derivative of the objective function $\mathscr{Q}_{n}(\boldsymbol{\beta},\hat{\boldsymbol{\delta}}_{n}, \uptau)$ in the direction $e_{j}\in \mathbb{R}^{p+2}$ for some $a>0,$ is given by
\begin{equation}\label{eq:7.22}
\begin{split}
   \mathscr{H}_{n,j}(\boldsymbol{\beta},\boldsymbol{\delta}, \uptau)&=\underset{a\to 0}{Lim}\frac{\mathscr{Q}_{n}(\boldsymbol{\beta}+ae_{j},\hat{\boldsymbol{\delta}}_{n}, \uptau)-\mathscr{Q}_{n}(\boldsymbol{\beta},\hat{\boldsymbol{\delta}}_{n}, \uptau)}{a}\\
   &=-\frac{1}{n}\sum_{i=1}^{n}\Psi_{2,j}(\boldsymbol{s}_{i},\boldsymbol{\beta}, \hat{\boldsymbol{\delta}}_{n},\uptau)-\frac{1}{n}\sum_{i=1}^{n}\mathds{1}(\boldsymbol{x}^\top_{i}\boldsymbol{\beta}>0, y_{i}=\boldsymbol{x}^\top_{i}\boldsymbol{\beta})x_{i,j}\psi_{\uptau}(x_{i,j})\\
   &+\frac{1}{n}\sum_{i=1}^{n}\mathds{1}(\boldsymbol{x}^\top_{i}\boldsymbol{\beta}=0)\{(1-\uptau)\mathds{1}(y_{i}>0,x_{i,j}>0)x_{i,j}-\uptau\mathds{1}(y_{i}=0,x_{i,j}>0)x_{i,j}\}
   \end{split}
\end{equation}
Since $\hat{\boldsymbol{\beta}}_{n}(\uptau)$ is minimizer of $\mathscr{Q}_{n}(\boldsymbol{\beta},\hat{\boldsymbol{\delta}}_{n},\uptau)$, note that for any $a>0$ and using the monotonicity of the $\mathscr{H}_{n,j}(\boldsymbol{\beta},\hat{\boldsymbol{\delta}}_{n}, \uptau)$ then we obtain
\begin{equation}\label{eq:7.23}
\begin{split}
    |\mathscr{H}_{n,j}(\hat{\boldsymbol{\beta}}_{n}(\uptau),\hat{\boldsymbol{\delta}}_{n},\uptau)|&\leq \mathscr{H}_{n,j}(\hat{\boldsymbol{\beta}}_{n}(\uptau)+ae_{j},\hat{\boldsymbol{\delta}}_{n},\uptau)-\mathscr{H}_{n,j}(\hat{\boldsymbol{\beta}}_{n}(\uptau)-ae_{j},\hat{\boldsymbol{\delta}}_{n},\uptau)\\
    &\leq \frac{1}{\sqrt{n}}\sum_{i=1}^{n}|x_{i,j}|[\mathds{1}(y_{i}=\boldsymbol{x}^\top_{i}\hat{\boldsymbol{\beta}}_{n}(\uptau))+\max\{\uptau, 1-\uptau\}\mathds{1}(\boldsymbol{x}^\top_{i}\hat{\boldsymbol{\beta}}_{n}(\uptau)=0)]\\
    &\leq\frac{\underset{i}{\max}|x_{i,j}|}{\sqrt{n}}\sum_{i=1}^{n}[\mathds{1}(y_{i}=\boldsymbol{x}^\top_{i}\hat{\boldsymbol{\beta}}_{n}(\uptau))+\max\{\uptau, 1-\uptau\}\mathds{1}(\boldsymbol{x}^\top_{i}\hat{\boldsymbol{\beta}}_{n}(\uptau)=0)]
\end{split}
\end{equation}
The sum on the right hand side of equation \eqref{eq:7.22} is finite with probability one for large $n$ under the assumption \ref{a:3.3.7}, and the uniform consistency of $\hat{\boldsymbol{\beta}}_{n}(\uptau)$ and $\hat{\boldsymbol{\delta}}_{n}.$ Moreover, for some $r^{'}>2,$ we can apply Jensen’s inequality 
\begin{equation}\label{eq:7.24}
\begin{split}
    \frac{\max |x_{i,j}|}{\sqrt{n}}=\Big[\max\Big(\frac{|x_{i,j}|}{\sqrt{n}}\Big)^{r'}\Big]^{\frac{1}{r'}}\leq (n)^{\frac{2-r^{'}}{2r^{'}}}\Big[\frac{1}{n}\sum_{i=1}^{n}|x_{i,j}|^{r'}\Big]^{\frac{1}{r'}}
    \end{split}
\end{equation}
By the strong law of large numbers for strongly mixing sequences (see Corollary 3.48 in \cite{white1984}), the term $\underset{i}{\max |x_{i,j}|}/\sqrt{n}$ converges to zero in probability
for some $r^{'}>2.$ Hence, we conclude that
$$\underset{\uptau \in \mathcal{T}}{\sup}\|\sqrt{n}\mathbb{E}_{n}(\Psi_{2}(\boldsymbol{s},\hat{\boldsymbol{\beta}}_{n}(\uptau),\hat{\boldsymbol{\delta}}_{n},\uptau)\|=o_{p}(1).$$
Similarly, it can be shown that $$\|\sqrt{n}\mathbb{E}_{n}\Psi_{1}(\vartheta,\boldsymbol{z}, \hat{\boldsymbol{\delta}}_{n})\|=o_{p}(1).$$ Combining these two results completes the proof. 

{\bf{Second step:}} Define the empirical process
\begin{equation}\label{eq:7.25}
    \mathbb{G}_{n}(\Psi(\boldsymbol{u},\boldsymbol{\beta} ,\boldsymbol{\delta},\uptau))=\frac{1}{\sqrt{n}}\sum_{i=1}^{n}\big\{\Psi(\boldsymbol{u}_{i},\boldsymbol{\beta} ,\boldsymbol{\delta},\uptau)-\mathbb{E}\big[\Psi(\boldsymbol{u}_{i},\boldsymbol{\beta} ,\boldsymbol{\delta},\uptau)\big]\big\},
\end{equation}
 where $(\boldsymbol{\beta},\boldsymbol{\delta},\uptau) \in \mathscr{B}\times\Delta\times \mathcal{T}$ is a compact set, $\mathbb{G}_{n}(\Psi(.))=(\mathbb{G}_{2,n}(\Psi_{2}(.)),\mathbb{G}_{1,n}(\Psi_{1}(.)))^\top$ and $\Psi(.)$ is same as defined in \eqref{eq:7.17}.
Next, we establish that $(\boldsymbol{\beta},\boldsymbol{\delta}, \uptau)\mapsto \\ \mathbb{G}_{n}\big(\Psi\left(\boldsymbol{u}, \boldsymbol{\beta},\boldsymbol{\delta}, \uptau\right)\big)$ is stochastically equicontinuous on $\mathscr{B}\times\Delta \times \mathcal{T},$ with respect to the $L_{2}(P)$ pseudo-metric $(d)$, which is defined as follows:
\begin{align}\label{eq:7.26}
d[(\boldsymbol{\beta}_{1},\boldsymbol{\delta}_{1},\uptau_{1}),(\boldsymbol{\beta}_{2},\boldsymbol{\delta}_{2},\uptau_{2})]=\Big[\mathbb{E}\Big(\big\|\Psi(\boldsymbol{u},\boldsymbol{\beta}_{1},\boldsymbol{\delta}_{1},\uptau_{1})-\Psi(\boldsymbol{u},\boldsymbol{\beta}_{2},\boldsymbol{\delta}_{2},\uptau_{2})\big\|^{r}\Big)\Big]^{\frac{1}{r}}
\end{align}
where $r$ is the same as defined in the Assumption \ref{a:3.3.4}. First, Consider the\\ $\mathbb{G}_{1,n}(\Psi_{1}(\vartheta,\boldsymbol{z},\boldsymbol{\delta})),$ is defined as
\begin{equation}\label{eq:7.27}
    \mathbb{G}_{1,n}(\Psi_{1}(\vartheta,\boldsymbol{z} ,\boldsymbol{\delta}))=\frac{1}{\sqrt{n}}\sum_{i=1}^{n}\big\{\Psi_{1}(\vartheta_{i},\boldsymbol{z}_{i},\boldsymbol{\delta})-\mathbb{E}\big[\Psi_{1}(\vartheta_{i},\boldsymbol{z}_{i},\boldsymbol{\delta})\big]\big\}.
\end{equation}
Define the class of functions $
\mathscr{G}=\{\Psi_{1}\big(\vartheta_{i},\boldsymbol{z}_{i},\boldsymbol{\delta}\big):\boldsymbol{\delta}\in \Delta\}= \{(w_{i}-\boldsymbol{z}_{i}^\top\boldsymbol{\delta})\boldsymbol{z}_{i}:\boldsymbol{\delta}\in \Delta\}$ 
This class belongs to the Euclidean class with envelope $\mathcal{G}_{1}=\|\Psi_{1}(.)\|+M\|\boldsymbol{z}_{i}\|$, where $M=2\sqrt{p+1}\sup_{\boldsymbol{\delta \in \Delta}}\|\boldsymbol{\delta}-\boldsymbol{\delta}_{0}\|$ by the lemma 2.13 in \cite{PakesandPollard}. The other component $\|\boldsymbol{z}_{i}\|$ is also belongs to the  Euclidean class with envelope $\|\boldsymbol{z}_{i}\|$. Then the product of Euclidean class is again a Euclidean class with envelope $ \mathcal{G}=\mathcal{G}_{1}.\|\boldsymbol{z}_{i}\|$. Moreover, $\mathbb{E}[\mathcal{G}^2]<\infty$ and $\mathbb{E}[\Psi_{1}(.)^2]$ is continuous at $\boldsymbol{\delta}_{0}$ by using the DCT under the assumptions \ref{a:3.3.1} and \ref{a:3.3.2}. Hence by using Lemma 2.14 of \cite{PakesandPollard}, it follow that $\boldsymbol{\delta}\mapsto \mathbb{G}_{1,n}\big(\Psi_{1}\big(\vartheta,\boldsymbol{z}, \boldsymbol{\delta}\big)\big)$ is stochastically continuous.
Next, consider the $\mathbb{G}_{2,n}(\Psi_{2}(.)),$ and define the class of functions 
\begin{equation}\label{eq:7.28}
\begin{split}
    \mathscr{F}&=\{\mathds{1}\left(\boldsymbol{x}_{i}^\top\boldsymbol{\beta}>0\right)\phi_{\uptau}(y_{i}-\boldsymbol{x}_{i}^\top\boldsymbol{\beta})\left(\boldsymbol{x}_{i}\right):\boldsymbol{\beta}\in \mathscr{B},\uptau \in \mathcal{T}\}\\
    &=\{\mathds{1}\left(\boldsymbol{x}_{i}^\top\boldsymbol{\beta}>0\right)\phi_{\uptau}(\boldsymbol{\varepsilon}_{\uptau,i}-\boldsymbol{x}^\top_{i}\boldsymbol{\delta}^{*}-\mathcal{Q}_{\varepsilon_{\uptau,i}}(\uptau))\left(\boldsymbol{x}_{i}\right):\boldsymbol{\beta}\in \mathscr{B},\uptau \in \mathcal{T}\},
\end{split}
\end{equation}
 where $\mathscr{B}\times\mathcal{T}$ is compact set and $\boldsymbol{\delta}^{*}=\boldsymbol{\beta}-\boldsymbol{\beta}_{0}(\uptau)$.
This class is formed as $\mathscr{F}_{0}(\mathcal{T}-\mathscr{F}_{1})(\boldsymbol{x}_{i}),$ where $\mathscr{F}_{1}=\{\mathds{1}(\boldsymbol{\varepsilon}_{\uptau,i}\leq \boldsymbol{x}^\top_{i}\boldsymbol{\delta}^{*}+\mathcal{Q}_{\varepsilon_{\uptau,i}}(\uptau)):\boldsymbol{\beta}\in \mathscr{B},\uptau \in \mathcal{T}\},$ $\mathcal{T}=\{\uptau:\uptau\in\mathcal{T}\}$ and $\mathscr{F}_{0}=\{\mathds{1}(\boldsymbol{x}_{i}^\top\boldsymbol{\beta}>0):\boldsymbol{\beta}\in \mathscr{B}\}.$
By Lemma 2.6.15 and Lemma 2.6.18 in \cite{vanderVaart1996}, the class of functions $\mathscr{F}_{0}$ and $\mathscr{F}_{1}$ is a (VC) subgraph class (see, \cite{MR0288823}). Therefore, the classes $\mathscr{F}_{0}$ and $\mathscr{F}_{1}$ are a Donsker classes with a constant envelopes. Theorem 2.10.6 in \cite{MR1385671} implies that $\mathcal{T}-\mathscr{F}_{0}$ is also Donsker with constant envelope. Then the product of $\mathscr{F}_{1}(\mathcal{T}-\mathscr{F}_{0})$ with $\boldsymbol{x}$ also forms a Donsker class with a square integrable envelope $4.\|\boldsymbol{x}\|$ by using Theorem 2.10.6 in \cite{MR1385671}. Hence $(\boldsymbol{\beta}, \boldsymbol{\delta}, \uptau) \mapsto \mathbb{G}_{2,n}\big(\Psi_{2}\big(\boldsymbol{s},\boldsymbol{\beta},\boldsymbol{\delta},\uptau\big)\big)$ is stochastically equicontinuous.
Since both components $\mathbb{G}_{1,n}\big(\Psi_{1}\big(\vartheta,\boldsymbol{z}, \boldsymbol{\delta}\big)\big)$ and $\mathbb{G}_{2,n}\big(\Psi_{2}\big(\boldsymbol{s},\boldsymbol{\beta},\boldsymbol{\delta},\uptau\big)\big)$  are stochastically equicontinuous, it follows that  $\mathbb{G}_{n}\big(\Psi\big(\boldsymbol{u},\boldsymbol{\beta},\boldsymbol{\delta}, \uptau\big)\big)$ is stochastically equicontinuous on $\mathscr{B}\times \Delta \times \mathcal{T}.$

\noindent{\bf{Third step:}} We now aim to establish that
\begin{equation}\label{eq:7.29}
\sup_{\uptau \in \mathcal{T}}\big\|\mathbb{G}_{n}\big(\Psi\big(\boldsymbol{u},\Hat{\boldsymbol{\beta}}_{n}(\uptau),\hat{\boldsymbol{\delta}}_{n},\uptau\big)\big)-\mathbb{G}_{n}\big(\Psi\big(\boldsymbol{u},\boldsymbol{\beta}_{0}(\uptau),\boldsymbol{\delta}_{0},\uptau\big)\big)\big\|\xrightarrow{p}0,
\end{equation}  as $n\rightarrow\infty$, 
where $\Psi(.)$ is defined in \eqref{eq:7.17}. To prove it, we show that uniform consistency of $\hat{\boldsymbol{\beta}}_{n}(\uptau)$ and consistency of $\hat{\boldsymbol{\delta}}_{n}$ imply that 
 \begin{equation}\label{eq:7.30}
\sup_{\uptau \in \mathcal{T}}d\left[\left(\uptau,b(\uptau),\boldsymbol{\delta}),(\uptau,\boldsymbol{\beta}_{0}(\uptau),\boldsymbol{\delta}_{0}\right)\right]^{r}|_{b(\uptau)=\Hat{\boldsymbol{\beta}}_{n}(\uptau), \delta = \hat{\boldsymbol{\delta}}_{n}}\xrightarrow{p}0,
\end{equation} as $n\rightarrow\infty$, 
where $d$ is a pseudometric defined in \eqref{eq:7.25}. Let us consider
\begin{equation}\label{eq:7.31}
  \begin{split}
  d\left[\left(\uptau,b(\uptau),\boldsymbol{\delta}),(\uptau,\boldsymbol{\beta}_{0}(\uptau),\boldsymbol{\delta}_{0}\right)\right]^{r} &=\mathbb{E}\big[\big\|\Psi\big(\boldsymbol{u},b(\uptau),\boldsymbol{\delta},\uptau\big)-\Psi\big(\boldsymbol{u},\boldsymbol{\beta}_{0}(\uptau),\boldsymbol{\delta}_{0},\uptau\big)\big\|^{r}\big]\\
  &\leq \mathbb{E}\big[\big\|\Psi_{1}\big(\vartheta,\boldsymbol{z}, \boldsymbol{\delta}\big)-\Psi_{1}\big(\vartheta,\boldsymbol{z}, \boldsymbol{\delta}_{0}\big)\big\|^{r}\big]\\ &+\mathbb{E}\big[\big\|\Psi_{2}\big(\boldsymbol{s}, b(\uptau),\boldsymbol{\delta},\uptau\big)-\Psi_{2}\big(\boldsymbol{s}, \boldsymbol{\beta}_{0}(\uptau),\boldsymbol{\delta}_{0},\uptau\big)\big\|^{r}\big]
 \end{split}  
\end{equation}
For the first component, then we obtain 
\begin{equation}\label{eq:7.32}
    \mathbb{E}\big[\big\|\Psi_{1}\big(\vartheta,\boldsymbol{z}, \boldsymbol{\delta}\big)-\Psi_{1}\big(\vartheta,\boldsymbol{z}, \boldsymbol{\delta}_{0}\big)\big\|^{r}\big]\leq \mathbb{E}\big[\|\boldsymbol{z}_{i}\|^{2r}\big] \big(\|\boldsymbol{\delta}-\boldsymbol{\delta}_{0}\|^{r}\big).
\end{equation}
For the second term, applying the Mean Value Theorem, the Cauchy–Schwarz inequality,  and the fact that for any real numbers 
$a$ and $b,$ $\mathds{1}(a>0)-\mathds{1}(b>0)\leq \mathds{1}(|a|\leq|a-b|),$ we obtain
\begin{equation}\label{eq:7.33}
  \begin{split}
\big\|\Psi_{2}\big(\boldsymbol{s}, b(\uptau),\boldsymbol{\delta},\uptau\big)-\Psi_{2}\big(\boldsymbol{s}, \boldsymbol{\beta}_{0}(\uptau),\boldsymbol{\delta}_{0},\uptau\big)\big\|^{r}&\leq\|\boldsymbol{x}_{i}\|^{r}\big\{\big\|\boldsymbol{x}\| \|b(\uptau)-\boldsymbol{\beta}_{0}(\uptau)\|^{r'}\\
&+\big[\mathds{1}(|\boldsymbol{x}^\top b(\uptau)|\leq\|\boldsymbol{x}\|\|b(\uptau)-\boldsymbol{\beta}_{0}(\uptau)\|)\big]\big\}\\
 \end{split}  
\end{equation}
Taking expectations and using Assumption \ref{a:3.3.7}, it follows that 
\begin{equation}\label{eq:7.34}
\begin{split}
    \mathbb{E}\big[\big\|\Psi_{2}\big(\boldsymbol{s}, b(\uptau),\boldsymbol{\delta},\uptau\big)-\Psi_{2}\big(\boldsymbol{s}, \boldsymbol{\beta}_{0}(\uptau),\boldsymbol{\delta}_{0},\uptau\big)\big\|^{r}&\leq \mathbb{E}\big[\big\|\boldsymbol{x}\big\|^{r+1}\big]\big\|b(\uptau)-\boldsymbol{\beta}_{0}(\uptau)\big\|\\
    &+K.\big(\mathbb{E}\big[\big\|\boldsymbol{x}\big\|^{2r}\big]\big)^{\frac{1}{2}}.\big\|b(\uptau)-\boldsymbol{\beta}_{0}(\uptau)\big\|
    \end{split}
\end{equation}
Combining \eqref{eq:7.31} and \eqref{eq:7.33}, and applying Assumptions \ref{a:3.3.2} and \ref{a:3.3.4}, we have
\begin{equation}\label{eq:7.35}
\begin{split}
    d\Big[\big(\uptau,b(\uptau),\boldsymbol{\delta}),(\uptau,\boldsymbol{\beta}_{0}(\uptau),\boldsymbol{\delta}_{0}\big)\Big]_{b(\uptau)= \Hat{\boldsymbol{\beta}}_{n} (\uptau), \boldsymbol{\delta} = \hat{\boldsymbol{\delta}}_{n}}
&\leq \Big\{\phi_{1}.\big(\|\hat{\boldsymbol{\delta}}_{n}-\boldsymbol{\delta}_{0}\|^{r}\big)+\phi_{2}.(1+K).\\
&\sup_{\uptau \in \mathcal{T}}\|\Hat{\boldsymbol{\beta}}_{n}(\uptau)-\boldsymbol{\beta}_{0}(\uptau)\|\Big\}^{\frac{1}{r}}.\\
\end{split}
\end{equation}
Therefore, using the consistency of $\hat{\boldsymbol{\beta}}_{n}(\uptau)$ and $\hat{\boldsymbol{\delta}}_{n}$, equation \eqref{eq:7.30} is thereby established. Moreover, $(\boldsymbol{\beta},\boldsymbol{\delta},\uptau)\mapsto \mathbb{G}_{n}(\boldsymbol{u},\boldsymbol{\beta},\boldsymbol{\delta},\uptau)$ is stochastically equicontinuous. Together with \eqref{eq:7.30}, this yields the desired result in \eqref{eq:7.29}.

\noindent{\bf{Fourth step:}} Recall \eqref{eq:7.21},  and we have
\begin{equation}\label{eq:7.36}
\begin{split}
o_{p}(1)&=\sqrt{n}\mathbb{E}_{n}\big[\Psi_{2}(\boldsymbol{s},\hat{\boldsymbol{\beta}}_{n},\hat{\boldsymbol{\delta}}_{n},\uptau)\big]\quad \text{in} \quad \ell^{\infty}(\mathscr{B} \times \mathcal{T})\\
o_{p}(1)&=\sqrt{n}\Big[\mathbb{E}_{n}\big[\Psi_{2}(\boldsymbol{s},\hat{\boldsymbol{\beta}}_{n},\hat{\boldsymbol{\delta}}_{n},\uptau)\big]-\mathbb{E}\big[\Psi_{2}(\boldsymbol{s},\hat{\boldsymbol{\beta}}_{n},\hat{\boldsymbol{\delta}}_{n},\uptau)\big]\Big]+\sqrt{n}\mathbb{E}\big[\Psi_{2}(\boldsymbol{s},\hat{\boldsymbol{\beta}}_{n},\hat{\boldsymbol{\delta}}_{n},\uptau)\big]\\
o_{p}(1)&=\mathbb{G}_{2,n}\big(\Psi_{2}(\boldsymbol{s},\hat{\boldsymbol{\beta}}_{n},\hat{\boldsymbol{\delta}}_{n},\uptau)\big)+\sqrt{n}\mathbb{E}\big[\Psi_{2}(\boldsymbol{s},\hat{\boldsymbol{\beta}}_{n},\hat{\boldsymbol{\delta}}_{n},\uptau)\big]\quad  \text{in} \quad \ell^{\infty}(\mathscr{B}\times \mathcal{T})
 \end{split}
  \end{equation}
 Since $\mathbb{G}_{2,n}\big(\Psi_{2}\big(.\big)\big)$ is stochastically equicontinuous, it follows from equation \eqref{eq:7.29} that
\begin{equation}\label{eq:7.37}
\mathbb{G}_{2,n}\big(\Psi_{2}(\boldsymbol{s},\boldsymbol{\beta}_{0},\boldsymbol{\delta}_{0},\uptau)\big)+\sqrt{n}\mathbb{E}\big[\Psi_{2}(\boldsymbol{s},\hat{\boldsymbol{\beta}}_{n},\hat{\boldsymbol{\delta}}_{n},\uptau)\big]=o_{p}(1)\quad  \text{in} \quad \ell^{\infty}(\mathscr{B}\times \mathcal{T})
\end{equation}
Next apply a first-order Taylor expansion (uniform in 
$\uptau$), of the population expectation around the true parameter $(\boldsymbol{\beta}_{0}(\uptau),\boldsymbol{\delta}_{0})$ then we obtain 
\begin{equation}\label{eq:7.38}
  \mathbb{E}\big[\Psi_{2}(\boldsymbol{s},\hat{\boldsymbol{\beta}}_{n},\hat{\boldsymbol{\delta}}_{n},\uptau)\big]=\mathbb{E}\big[\Psi_{2}(\boldsymbol{s},\boldsymbol{\beta}_{0},\boldsymbol{\delta}_{0},\uptau)\big] +\mathbb{J}_{2,\boldsymbol{\beta}}(\uptau)\big(\hat{\boldsymbol{\beta}}_{n}(\uptau)-\boldsymbol{\beta}_{0}(\uptau)\big) +\mathbb{J}_{2,\boldsymbol{\delta}}(\uptau)\big(\hat{\boldsymbol{\delta}}_{n}-\boldsymbol{\delta}_{0}\big),
\end{equation}
where the $\mathbb{J}_{2,\boldsymbol{\beta}}(\uptau)$ and $\mathbb{J}_{2,\boldsymbol{\delta}}(\uptau)$ are block matrix defined in \eqref{eq:3.3.7}. Since $\mathbb{E}\big[\Psi_{2}(\boldsymbol{s},\boldsymbol{\beta}_{0},\boldsymbol{\delta}_{0},\uptau)\big]=0,$ and Substituting this expression into the previous equation gives
\begin{equation}\label{eq:7.39}
\mathbb{J}_{2,\boldsymbol{\beta}}(\uptau)\sqrt{n}\big(\hat{\boldsymbol{\beta}}_{n}(\uptau)-\boldsymbol{\beta}_{0}(\uptau)\big) +\mathbb{J}_{2,\boldsymbol{\delta}}(\uptau)\sqrt{n}\big(\hat{\boldsymbol{\delta}}_{n}-\boldsymbol{\delta}_{0}\big)=\mathbb{G}_{2,n}\big(\Psi_{2}(\boldsymbol{s},\boldsymbol{\beta}_{0},\boldsymbol{\delta}_{0},\uptau)\big)+o_{p}(1)
\end{equation}
Similarly, from equation \eqref{eq:7.21} and applying Taylor expansion, we have
\begin{equation}\label{eq:7.40}
\begin{split}
\sqrt{n}\mathbb{E}_{n}\big[\Psi_{1}\big(\vartheta,\boldsymbol{z},\hat{\boldsymbol{\delta}}_{n}\big)\big]&=o_{p}(1)\\
\mathbb{J}_{1,\boldsymbol{\delta}}\sqrt{n}\big(\hat{\boldsymbol{\delta}}_{n}-\boldsymbol{\delta}_{0}\big)&=\mathbb{G}_{n,1}\big(\vartheta,\boldsymbol{z},\boldsymbol{\delta}_{0}\big)+o_{p}(1)
    \end{split},
\end{equation}
where $\mathbb{J}_{1,\boldsymbol{\delta}}$ matrix is same defined in \eqref{eq:3.3.7}. Under the Assumptions \ref{a:2.1}, \ref{a:3.3.1} and \ref{a:3.3.2} and applying CLT for $\alpha-$ mixing processes (see, Theorem~5.20 in \cite{white1984}, p.p.-130) this implies $\sqrt{n}\big(\hat{\boldsymbol{\delta}}_{n}-\boldsymbol{\delta}_{0}\big)=O_{p}(1).$
Combining equations \eqref{eq:7.39} and \eqref{eq:7.40}, and stacking the two vector equations, we obtain the joint representation
\begin{equation}\label{eq:7.41}
\mathbb{J}(\uptau) \sqrt{n}
\begin{pmatrix}
          \begin{bmatrix}
          \hat{\boldsymbol{\beta}}_{n}(\uptau)\\
          \hat{\boldsymbol{\delta}}_{n}
          \end{bmatrix} -
          \begin{bmatrix}
           \boldsymbol{\beta}_{0}(\uptau)\\
          \boldsymbol{\delta}_{0}
         \end{bmatrix}
         \end{pmatrix}=\mathbb{G}_{n} \big( \Psi(\boldsymbol{u},
\boldsymbol{\beta}_{0}(\tau), \boldsymbol{\delta}_{0}, \tau) \big)+o_{p}(1), \quad  \text{in} \quad \Delta \times\ell^{\infty}(\mathscr{B} \times \mathcal{T}),
\end{equation}
where $\mathbb{G}_{n} \big( \Psi(\boldsymbol{u},
\boldsymbol{\beta}_{0}(\tau), \boldsymbol{\delta}_{0}, \tau) \big)=\begin{pmatrix}
     \mathbb{G}_{2,n}\big(\Psi_{2}(\boldsymbol{s},\boldsymbol{\beta}_{0},\boldsymbol{\delta}_{0},\uptau)\big)\\
     \mathbb{G}_{n,1}\big(\vartheta,\boldsymbol{z},\boldsymbol{\delta}_{0}\big)
    \end{pmatrix}$ and $\mathbb{J}(\uptau)$ is Jacobian matrix defined in \eqref{eq:3.3.6}. For any positive definite matrix $A$ and a nonzero vector $b,$ we have $\|Ab\|\geq \lambda_{(1)}\|b\|,$ where $\lambda_{(1)}$ is the smallest eigen value of $A$ (see, e.g., \cite{Angrist}), under Assumptions \ref{a:3.3.2}, \ref{a:3.3.3}, and \ref{a:3.3.4}, it follows that $\mathbb{J}(\uptau)$ is uniformly positive definite. Consequently, this bound holds uniformly for $\mathbb{J}(\uptau).$  Then, we have
\begin{equation}\label{eq:7.42}
\begin{split}
\sup_{\tau \in \mathcal{T}} 
\left\| \mathbb{G}_{n} \big( \Psi(\boldsymbol{u},
\boldsymbol{\beta}_{0}(\tau), \boldsymbol{\delta}_{0}, \tau) \big) + o_{p}(1) \right\|  
&= \sup_{\tau \in \mathcal{T}} 
\Bigg\| 
\mathbb{J}(\tau)\sqrt{n} 
\begin{pmatrix}
          \begin{bmatrix}
          \hat{\boldsymbol{\beta}}_{n}(\uptau)\\
          \hat{\boldsymbol{\delta}}_{n}
          \end{bmatrix} -
          \begin{bmatrix}
           \boldsymbol{\beta}_{0}(\uptau)\\
          \boldsymbol{\delta}_{0}
         \end{bmatrix}
         \end{pmatrix}\\
&+ o_{p}\!\Bigg( \sqrt{n}
\begin{pmatrix}
          \begin{bmatrix}
          \hat{\boldsymbol{\beta}}_{n}(\uptau)\\
          \hat{\boldsymbol{\delta}}_{n}
          \end{bmatrix} -
          \begin{bmatrix}
           \boldsymbol{\beta}_{0}(\uptau)\\
          \boldsymbol{\delta}_{0}
         \end{bmatrix}
         \end{pmatrix}
\Bigg)
\Bigg\|,\\
&\geq \big(\lambda_{( 1)} +o_{p}(1)\big)\left\|\sqrt{n} 
\begin{pmatrix}
          \begin{bmatrix}
          \hat{\boldsymbol{\beta}}_{n}(\uptau)\\
          \hat{\boldsymbol{\delta}}_{n}
          \end{bmatrix} -
          \begin{bmatrix}
           \boldsymbol{\beta}_{0}(\uptau)\\
          \boldsymbol{\delta}_{0}
         \end{bmatrix}
         \end{pmatrix}\right\|
\end{split}
\end{equation}
 where $\lambda_{( 1)}$ is the smallest eigen value of $\mathbb{J}(\uptau)$ uniformly in $\uptau \in \mathcal{T}.$

\noindent {\bf{Fifth step:}}
 Observe that the mapping $\uptau \mapsto \boldsymbol{\beta}_{0}(\uptau)$ is Lipchitz continuous in $\uptau,$  by the Assumption \ref{a:3.3.6}. Moreover, $\boldsymbol{\beta}_{0}(\uptau)$ is a solution to $\boldsymbol{\beta}$ in $$\mathbb{E}\left\{\mathds{1}\left(\boldsymbol{x}_{i}^\top\boldsymbol{\beta}>0\right)\varphi_{\uptau}(\varepsilon_{\uptau,i}-\boldsymbol{x}^\top_{i}\boldsymbol{\delta}^{*}-Q_{\varepsilon_{\uptau}}(\uptau))\left(\boldsymbol{x}_{i}\right)\right\}=0.$$ Hence, due to all the aforesaid facts, the process  
$\uptau \mapsto \mathbb{G}_{2,n}\left[\Psi_{2}\left(\boldsymbol{s}, \boldsymbol{\beta}_{0}(\uptau),\boldsymbol{\delta}_{0},\uptau\right)\right]$ is stochastically equicontinuous over $\mathcal{T}$ with respect to the pseudometric given by $$d\big[\big(\uptau_{1},\boldsymbol{\beta}_{0}(\uptau_{1})),(\uptau_{2},\boldsymbol{\beta}_{0}(\uptau_{2})\big)\big]^{r},$$ where $d$ is defined in \eqref{eq:7.26}. Consequently, the process $\uptau \mapsto \mathbb{G}_{n}\left[\Psi\left(\boldsymbol{u}, \boldsymbol{\beta}_{0}(\uptau),\boldsymbol{\delta}_{0},\uptau\right)\right]$ is stochastically equicontinuous over $\mathcal{T}.$ Next, we show that the finite-dimensional distributions converge. For any $k\geq1$ and $\uptau_{1},\ldots, \uptau_{k} \in \mathcal{T},$ 
$\mathbb{G}_{n}\big(\Psi\big(\boldsymbol{u},\boldsymbol{\beta}_{0}(\uptau),\boldsymbol{\delta}_{0},\uptau\big)\big)$ converges in distribution to $\mathbb{G}\big(\Psi(\boldsymbol{u},\boldsymbol{\beta},\boldsymbol{\delta},\uptau\big)\big).$ This claim is proved if we show that $\mathbb{G}_{n}\big(\Psi\big(\boldsymbol{u},\boldsymbol{\beta},\boldsymbol{\delta},\uptau_{1}\big)\big),\cdots,\mathbb{G}_{n}\big(\Psi\big(\boldsymbol{u},\boldsymbol{\beta},\boldsymbol{\delta},\uptau_{k}\big)\big)$ converges in distribution to $\mathbb{G}\big(\Psi\big(\boldsymbol{u},\boldsymbol{\beta},\boldsymbol{\delta},\uptau_{1}\big)\big),\cdots,\mathbb{G}\big(\Psi\big(\boldsymbol{u},\boldsymbol{\beta},\boldsymbol{\delta},\uptau_{k}\big)\big)$ for any $k>1.$ By the Cramer-Wold device, we need to show that, for any unit vector $a=(a_{1},\cdots, a_{k})^\top\in \mathbb{R}^{k},$ then $\mathscr{R}_{n}=\sum_{i=1}^{k} a_{i}\mathbb{G}_{n}\big(\Psi\big(\boldsymbol{u},\boldsymbol{\beta},\boldsymbol{\delta},\uptau_{i}\big)\big)$ converges in distribution to $\mathscr{R}=\sum_{i=1}^{k}a_{i}\mathbb{G}\big(\Psi\big(\boldsymbol{u},\boldsymbol{\beta},\boldsymbol{\delta},\uptau_{i}\big)\big).$ To show this convergence, we adapt the blocking approach (see,e.g, \cite{DEHLING20093699}). The proof is an adaptation of that of Lemma~B.2 in \cite{BUCHER201683}, which is very similar to the proof of the Theorem ~\ref{th:3.3.4} in Appendix~\ref{appendix}. For brevity, the details are omitted here. Finally, applying the Functional CLT for $\alpha$-mixing processes (see Theorem 2.1 in \cite{arcones1994central} and \cite{CHANTSAY}), the weak convergence result follows.
\begin{equation}\label{eq:7.43}
 \mathbb{G}_{n}\left(\Psi\left(\boldsymbol{u},\boldsymbol{\beta}_{0}(\uptau),\boldsymbol{\delta}_{0},\uptau\right)\right) \Rightarrow \mathbb{G}(\uptau)
\hspace{0.4cm} in \quad \Delta \times \ell^{\infty}(\mathscr{B}\times\mathcal{T}),
\end{equation}
where $\mathbb{G}(\uptau)$ is a centered Gaussian vector process with covariance function $\mathbb{E}[\mathbb{G}(\uptau)\mathbb{G}(\uptau')^\top]=\mathbb{V}(\uptau,\uptau')$ defined in \eqref{eq:3.3.5}. Therefore, $\displaystyle \sup_{\uptau \in \mathcal{T}}\|\mathbb{G}_{n}\left(\Psi\left(\boldsymbol{u}, \boldsymbol{\beta}_{0}(\uptau),\boldsymbol{\delta}_{0},\uptau\right)\right)+o_{p}(1)\| = O_{p}(n^{-1/2}),$ and which implies from \eqref{eq:7.42} that 
\begin{equation}\label{eq:7.44}
\displaystyle \sup_{\uptau \in \mathcal{T}} \left\|\sqrt{n} 
\begin{pmatrix}
          \begin{bmatrix}
          \hat{\boldsymbol{\beta}}_{n}(\uptau)\\
          \hat{\boldsymbol{\delta}}_{n}
          \end{bmatrix} -
          \begin{bmatrix}
           \boldsymbol{\beta}_{0}(\uptau)\\
          \boldsymbol{\delta}_{0}
         \end{bmatrix}
         \end{pmatrix}\right\|=O_{p}(1)
\end{equation}
Finally, using \eqref{eq:7.41} and \eqref{eq:7.43}, we obtain
\begin{equation}\label{eq:7.45}
\sqrt{n} 
\begin{pmatrix}
          \begin{bmatrix}
          \hat{\boldsymbol{\beta}}_{n}(\uptau)\\
          \hat{\boldsymbol{\delta}}_{n}
          \end{bmatrix} -
          \begin{bmatrix}
           \boldsymbol{\beta}_{0}(\uptau)\\
          \boldsymbol{\delta}_{0}
         \end{bmatrix}
         \end{pmatrix} =\mathbb{J}(\uptau)^{-1}.\mathbb{G}_{n}\left(\Psi\left(\boldsymbol{u}, \boldsymbol{\beta}_{0}(\uptau),\boldsymbol{\delta}_{0},\uptau\right)\right)+o_{p}(1)\Rightarrow \mathbb{G}(.),
\end{equation}in $\Delta \times \ell^{\infty}(\mathscr{B}\times\mathcal{T})$ and 
\begin{equation}\label{eq:7.46}
    \mathbb{J}(\uptau)^{-1}=\begin{bmatrix}
     \mathbb{J}_{2,\boldsymbol{\beta}}(\uptau)&\mathbb{J}_{2,\boldsymbol{\delta}}(\uptau)\\
     \boldsymbol{0} &\mathbb{J}_{1,\boldsymbol{\delta}}
    \end{bmatrix}^{-1}=\begin{bmatrix}
    \mathbb{J}_{2,\boldsymbol{\beta}}^{-1}(\uptau)&-\mathbb{J}_{2,\boldsymbol{\beta}}^{-1}(\uptau)\mathbb{J}_{2,\boldsymbol{\delta}}(\uptau)\mathbb{J}_{1,\boldsymbol{\delta}}^{-1}\\
    \boldsymbol{0}&\mathbb{J}_{1,\boldsymbol{\delta}}^{-1}.
    \end{bmatrix}
\end{equation} 
    Now, using the continuity mapping theorem, the weak convergence of \eqref{eq:7.39} will also hold. Hence, this completes the proof. 
 \end{proof}
\subsection{B.5 Asymptotic Normality of \texorpdfstring{$\mathscr{L}_{n}$}{}}\label{B.5}
 \begin{proof}[{\bf{Proof of Theorem \ref{th:3.3.4}:}}] Theorem \ref{th:3.3.3} asserts that
 \begin{equation}
 \begin{split}
     \sqrt{n}\big(\hat{\boldsymbol{\beta}}_{n}(\uptau)-\boldsymbol{\beta}_{0}(\uptau)\big)&=\mathbb{J}^{-1}_{2,\boldsymbol{\beta}}(\uptau)\Big\{\frac{1}{\sqrt{n}}\sum_{i=1}^{n}\Psi_{2}\big(\boldsymbol{s}_{i},\boldsymbol{\beta}_{0}(\uptau),\boldsymbol{\delta}_{0},\uptau\big)-\mathbb{J}_{2,\boldsymbol{\delta}}(\uptau)\mathbb{J}^{-1}_{1,\boldsymbol{\delta}}\\
     &\frac{1}{\sqrt{n}}\sum_{i=1}^{n}\Psi_{1}\big(\vartheta_{i},\boldsymbol{z}_{i},\boldsymbol{\delta}_{0}\big)\Big\} +o_{p}(1)\quad \text{in}\quad \ell^{\infty}(\mathscr{B}\times \mathcal{T})
     \end{split}
     \end{equation}
From the definition of the L–estimator in \eqref{eq:2.12} and \eqref{eq:2.13}, it follows that
     \begin{equation}
 \begin{split}
     \mathscr{L}_{n}-\mathscr{L}_{0}&=\int_{0}^{1}\Big[\mathbb{J}^{-1}_{2,\boldsymbol{\beta}}(\uptau)\Big\{\frac{1}{\sqrt{n}}\sum_{i=1}^{n}\Psi_{2}\big(\boldsymbol{s}_{i},\boldsymbol{\beta}_{0}(\uptau),\boldsymbol{\delta}_{0},\uptau\big)-\mathbb{J}_{2,\boldsymbol{\delta}}(\uptau)\mathbb{J}^{-1}_{1,\boldsymbol{\delta}}\\&\frac{1}{\sqrt{n}}\sum_{i=1}^{n}\Psi_{1}\big(\vartheta_{i},\boldsymbol{z}_{i},\boldsymbol{\delta}_{0}\big)\Big\}\Big]J_{1}(\uptau)d(\uptau)+o_{p}(1)
     \end{split}
     \end{equation}
     Since the matrices $\mathbb{J}^{-1}_{2,\boldsymbol{\beta}}(\uptau)$ and $\mathbb{J}^{-1}_{1,\boldsymbol{\delta}}$ are uniformly bounded in $\uptau \in \mathcal{T}$ by the assumptions \ref{a:3.3.3}, \ref{a:3.3.4} and \ref{a:3.3.2} and the $\mathbb{J}_{2,\boldsymbol{\delta}}(\uptau)$ by assumption \ref{a:3.3.8}. Here, we use $\|.\|_{2}$ spectral norm for any matrix. In addition, the functions $\Psi_{2}(.)$ and $\Psi_{1}(.)$ are dominated by an integrable functions using the assumptions \ref{a:3.3.2} and \ref{a:3.3.4}. Therefore, by the DCT, the order of integration and summation can be interchanged. Then, we obtain
     \begin{equation}
     \mathscr{L}_{n}-\mathscr{L}_{0}=\frac{1}{\sqrt{n}}\sum_{i=1}^{n}h_{i}+o_{p}(1),
    \end{equation}
 where  
 \begin{equation}
 \begin{split}
     h_{i}&=\int_{0}^{1}\Tilde{v}_{i}(\uptau)J_{1}(\uptau)d(\uptau)\\ 
     \Tilde{v}_{i}(\uptau)&=\mathbb{J}^{-1}_{2,\boldsymbol{\beta}}(\uptau)\big\{\Psi_{2}\big(\boldsymbol{s}_{i},\boldsymbol{\beta}_{0}(\uptau),\boldsymbol{\delta}_{0},\uptau\big)-\mathbb{J}_{2,\boldsymbol{\delta}}(\uptau)\mathbb{J}^{-1}_{1,\boldsymbol{\delta}}
 \Psi_{1}\big(\vartheta_{i},\boldsymbol{z}_{i},\boldsymbol{\delta}_{0}\big)\big\}.
 \end{split}
 \end{equation}
 To establish the asymptotic distribution of the leading term
$$\frac{1}{\sqrt{n}}\sum_{i=1}^n h_{i},$$ we employ a blocking technique (see, e.g., \cite{DEHLING20093699}, 2002). The proof is an adaptation of Lemma 8.2 in \cite{BUCHER201683} and relies on the combination of the small-block/large-block method and the Cramer–Wold device. In view of the Cramer-Wold device, it is sufficient to show that for any unit vector $a\in\mathbb{R}^{p+2}$,
\begin{equation}\label{eq:7.50}
\mathbb{T}_{n} \,=\, \frac{1}{\sqrt{n}}\sum_{i=1}^n a^{\top}h_{i} \,=\, \frac{1}{\sqrt{n}}\sum_{i=1}^n \xi_{i} \ \xrightarrow{d}\mathbb{T},
\end{equation}
where $\xi_{i}=a^{\top}h_{i}$ and $\mathbb{T}$ is a $\mathbb{R}^{p+2}$ valued random vector associated with a $p+2$-dimensional multivariate normal distribution, with  covariance matrix $a^{\top}\boldsymbol{\Omega} a$, i.e. $\mathbb{T}\sim N\big(\boldsymbol{0},\,a^{\top}\boldsymbol{\Omega} a\big)$,
and $\boldsymbol{\Omega}$ is the same as defined in \eqref{eq:3.3.9}.
To establish this convergence, we apply Bernstein’s small-block and large-block technique. Namely, we partition
$$ C_n = \{1,2,\dots,n\} $$
into $(2k+1)$ disjoint subsets: the large blocks $L_{nm_{1}}$ and the small blocks $S_{nm_{1}}$ for $m_{1}=1,\dots,k$, together with a remainder block $R_{nk}$. The blocks are defined as follows:
For $m_{1}=1,\dots,k$, the large block is 
$$L_{nm_{1}} = \big\{ (m_{1}-1)(p_{1}+q_{1})+1,\dots,(m_{1}-1)(p_{1}+q_{1})+p_{1} \big\}.$$
and the small block is 
$$S_{nm_{1}} = \big\{ (m_{1}-1)(p_{1}+q_{1})+p_{1}+1,\dots, m_{1}(p_{1}+q_{1}) \big\}.$$
The remainder block is given by
$$R_{nk} = \big\{ k(p_{1}+q_{1})+1,\dots,n \big\}.$$
Thus, each $L_{nm_{1}}$ is a large block of length $p_{1}=p_{1n}$, each $S_{nm_{1}}$ is a small block of length $q_{1}=q_{1n}$. The total number of blocks is given by
$$k = k_{n}= \Big[ \frac{n}{p_{1}+q_{1}} \Big],$$
where $[.]$ denotes the integer part. Since $(p_{1}+q_{1}) \leq n$, for large $n$, the block lengths $p_{1n}$, $q_{1n}$, and the number of blocks $k_{n}$ all depend on $n$. For simplicity, we suppress the dependence on $n$ without confusing.
Define
$$\mathbb{T}_{n} = n^{-1/2} \big[ Z_{n}^{(L)} + Z_{n}^{(S)} + Z_{n}^{(R)} \big],$$
where for a large block,
\begin{equation}
    Z_{n}^{(L)} = \sum_{m_{1}=1}^K C_{nm_{1}}^{(L)}, \quad
   C_{nm_{1}}^{(L)} = \sum_{i\in L_{nm_{1}}}  \xi_{i}= \sum_{i=a_{m_{1}}}^{a_{m_{1}}+p_{1}-1} \xi_i\;,
\end{equation}
where $a_{m_1} = (m_1-1)(p_1+q_1)+1.$ Thus,
\begin{equation}
    Z_{n}^{(L)} = \sum_{m_{1}=1}^k \sum_{i=a_{m_{1}}}^{a_{m_{1}}+p_{1}-1} \xi_{i}.
\end{equation}
Similarly, for the small blocks,
\begin{equation}
    Z_{n}^{(S)} = \sum_{m_{1}=1}^{k} C_{nm_{1}}^{(S)}, \quad
C_{nm_{1}}^{(S)} = \sum_{j \in S_{nm_{1}}} \xi_{j}=\sum_{j=b_{m_{1}}}^{b_{m_{1}}+q_{1}-1} \xi_{j}, 
\end{equation}
where $b_{m_{1}} = a_{m_{1}}+p_{1}.$ Thus,
\begin{equation}
Z_n^{(S)} = \sum_{m_{1}=1}^k \sum_{j=b_{m_{1}}}^{b_{m_{1}}+q_{1}-1} \xi_{j}.
\end{equation}
Finally, the remainder block is given by 
\begin{equation}
    Z_{n}^{(R)} = C_{n,k+1}^{(R)}, \quad C_{n,k+1}^{(R)} = \sum_{i \in R_{nk}} \xi_i = \sum_{i=k(p_{1}+q_{1})+1}^{n} \xi_{i}.
\end{equation}

To prove \eqref{eq:7.50}, it suffices to verify that, as $n \to \infty,$ conditions \eqref{eq:A1}-\eqref{eq:A4} hold:
\begin{align}
 n^{-1}\,\mathbb{E}\Big[(Z_n^{(S)})^2\Big] \to 0, \quad n^{-1}\,\mathbb{E}\Big[(Z_n^{(R)})^2\Big] &\to 0, && \tag{\bf{A.1}}\label{eq:A1} \\
 \Big| \mathbb{E}\Big[ e^{it Z_n^{(L)}} \Big] - \prod_{m_1=1}^K \mathbb{E}\Big[e^{it C_{nm_1}^{(L)}}\Big] \Big| &\to 0, && \tag{\bf{A.2}}\label{eq:A2} \\
 \mathrm{Var}\big(n^{-1/2} Z_n^{(L)}\big) \to a^T \Sigma a &, && \tag{\bf{A.3}}\label{eq:A3} \\
 n^{-1} \sum_{m_{1}=1}^K \mathbb{E}\Big[ (C_{nm_{1}}^{(L)})^2 \mathbf{1}\big(|C_{nm_{1}}^{(L)}| > \Tilde{\epsilon} \sqrt{n}\big) \Big] &\to 0,  && \tag{\bf{A.4}}\label{eq:A4}
\end{align}
for $\Tilde{\epsilon}>0.$
From condition \ref{eq:A1}, we conclude that the $Z_n^{(S)}$ and $Z_n^{(R)}$ are asymptotically negligible in probability. Next, condition \ref{eq:A2} shows that the summand $C_{nm_1}^{(L)}$ in $Z_n^{(L)}$ are asymptotically independent. Furthermore, Conditions \ref{eq:A3} and \ref{eq:A4} are standard Lindeberg--Feller type conditions for asymptotic normality of $Z_{n}^{(L)}$ in the independent setup.

To verify \eqref{eq:A1}, we carefully choose the large-block size $p_{1}$ and small-block size $q_{1}$ so that
$$\frac{q_{1}}{p_{1}} \to 0, \quad \frac{p_{1}}{n} \to 0, \quad p_{1} \to \infty, \quad q_{1} \to \infty, \quad \frac{n\,\alpha(n)}{p_{1}} \to 0, \quad \text{as } n\to\infty,$$
where $\alpha(.)$ is same as defined in \eqref{eq:3.1}.
Now, consider \eqref{eq:A1},
\begin{equation}
\begin{split}
    \mathbb{E}\Big[(Z_n^{(S)})^2\Big] = \mathbb{E}\Big[ \Big( \sum_{m_{1}=1}^{k} C_{nm_{1}}^{(S)} \Big)^2 \Big]&= \sum_{m_{1}=1}^{k} \mathrm{Var}\big(C_{nm_{1}}^{(S)}\big)
+ 2 \sum_{1 \leq i < j \leq k} \mathrm{Cov}\big(C_{ni}^{(S)}, C_{nj}^{(S)}\big)\\
&= A_{1}+A_{2},
\end{split}
\end{equation}
where $$ A_{1} = \sum_{m_{1}=1}^k \mathrm{Var}\big(C_{nm_{1}}^{(S)}\big)\; \text{and}\; 
A_{2} = 2 \sum_{1 \leq i < j \leq k} \mathrm{Cov}\big(C_{ni}^{(S)}, C_{nj}^{(S)}\big).$$
We first analyze the term $A_{1}.$ Clearly,
$$A_{1} = \sum_{m_{1}=1}^k \mathrm{Var}\left(C_{n,m_{1}}^{(S)}\right),$$
where,
\begin{equation}
\mathrm{Var}\left(C_{n,m_{1}}^{(S)}\right) 
= \sum_{i=b_{m_{1}}}^{b_{m_{1}}+q_{1}-1} \mathrm{Var}(\xi_{i})
+ 2 \sum_{b_{m_{1}} \leq i < j \leq b_{m_{1}}+q_{1}-1} \mathrm{Cov}(\xi_{i},\xi_{j}).
\end{equation}
Therefore, we obtain
\begin{equation}
    A_{1} \leq k \sum_{i=b_{m_1}}^{b_{m_1}+q_1-1} \mathrm{Var}(\xi_i) 
+ 2k \sum_{b_{m_1} \leq i < j \leq b_{m_1}+q_1-1} \mathrm{Cov}(\xi_i,\xi_j).
\end{equation}
By stationarity and Davydov’s inequality (see, e.g, \cite{hall1980martingale} in Corollary A.2), for all $i<j$, we have
$$
\big| \mathrm{Cov}(\xi_{i}, \xi_{j}) \big| 
\leq \bar{C} \, \alpha(j-i)^{\frac{r-2}{r}}
\left(\mathbb{E}|\xi_{i}|^{r}\right)^{\frac{1}{r}}\left(\mathbb{E}|\xi_{j}|^{r}\right)^{\frac{1}{r}},
$$
where $\alpha(.)$ is same defined in \eqref{eq:3.1} and by applying Fubini’s theorem together with Minkowski inequality under Assumptions~\ref{a:3.3.2}–\ref{a:3.3.4}, we obtain $\mathbb{E}[|\xi_{j}|^{r}]<\infty$ for some $r>2$, as required in Assumption~\ref{a:3.3.4}. Then it implies that
\begin{equation}
\begin{split}
    A_{1} &\leq k q_{1} \, \mathrm{Var}(\xi_{i}) 
+ 2k  \sum_{b_{m_1} \leq i < j \leq b_{m_1}+q_1-1}\bar{C} \alpha(j-i)^{\frac{r-2}{r}}\\
&\leq k q_{1} \, \mathrm{Var}(\xi_{i}) 
+ 2k \Bar{C}\sum_{i=1}^{q_{1}}\;\sum_{\ell=1}^{q_{1}-1}\alpha(\ell)^{\frac{r-2}{r}}\\
&\leq k q_{1} \left[ \mathrm{Var}(\xi_i) + 2\bar{C} \right],
\end{split}
\end{equation}
by the condition \ref{a:3.3.2}.
Note that,
\begin{equation}
\begin{split}
\mathrm{Cov}\!\left(C_{n,i}^{(S)}, C_{n,j}^{(S)}\right) 
&= \mathrm{Cov}\!\left( \sum_{\ell = b_{i}}^{b_{i}+q_{1}-1}\xi_\ell, \; 
\sum_{p_{1} = b_{j}}^{b_{j}+q_{1}-1} \xi_{p_{1}} \right)\\
&= \sum_{\ell = b_{i}}^{b_{i}+q_{1}-1}\; \sum_{\rho = b_{j}}^{b_{j}+q_{1}-1} 
\mathrm{Cov}(\xi_\ell, \xi_{p_{1}}).\\
\end{split}
\end{equation}
Therefore, we have
\begin{equation}
\begin{split}
    A_{2} &\leq 2 \sum_{1 \leq i < j \leq k} 
\left| \mathrm{Cov}\!\left(C_{n,i'}^{(S)}, C_{n,j}^{(S)}\right) \right|\\
 &\leq 2 \sum_{1 \leq i < j \leq k} 
\sum_{\ell = b_{i}}^{b_{i}+q_{1}-1}\; \sum_{p_{1} = b_{j}}^{b_{j}+q_{1}-1} 
\left| \mathrm{Cov}(\xi_\ell, \xi_{p_{1}}) \right|\\
 &\leq 2 \sum_{i=1}^{n-p_{1}} \sum_{j=i+p_{1}}^{n}\left| \mathrm{Cov}(\xi_{i}, \xi_{j}) \right|\leq  \sum_{i=1}^{n-p_{1}} \sum_{j=i+p_{1}}^{n}\;\bar{C} \, \alpha(j-i)^{\frac{r-2}{r}}\\
&\leq  \sum_{i=1}^{n} \sum_{\ell=p_{1}}^{n} \;\bar{C} \, \alpha(\ell)^{\frac{r-2}{r}}\leq \bar{C} \,n \sum_{i=p_{1}}^{\infty}\alpha(\ell)^{\frac{r-2}{r}}
\end{split}
\end{equation}
Hence, it follows that
\begin{equation}
    n^{-1} \, \mathbb{E}\!\left[ \big(Z_n^{(S)}\big)^2 \right] 
\leq \frac{q_{1} k}{n} \Big( \mathrm{Var}(\xi_{i}) + \bar{C} \Big) 
+ \bar{C} \sum_{\ell = p_{1}}^{\infty} \alpha(\ell)^{\frac{r - 2}{r}}.
\end{equation}
Since $\frac{q_{1}}{p_{1}} \to 0, \quad p_{1} \to \infty, \quad \text{as } n \to \infty,$ then
we conclude that
$$
n^{-1} \, \mathbb{E}\!\left[ \big(Z_n^{(S)}\big)^2 \right] \to 0.
$$
Similarly, for the remainder block, we have 
\begin{equation}
\begin{split}
    \mathbb{E}\!\left[ \big(Z_n^{(R)}\big)^2 \right] 
&= \mathrm{Var}\!\left( \sum_{\ell = k(p_{1}+q_{1})+1}^{n} \xi_\ell \right)\\
&\leq \sum_{\ell = k(p_{1}+q_{1})+1}^{n} \mathrm{Var}(\xi_\ell)
+ 2 \sum_{\ell = k(p_{1}+q_{1})+1\leq i<j\leq n}\left| \mathrm{Cov}(\xi_{i},\xi_{j}) \right|\\
&\leq\big[n - k(p_{1}+q_{1})\big] \, \mathrm{Var}(\xi_{1}) 
+ 2 \sum_{i = k(p_{1}+q_{1})+1}^{n}\; \sum_{j = i+1}^{n} 
\left| \mathrm{Cov}(\xi_{i},\xi_{j}) \right|\\
&\leq\big[n - k(p_{1}+q_{1})\big] \, \mathrm{Var}(\xi_{1}) 
+ 2 \sum_{i = k(p_{1}+q_{1})+1}^{n}\; \sum_{\ell = 1}^{k(p_{1}+q_{1})-1} \Bar{C}\alpha(\ell)^{\frac{r-2}{r}}\\
&\leq\big[n - k(p_{1}+q_{1})\big] \, (\mathrm{Var}(\xi_{1})+2\Bar{C})
\end{split}
    \end{equation}
Then, it follows that
\begin{equation}
n^{-1} \mathbb{E}\big[ (Z_{n}^{(R)})^2 \big]
\leq \frac{(n-k(p_{1}+q_{1}))}{n}\big( \operatorname{Var}(\xi_{1}) + 2\Bar{C} \big) \to 0,
\end{equation}
 as $n \to \infty,$ provided that
$\frac{p_1}{n} \to 0, \quad \frac{q_1}{p_1} \to 0, \quad p_1 \to \infty,$ as $n \to \infty.$ Hence \eqref{eq:A1} holds.

For \eqref{eq:A2}, by Lemma 1.1 of \cite{Volkonskii}, we have
\begin{equation}
\Big| \mathbb{E}\big[ e^{it Z_n^{(L)}} \big] - \prod_{m_{1}=1}^k \mathbb{E}\big[e^{it C_{nm_{1}}^{(L)}}\big] \Big|
\leq 16k \alpha(q_1) \simeq 16 \alpha(q_1) \frac{n}{p_1} \to 0.
\end{equation}
as $n \to \infty.$
For \eqref{eq:A3}, using stationarity we obtain
\begin{equation}
\mathrm{Var}\big[n^{-1/2} Z_n^{(L)}\big]
= n^{-1} \sum_{m_{1}=1}^k \mathrm{Var}\big(C_{nm_{1}}^{(L)}\big)
+ \frac{2}{n} \sum_{1 \leq i < j \leq k} \mathrm{Cov}\big(C_{ni}^{(L)}, C_{nj}^{(L)}\big).
\end{equation}
Then,
\begin{equation}
    \begin{split}
       n^{-1} \sum_{m_{1}=1}^k \mathrm{Var}(C_{nm_{1}}^{(L)})
&= \frac{k}{n} \mathrm{Var}(C_{n1}^{(L)})
= \Big(\frac{k p_{1}}{n}\Big)\Big(\frac{1}{p_{1}}\Big) \mathrm{Var}\Big(\sum_{i=a_{m_{1}}}^{a_{m_{1}}+p_1-1} \xi_{i}\Big)\\ 
&=\Big(\frac{k p_{1}}{n}\Big)\Big(\frac{1}{p_{1}}\Big) \mathrm{Var}\Big(\sum_{i=a_{m_{1}}}^{a_{m_{1}}+p_{1}-1} \xi_{i}\Big)
\to \mathrm{Var}(\xi_{i}) = a^\top \Sigma a
    \end{split}
\end{equation}
Moreover, by Lemma 2.1 in (\cite{ROUSSAS1992262}, p.p-269),
\begin{equation}
    \begin{split}
        \left| \frac{2}{n} \sum_{1 \leq i < j \leq k} \;\mathrm{Cov}(C_{ni}^{(L)}, C_{nj}^{(L)}) \right|
&\leq \frac{2}{n} \sum_{1 \leq i < j \leq k}
\sum_{\ell=a_{i}}^{a_{i}+p_{1}+1}\;
\sum_{p_{1}=a_{j}}^{a_{j}+p_{1}-1}
\left|\mathrm{Cov}(\xi_{\ell}, \xi_{p_1})\right|\\
&\leq \frac{2}{n} \sum_{i=1}^{n-q_{1}}
\sum_{j=i+q_{1}}^{n} \Bar{C} \, \alpha(j-i)^{\frac{r-2}{r}}\\
&\leq \frac{2(n-q_{1})\Bar{C}}{n} \sum_{\ell=q_{1}}^{\infty} \alpha(\ell)^{\frac{r-2}{r}}\\
&\leq \Bar{C} \sum_{\ell=q_{1}}^{\infty}\alpha(\ell)^{\frac{r-2}{r}}\;\to 0,
    \end{split}
\end{equation}
by Assumption \ref{a:3.3.2} and the fact that $q_{1} \to \infty.$
From \eqref{eq:A4}, by applying Theorem 4.1 of \cite{Qi-ManShaoandHaoYu}, for any $r > 2$, and $\Tilde{\epsilon}>0,$ we obtain
\begin{equation}
    \begin{split}
        \mathbb{E}\Big[ (C_{n1}^{(L)})^2 \mathbf{1}\big(|C_{n1}^{(L)}| > \Tilde{\epsilon} \sqrt{n} \, \sqrt{a^T \Sigma a}\big) \Big]
\leq \Bar{C}\; n^{1-r/2} \mathbb{E}[|C_{n1}^{(L)}|^{r}]\\
\leq \Bar{C}\; n^{1-r/2} (p_{1})^{r/2} \mathbb{E}[|C_{n1}^{(L)}|^{r}]
    \end{split}
\end{equation}
Then by the Assumption \ref{a:3.3.2}-\ref{a:3.3.4}, $\mathbb{E}|C_{n1}^{(L)}|^{r} < \infty.$ Hence
\begin{equation}
n^{-1} \sum_{m_{1}=1}^k \mathbb{E}\Big[ (C_{nm_{1}}^{(L)})^2 \mathbf{1}\big(|C_{nm_{1}}^{(L)}| > \epsilon \sqrt{n}\big) \Big]
\leq \frac{k}{n}\Bar{C} n^{1-r/2} (p_{1})^{r/2}
= \Bar{C} \Big(\frac{p_{1}}{n}\Big)^{r/2} \to 0,
\end{equation}
as $n \to \infty$, this completes the proof.
\end{proof}
\subsection{B.6 Consistent Estimation of  \texorpdfstring{$\boldsymbol{\Omega}$}{}}\label{B.6}To prove Theorem~\ref{th:3.3.5}, we need Lemma~\ref{l:7.2.5} that establishes uniform consistency of each component of $\hat{\mathbb{J}}(\uptau)$. As a first step, we introduce the consistent estimator $\hat{\mathbb{J}}(\uptau)$, defined as
\begin{equation}\label{eq:7.73}
    \hat{\mathbb{J}}(\uptau) = \begin{pmatrix}
        \hat{\mathbb{J}}_{2,\boldsymbol{\beta}}(\uptau) & \hat{\mathbb{J}}_{2,\boldsymbol{\delta}}(\uptau) \\
        \boldsymbol{0}&\hat{\mathbb{J}}_{1,\boldsymbol{\delta}}
    \end{pmatrix},
\end{equation}
where each component is defined as follows: 
\begin{equation}\label{eq:7.74}
\begin{split}
\hat{\mathbb{J}}_{1,\boldsymbol{\delta}}&= \frac{1}{n}\sum_{i=1}^{n}\boldsymbol{z}_{i}\boldsymbol{z}^\top_{i},\\
\hat{\mathbb{J}}_{2,\boldsymbol{\beta}}(\uptau) &=\frac{1}{2h_{n}}\sum_{i=1}^{n}\mathds{1}\big(\boldsymbol{x}_{i}^\top\hat{\boldsymbol{\beta}}_{n}(\uptau)>0\big)\mathds{1}\big(|y_{i}-\boldsymbol{x}^\top_{i}\hat{\boldsymbol{\beta}}_{n}(\uptau)|\leq h_{n}\big)\hat{\boldsymbol{x}}_{i}\hat{\boldsymbol{x}}^\top_{i},\\
\hat{\mathbb{J}}_{2,\boldsymbol{\delta}}(\uptau) &=\frac{1}{2h_{n}}\sum_{i=1}^{n}\hat{\rho}_{1}\mathds{1}\big(\boldsymbol{x}_{i}^\top\hat{\boldsymbol{\beta}}_{n}(\uptau)>0\big)\mathds{1}\big(|y_{i}-\boldsymbol{x}^\top_{i}\hat{\boldsymbol{\beta}}_{n}(\uptau)|\leq h_{n}\big)\hat{\boldsymbol{x}}_{i}\boldsymbol{z}^\top_{i}.
    \end{split}
\end{equation}

 \begin{l1}\label{l:7.2.5} Under the same assumptions as in Theorem~\ref{th:3.3.5}, we have 
 \begin{enumerate}
    \item $\hat{\mathbb{J}}_{1,\boldsymbol{\delta}}\xrightarrow{p} \mathbb{J}_{1,\;\boldsymbol{\delta}}.$
    \item $\underset{\uptau \in \mathcal{T}}{\sup}\;\|\hat{\mathbb{J}}_{2,\boldsymbol{\beta}}(\uptau)-\mathbb{J}_{2,\boldsymbol{\beta}}(\uptau)\|=o_{p}(1).$
    \item $\underset{\uptau \in \mathcal{T}}{\sup}\;\|\hat{\mathbb{J}}_{2,\boldsymbol{\delta}}(\uptau)-\mathbb{J}_{2,\boldsymbol{\delta}}(\uptau)\|=o_{p}(1).$
\end{enumerate}
 \end{l1}
 \begin{proof}[\bf{Proof of the Lemma~\ref{l:7.2.5}}] It is easy to prove (1) by using the law of large numbers for an $\alpha$-mixing process. The proof of the uniform consistency of $\hat{\mathbb{J}}_{2,\boldsymbol{\beta}}(\uptau)$ and $\hat{\mathbb{J}}_{2,\boldsymbol{\delta}}(\uptau)$ is similar to the proof given in Appendix A.1.4 of \cite{Angrist} with some modifications. For any compact set $\mathscr{B}$ and positive
constant $\mathbb{H}>0$, the functional class
$$\{f(\boldsymbol{\beta},h)=\mathds{1}\big(\boldsymbol{x}^\top\boldsymbol{\beta}>0\big)\mathds{1}\big(|y-\boldsymbol{x}^\top\boldsymbol{\beta}|\leq h\big)\boldsymbol{x}\boldsymbol{x}^\top:\boldsymbol{\beta}\in \mathscr{B},h \in (0, \mathbb{H} ]\}$$
which is a Donsker class with a square-integrable envelope by Theorem 2.10.6 in \cite{vanderVaart1996}, because this is a product of the following two classes:
$\{\mathds{1}(|y_{i}-\boldsymbol{x}_{i}^\top\boldsymbol{\beta}|\leq h); \boldsymbol{\beta} \in \mathscr{B},h\in (0,\mathbb{H}]\}$, which is a VC class with constant envelope, and $\{\mathds{1}(\boldsymbol{x}^\top\boldsymbol{\beta} >0);\boldsymbol{\beta}\in \mathscr{B}\}$, which is also VC class with constant envelope, and consequently, the product of these two classes is also a VC class having a square-integrable matrix $\boldsymbol{x}_{i}\boldsymbol{x}^\top$ envelope by Assumption \ref{a:3.3.3}. Hence, 
\begin{equation}
\begin{split}
    (\boldsymbol{\beta}, h)\mapsto & \sqrt{n}\Big[\frac{1}{n}\sum_{i=1}^{n}\mathds{1}\big(\boldsymbol{x}_{i}^\top\boldsymbol{\beta}>0\big)\mathds{1}\big(|y_{i}-\boldsymbol{x}_{i}^\top\boldsymbol{\beta}|\leq h\big)\boldsymbol{x}_{i}\boldsymbol{x}_{i}^\top\\
&-\mathbb{E}\big[\mathds{1}\big(\boldsymbol{x}_{i}^\top\boldsymbol{\beta}>0\big)\mathds{1}\big(|y_{i}-\boldsymbol{x}_{i}^\top\boldsymbol{\beta}|\leq h\big)\boldsymbol{x}_{i}\boldsymbol{x}_{i}^\top\big]\Big]
    \end{split}
\end{equation}
converges weakly to a Gaussian process in $\ell^\infty\big(\mathscr{B}\times(0,\mathbb{H}]\big).$ Consequently, it follows that, 
\begin{equation}
\begin{split}
   &\displaystyle \underset{\boldsymbol{\beta}\in \mathscr{B},\; 0<h\leq \mathbb{H}}{\sup}\Big\|\frac{1}{n}\sum_{i=1}^{n}\mathds{1}\big(\boldsymbol{x}_{i}^\top\boldsymbol{\beta}>0\big)\mathds{1}\big(|y_{i}-\boldsymbol{x}_{i}^\top\boldsymbol{\beta}|\leq h\big)\boldsymbol{x}_{i}\boldsymbol{x}_{i}^\top\\
  & -\mathbb{E}\big[\mathds{1}\big(\boldsymbol{x}_{i}^\top\boldsymbol{\beta}>0\big)\mathds{1}\big(|y_{i}-\boldsymbol{x}_{i}^\top\boldsymbol{\beta}|\leq h\big)\boldsymbol{x}_{i}\boldsymbol{x}_{i}^\top\big]\Big\|= O_{p}(n^{-1/2}).
   \end{split}
\end{equation}
Let $\mathscr{B}$ be any compact set that covers $\underset{t \in \mathcal{T}}{\cup} \boldsymbol{\beta}_{0}(\uptau),$ and $\hat{\varepsilon}_{\uptau,i}=y_{i}-\hat{\boldsymbol{x}}_{i}^\top\hat{\boldsymbol{\beta}}_{n}(\uptau)$ then 
\begin{equation}
\begin{split}
   \displaystyle \underset{\uptau \in \mathcal{T}}{\sup}\;\Big\|\frac{1}{n}\sum_{i=1}^{n}\mathds{1}\big(\hat{\boldsymbol{x}}_{i}^\top\hat{\boldsymbol{\beta}}_{n}(\uptau)>0\big)\mathds{1}\big(|\hat{\varepsilon}_{\uptau,i}|\leq h_{n}\big)\hat{\boldsymbol{x}}_{i}\hat{\boldsymbol{x}}_{i}^\top
   &-\mathbb{E}\big[\mathds{1}\big(\hat{\boldsymbol{x}}_{i}^\top\hat{\boldsymbol{\beta}}_{n}(\uptau)>0\big)\mathds{1}\big(|\hat{\varepsilon}_{\uptau,i}|\leq h_{n}\big)\hat{\boldsymbol{x}}_{i}\hat{\boldsymbol{x}}_{i}^\top\big]\Big\|\\
   &= O_{p}(n^{-1/2}).
   \end{split}
\end{equation}
Now, observe that, 
\begin{align*}
\hat{\mathbb{J}}_{2,\;\hat{\boldsymbol{\beta}}_{n}(\uptau)}(\uptau)&=\frac{1}{2.h_{n}}\mathbb{E}\big[\mathds{1}\big(\hat{\boldsymbol{x}}_{i}^\top\hat{\boldsymbol{\beta}}_{n}(\uptau)>0\big)\mathds{1}\big(|\hat{\varepsilon}_{\uptau,i}|\leq h_{n}\big)\hat{\boldsymbol{x}}_{i}\hat{\boldsymbol{x}}_{i}^\top\big]+O_{p}(n^{-1/2})\\
&=\mathbb{E}\big[\mathds{1}\big(\boldsymbol{x}_{i}^\top\boldsymbol{\beta}_{0}(\uptau)>0\big)f_{\varepsilon_{\uptau}}(0|\boldsymbol{x})\boldsymbol{x}_{i}\boldsymbol{x}_{i}^\top\big]+o_{p}(1)\\
&=\mathbb{J}_{2,\;\boldsymbol{\beta}(\uptau)}+o_{p}(1)
\end{align*} 
uniformly in $\uptau \in \mathcal{T}$, and using the fact form Assumption~\ref{a:3.3.9} that bandwidth $h_{n}$ satisfies $h_{n} \to 0 $ and $nh^{2}_{n}\to \infty$ as $n \to \infty.$ Hence, 
$\hat{\mathbb{J}}_{2,\boldsymbol{\beta}}(\uptau)\xrightarrow{p}\mathbb{J}_{2,\boldsymbol{\beta}}(\uptau)
$ uniformly in $\uptau\in \mathcal{T}.$

Similarly, the functional class
$$\{g(\boldsymbol{\beta},h)=\rho_{1}\mathds{1}\big(\boldsymbol{x}^\top\boldsymbol{\beta}>0\big)\mathds{1}\big(|y-\boldsymbol{x}^\top\boldsymbol{\beta}|\leq h\big)\boldsymbol{x}\boldsymbol{z}^\top:\boldsymbol{\beta}\in \mathscr{B},h \in (0, \mathbb{H} ]\}$$
is Donsker class then using the same argument we can show that $\hat{\mathbb{J}}_{2,\boldsymbol{\delta}}(\uptau)\xrightarrow{p}\mathbb{J}_{2,\boldsymbol{\delta}}(\uptau),$ uniformly in $\uptau \in \mathcal{T}$
\end{proof}
\begin{proof}[{\bf{Proof of the Theorem \ref{th:3.3.5}:}}] It follows directly from Lemma~\ref{l:7.2.5} that $\hat{\mathbb{J}}(\uptau)$ is uniformly consistent estimator of $\mathbb{J}(\uptau)$. 

To estimate the long-run covariance matrix $\mathbb{V}(\uptau,\uptau'),$ we adopt the kernel HAC estimator suggested by \cite{Galvao2024}. Particularly, we define
\begin{equation}\label{eq:7.78}
    \hat{\mathbb{V}}(\uptau, \uptau') = \begin{pmatrix}
        \hat{\mathbb{V}}_{\boldsymbol{\beta},\boldsymbol{\beta}}(\uptau, \uptau') & \hat{\mathbb{V}}_{\boldsymbol{\beta},\boldsymbol{\delta}}(\uptau) \\
        \hat{\mathbb{V}}_{\boldsymbol{\delta},\boldsymbol{\beta}}(\uptau) &\hat{\mathbb{V}}_{\boldsymbol{\delta},\boldsymbol{\delta}}
    \end{pmatrix},
\end{equation}
where the block components are given by
\begin{equation}\label{eq:7.79}
    \begin{split}
       \hat{\mathbb{V}}_{\boldsymbol{\delta},\boldsymbol{\delta}}&=\frac{1}{n}\sum_{j=-n+1}^{n-1}\sum_{i=1}^{n-j}\mathbb{K}\Big(\frac{j}{b_{n}}\Big)\Psi_{1}\big(\vartheta_{i},\boldsymbol{z}_{i},\hat{\boldsymbol{\delta}}_{n}\big)\Psi_{1}\big(\vartheta_{i+j},\boldsymbol{z}_{i+j},\hat{\boldsymbol{\delta}}_{n}\big)^\top\\
       \hat{\mathbb{V}}_{\boldsymbol{\beta},\boldsymbol{\beta}}(\uptau, \uptau')&=\frac{1}{n}\sum_{j=-n+1}^{n-1}\sum_{i=1}^{n-j}\mathbb{K}\Big(\frac{j}{b_{n}}\Big)\Psi_{2}\big(s_{i},\hat{\boldsymbol{\beta}}_{n}(\uptau),\hat{\boldsymbol{\delta}}_{n},\uptau\big)\Psi_{2}\big(s_{i+j},\hat{\boldsymbol{\beta}}_{n}(\uptau),\hat{\boldsymbol{\delta}}_{n},\uptau'\big)^\top\\
       \hat{\mathbb{V}}_{\boldsymbol{\beta},\boldsymbol{\delta}}(\uptau)&=\frac{1}{n}\sum_{j=-n+1}^{n-1}\sum_{i=1}^{n-j}\mathbb{K}\Big(\frac{j}{b_{n}}\Big)\Psi_{2}\big(s_{i},\hat{\boldsymbol{\beta}}_{n}(\uptau),\hat{\boldsymbol{\delta}}_{n},\uptau\big)\Psi_{1}\big(\vartheta_{i+j},\boldsymbol{z}_{i+j},\hat{\boldsymbol{\delta}}_{n}\big)^\top\\
       \hat{\mathbb{V}}_{\boldsymbol{\delta},\boldsymbol{\beta}}(\uptau)&= \hat{\mathbb{V}}_{\boldsymbol{\beta},\boldsymbol{\delta}}(\uptau)^\top.
    \end{split}
\end{equation}
Let $\mathbb{K}(t) = (1 - |t|)\mathds{1}(|t|\leq 1)$ denote the Bartlett kernel, and consider the lag truncation parameter $b_{n} = \bigl[(4n/100)^{1/3}\bigr] > 0$, where $[,\cdot,]$ denotes the integer part. This choice of $b_{n}$ satisfies Assumption~\ref{a:3.3.9}, ensuring that the truncation error vanishes asymptotically. Moreover, since $\hat{\boldsymbol{\beta}}_{n}(\uptau)$ and $\hat{\boldsymbol{\delta}}_{n}$ are uniformly consistent over $\uptau \in \mathcal{T}$, the consistency of each component in \eqref{eq:7.79} follows from arguments analogous to those used in the analysis of HAC estimators for quantile regression in \cite{Galvao2024} and least squares in \cite{Andrews1991}, under Assumptions~\ref{a:3.3.1}–\ref{a:3.3.9}. Although \cite{Galvao2024} develops the approach for quantile regression, the extension to censored quantile regression is straightforward under the stated regularity conditions. Hence, we obtain that $\hat{\mathbb{V}}(\uptau, \uptau')\xrightarrow{p}\mathbb{V}(\uptau, \uptau')$ uniformly over $\uptau \in \mathcal{T}.$ 
\end{proof}
\begin{proof}[{\bf{Proof of the Theorem \ref{th:3.3.6}:}}] The proof of consistency of the HAC estimators of $\hat{\boldsymbol{\Omega}}$ to $\boldsymbol{\Omega}$ follows directly from \cite{Galvao2024} and \cite{Andrews1991} under Assumptions \ref{a:3.3.1}-\ref{a:3.3.9}.
\end{proof}
\begin{proof}[{\bf{Proof of Corollary \ref{cor1}:}}] \label{pcor1} The assertion in Theorem \ref{th:3.3.6} and, in view of the assertion in Theorem \ref{th:3.3.4}, the result follows.   
\end{proof}
\section{Simulation Results in Tabular Form}\label{SST1}
This section provides the simulation results of Section \ref{FSSS} in tabular form. To be specific, each table reports the EMSE, Ebias, and C.P. corresponding to the estimators for different sample sizes.
\begin{table}[H]
    \small
    \centering
    \caption{Ebias and  of estimates of $\mathscr{L}_{n}$ in Example~\ref{e3} with C.P.}
    \resizebox{\textwidth}{!}{%
    \begin{tabular}{|c |c|c|c |c|c|c |c|c|c |c|c|c|}
        \toprule
          & \multicolumn{3}{|c|}{$\boldsymbol{n = 50}$} 
          & \multicolumn{3}{|c|}{$\boldsymbol{n = 100}$} 
          & \multicolumn{3}{|c|}{$\boldsymbol{n = 500}$} 
          & \multicolumn{3}{|c|}{$\boldsymbol{n = 1000}$} \\
       \cline{2-13}
        \textbf{Variables} & \textbf{Ebias} & \textbf{EMSE} & \textbf{C.P.} 
                           & \textbf{Ebias} & \textbf{EMSE} & \textbf{C.P.} 
                           & \textbf{Ebias} & \textbf{EMSE} & \textbf{C.P.} 
                           & \textbf{Ebias} & \textbf{EMSE} & \textbf{C.P.} \\
        \midrule
        \multicolumn{12}{|c|}{$\boldsymbol{\varepsilon_{i}=\rho^{*}\varepsilon_{i-1}+\eta_{1,i}}, \;\rho^{*}=0.5,\; \eta_{1,i}\sim N(0,1)$} \\
        \midrule
        $\mathcal{L}_{n0}$ & 0.1118 &0.1131  &  & 0.0866 & 0.877 &  & 0.0227 &0.0256  &  &-0.0026  & 0.0121 &  \\
        $\mathcal{L}_{n1}$ &-0.0236  & 0.0267 &28$\%$  & -0.0148 & 0.0197 & 30$\%$ &0.0062  & 0.0146 & 29$\%$ & 0.0024 & 0.0130 &28$\%$  \\
        $\mathcal{L}_{n2}$ &-0.2651  &0.2663  &  &-0.1923  &0.1935  &  &-0.0366  &0.0414  &  &-0.0061  &0.0206  &  \\
        $\mathcal{L}_{n3}$ & 0.2970 &0.2977  &  & 0.2055 &0.2059  &  &0.0468  & 0.0472 &  & 0.0128 &  0.0137&  \\
        \bottomrule
    \end{tabular}}
    \label{table:8c}
\end{table}

\begin{table}[H]
    \small
    \centering
    \caption{Ebias and  of estimates of $\mathscr{L}_{n}$ in Example~\ref{e1} 
             for $\Tilde{\alpha} = 0.01, 0.02, 0.20$ with C.P.}
    \resizebox{\textwidth}{!}{%
    \begin{tabular}{|c| c|c|c| c|c|c| c|c|c| c|c|c|}
        \toprule
          & \multicolumn{3}{|c|}{$\boldsymbol{n = 50}$} 
          & \multicolumn{3}{|c|}{$\boldsymbol{n = 100}$} 
          & \multicolumn{3}{|c|}{$\boldsymbol{n = 500}$} 
          & \multicolumn{3}{|c|}{$\boldsymbol{n = 1000}$} \\
        \cline{2-13}
        \textbf{Variables} & \textbf{Ebias} & \textbf{EMSE} & \textbf{C.P.} 
                           & \textbf{Ebias} & \textbf{EMSE} & \textbf{C.P.} 
                           & \textbf{Ebias} & \textbf{EMSE} & \textbf{C.P.} 
                           & \textbf{Ebias} & \textbf{EMSE} & \textbf{C.P.} \\
        \midrule
        \multicolumn{12}{|c|}{$\boldsymbol{\varepsilon_{i}=\rho^{*}\varepsilon_{i-1}+\eta_{1,i}}, \;\rho^{*}=0.5,\; \eta_{1,i}\sim N(0,1)$} \\
        \midrule
        \multicolumn{12}{|c|}{$\boldsymbol{\Tilde{\alpha} = 0.01}$} \\
        \midrule
        $\mathcal{L}_{n0}$ &0.0340  &0.0914  && 0.0833 & 0.0850 && 0.0244 &0.0289 &&-0.0071  & 0.0178 &  \\
        $\mathcal{L}_{n1}$ &-0.2669  & 0.2668 & 30$\%$  &-0.1634  & 0.1637 &38$\%$  &-0.0293 &0.0297  & 31$\%$ & -0.0109 & 0.0115 &35$\%$  \\
        $\mathcal{L}_{n2}$ &0.0048  &0.0156  && 0.0059 &0.0073  && 0.0024 &0.0030  && 0.0014 & 0.0019 &  \\
        $\mathcal{L}_{n3}$ &0.0381 & 0.0386 &  & 0.0057 & 0.0174 &  &0.0018  & 0.0027 &  & 0.0007 & 0.0017 &  \\
        \midrule
        \multicolumn{12}{|c|}{$\boldsymbol{\Tilde{\alpha} = 0.02}$} \\
        \midrule
        $\mathcal{L}_{n0}$ & 0.1046 &0.1061  &  & 0.0673 &0.0693  &  &0.0185  & 0.0241 &  & -0.0008  &  0.0157 &  \\
        $\mathcal{L}_{n1}$ & -0.2635 & 0.2639 &36$\%$ &-0.1594  & 0.1597 & 32$\%$ & -0.0257 &0.0262  &30$\% $& -0.0085  & 0.0092  &35$\%$  \\
        $\mathcal{L}_{n2}$ &0.0060  & 0.0073 &  &0.0035  &0.0071  &  & 0.0024 & 0.0030 &  &0.0014  & 0.0019 &  \\
        $\mathcal{L}_{n3}$ &0.0372  & 0.0379 &  & 0.0142 & 0.0150 &  & 0.0014 & 0.0026 &  & 0.0007 &0.0017  &  \\
        \midrule
        \multicolumn{12}{|c|}{$\boldsymbol{\Tilde{\alpha} = 0.20}$} \\
        \midrule
        $\mathcal{L}_{n0}$ & 0.1272 & 0.1288 &  & 0.0575 & 0.0598 &  & 0.0137 & 0.0188 &  &  8.2$e^{-05}$& 0.0130 &  \\
        $\mathcal{L}_{n1}$ &-0.3433  &0.3439  &38$\%$  &-0.1941 &0.1947  & 33$\%$ & -0.0214 &0.0225 & 29$\%$&-0.0056&0.0074  &34$\%$ \\
        $\mathcal{L}_{n2}$ &0.0060&0.0085  &  & 0.0021 & 0.0083 &  &0.0026  &0.0036&& 0.0014&0.0022& \\
        $\mathcal{L}_{n3}$ & 0.0384 &0.0397 & &0.0117  & 0.0136 &  & 0.0011 & 0.0031 & & 0.0009&0.0022& \\
        \bottomrule
    \end{tabular}}
    \label{table:6c}
\end{table}
\begin{table}[H]
    \small
    \centering
    \caption{Ebias and  of estimates of $\mathscr{L}_{n}$ in Example~\ref{e2} 
             for $\Tilde{\alpha} = 0.01, 0.02, 0.20$ with C.P.}
    \resizebox{\textwidth}{!}{%
    \begin{tabular}{|c| c|c|c| c|c|c| c|c|c| c|c|c|}
        \toprule
          & \multicolumn{3}{c}{$\boldsymbol{n = 50}$} 
          & \multicolumn{3}{c}{$\boldsymbol{n = 100}$} 
          & \multicolumn{3}{c}{$\boldsymbol{n = 500}$} 
          & \multicolumn{3}{c}{$\boldsymbol{n = 1000}$} \\
        \cmidrule(lr){2-4}\cmidrule(lr){5-7}\cmidrule(lr){8-10}\cmidrule(lr){11-13}
        \textbf{Variables} & \textbf{Ebias} & \textbf{EMSE} & \textbf{C.P.} 
                           & \textbf{Ebias} & \textbf{EMSE} & \textbf{C.P.} 
                           & \textbf{Ebias} & \textbf{EMSE} & \textbf{C.P.} 
                           & \textbf{Ebias} & \textbf{EMSE} & \textbf{C.P.} \\
        \midrule
        \multicolumn{12}{|c|}{$\boldsymbol{\varepsilon_{i}=\rho^{*}\varepsilon_{i-1}+\eta_{1,i}}, \;\rho^{*}=0.5,\; \eta_{1,i}\sim N(0,1)$} \\
        \midrule
        \multicolumn{12}{|c|}{$\boldsymbol{\Tilde{\alpha} = 0.01}$} \\
        \midrule
        $\mathcal{L}_{n0}$ & 0.1225 & 0.1244 &  & 0.0917 & 0.0943 &  &0.0163  & 0.0232 &  &0.0016  & 0.0179 &  \\
        $\mathcal{L}_{n1}$ & -0.2690 & 0.2698 & 38$\%$&-0.2682  & 0.2691 & 38$\%$& -0.0339 & 0.0359 &35$\%$  &-0.0153  &0.0179  & 34$\%$ \\
        $\mathcal{L}_{n2}$ &0.0027  &0.0105  &  &0.0024  &0.0101  &  &0.0026  & 0.0054 &  &0.0015  &0.0038  &  \\
        $\mathcal{L}_{n3}$ &0.0395  & 0.0409 &  &0.0389  &0.0408  &  &0.0020  &0.0057  &  &0.0012  & 0.0041 &  \\
        \midrule
        \multicolumn{12}{|c|}{$\boldsymbol{\Tilde{\alpha} = 0.02}$} \\
        \midrule
        $\mathcal{L}_{n0}$ &0.0837  & 0.0891 &  & 0.0794 & 0.0839 &  & 0.0129 & 0.0224 &  & 0.0022 &0.0176  &  \\
        $\mathcal{L}_{n1}$ & -0.2659 & 0.2683 & 46$\%$ & -0.1646 & 0.1674 &35$\%$ & -0.0417 & 0.0465 &40$\%$  &-0.0133  & 0.0203 & 32 $\%$ \\
        $\mathcal{L}_{n2}$ & 0.0034 &0.0180  &  &0.0068  &0.0152  &  &0.0032  & 0.0087 &  &0.0019  & 0.0058 &  \\
        $\mathcal{L}_{n3}$ &0.0380  & 0.0430 &  & 0.0152 &0.0222  &  &0.0026  &0.0092  &  & 0.0007 &  0.0061&  \\
        \midrule
        \multicolumn{12}{|c|}{$\boldsymbol{\Tilde{\alpha} = 0.20}$} \\
        \midrule
        $\mathcal{L}_{n0}$ & 0.0877 & 0.2134 &  &0.0514  &0.1528  &  &0.0039  &0.0657  &  &-0.0003  &0.0448  &  \\
        $\mathcal{L}_{n1}$ & -0.2196 & 0.3538 & 34$\%$ & -0.1283 & 0.2416  &33$\%$ & -0.0179 &0.0945  & 34$\%$ &0.0039  &0.0633 &34$\%$ \\
        $\mathcal{L}_{n2}$ &-0.0053  &0.1098  &  &0.0039  & 0.0766 &  &0.0037  &0.0333  &  &-0.0006  &0.0222  &  \\
        $\mathcal{L}_{n3}$ &0.0370  & 0.1358 &  &0.0098  & 0.0847 &  & -0.0012 & 0.0371 &  &0.0020  & 0.0264 &  \\
        \bottomrule
    \end{tabular}}
    \label{table:7c}
\end{table}
\bibliographystyle{apalike}
\bibliography{Reference}
\end{document}